\documentclass[11pt,a4paper]{article}
\usepackage{fullpage}
\usepackage[sort]{cite}

\usepackage{amssymb,amsmath,amsthm}
\usepackage{amsfonts}
\usepackage{mathtools}
\usepackage{graphicx}
\usepackage{enumerate}
\usepackage{tikz}
\usetikzlibrary{shapes}
\usetikzlibrary{calc}
\usetikzlibrary{patterns.meta}

\usepackage[T1]{fontenc}

\usepackage{enumitem}
\usepackage{todonotes}
\usepackage{amsmath}

\definecolor{blue}{rgb}{0.1,0.2,0.5}
\definecolor{brown}{rgb}{0.6,0.6,0.2}
\usepackage[ocgcolorlinks, linkcolor={blue}, citecolor={brown}]{hyperref}
\usepackage{cleveref}

\usepackage{enumerate}
\usepackage{latexsym}

\usepackage{todonotes}

\newtheorem{lemma}{Lemma}[section]
\newtheorem{theorem}[lemma]{Theorem}
\newtheorem{corollary}[lemma]{Corollary}
\newtheorem{claim}[lemma]{Claim}
\theoremstyle{definition}
\newtheorem{definition}[lemma]{Definition}
\crefname{section}{Section}{Sections}
\crefname{corollary}{Corollary}{Corollaries}
\crefname{theorem}{Theorem}{Theorems}
\crefname{lemma}{Lemma}{Lemmas}
\crefname{claim}{Claim}{Claims}
\crefname{figure}{Figure}{Figures}

\renewcommand{\leq}{\leqslant}
\renewcommand{\geq}{\geqslant}
\renewcommand{\le}{\leqslant}
\renewcommand{\ge}{\geqslant}
\renewcommand{\setminus}{-}

\newcommand{\Oh}{\mathcal{O}}
\newcommand{\Otilde}{\widetilde{\mathcal{O}}}

\newcommand{\dist}{\textup{dist}}
\newcommand{\Pub}{\textup{Pub}}
\newcommand{\Priv}{\textup{Priv}}
\newcommand{\floor}[1]{\left\lfloor #1 \right\rfloor}
\newcommand{\ceil}[1]{\left\lceil #1 \right\rceil}

\newcommand{\zZ}{\mathcal{Z}}

\newcommand{\wZ}{\widehat{\zZ}}

\newenvironment{claimproof}[1][Proof of the claim.]{%
	\begin{proof}[#1]%
	}{%
	\end{proof}%
}

\newcommand{\executeiffilenewer}[3]{%
\ifnum\pdfstrcmp{\pdffilemoddate{#1}}%
{\pdffilemoddate{#2}}>0%
{\immediate\write18{#3}}\fi%
}

\newcommand{\svg}[2]{\def\svgwidth{#1}%
\executeiffilenewer{#2.svg}{#2.pdf}%
{inkscape -D #2.svg %
--export-filename=#2.pdf --export-latex}%
\everymath{\color{black}}%
\textcolor{black}{\input{#2.pdf_tex}}\everymath{\color{darkred}}%
}

\newcommand{\svgc}[2]{\begin{center}\svg{#1}{#2}\end{center}}

\begin{document}
\begin{titlepage}
\iffalse
\author{Anonymous authors}
\title{Pattern-Sparse Tree Decompositions in $H$-Minor-Free Graphs}
\fi

\iftrue
\author{
  D\'{a}niel Marx%
  \thanks{CISPA Helmholtz Center for Information Security, Germany. \texttt{marx@cispa.de}}
  \and
  Marcin Pilipczuk%
  \thanks{Institute of Informatics, University of Warsaw, Poland. \texttt{m.pilipczuk@uw.edu.pl}}
  \and
  Micha\l{} Pilipczuk%
  \thanks{Institute of Informatics, University of Warsaw, Poland. \texttt{michal.pilipczuk@mimuw.edu.pl}}}
\title{Pattern-Sparse Tree Decompositions in $H$-Minor-Free Graphs%
\thanks{
This work is 
  a part of project BUKA (Ma. P.) and BOBR (Mi. P.)
  that have received funding from the European Research Council (ERC) 
under the European Union's Horizon 2020 research and innovation programme (grant agreement No.~948057, and~101126229, respectively).}}
\fi

\date{}

\maketitle

\begin{abstract}
  Given an $H$-minor-free graph $G$ and an integer $k$, our main technical contribution is sampling 
  in randomized polynomial time an induced subgraph $G'$ of $G$ and a tree decomposition of $G'$ of width $\Otilde(k)$ such that for every $Z\subseteq V(G)$ of size $k$, with probability at least $\left(2^{\Otilde(\sqrt{k})}|V(G)|^{\Oh(1)}\right)^{-1}$, we have $Z \subseteq V(G')$ and every bag of the tree decomposition contains at most $\Otilde(\sqrt{k})$ vertices of $Z$. Having such a tree decomposition allows us to solve a wide range of problems in (randomized) time $2^{\Otilde(\sqrt{k})}n^{\Oh(1)}$ where the solution is a pattern $Z$ of size $k$, e.g., \textsc{Directed $k$-Path}, \textsc{$H$-Packing}, etc. In particular, our result recovers all the algorithmic applications of the pattern-covering result of   
    Fomin et al.~[\emph{SIAM J.~Computing} 2022] (which requires the pattern to be connected)
    and the planar subgraph-finding algorithms of Nederlof~[STOC~2020].

    Furthermore, for $K_{h,3}$-free graphs (which include bounded-genus graphs) and for a fixed constant $d$, we signficantly strengthen the result by ensuring that not only $Z$ has intersection $\Otilde(\sqrt{k})$ with each bag, but even the distance-$d$ neighborhood $N^d_{G}[Z]$ as well. This extension makes it possible to handle a wider range of problems where the neighborhood of the pattern also plays a role in the solution, such as partial domination problems and problems involving distance constraints.
    \end{abstract}
\def\thepage{}
\thispagestyle{empty}
\end{titlepage}

\thispagestyle{empty}
\setcounter{page}{0}
\tableofcontents
\clearpage

\section{Introduction}

Significant research efforts have been dedicated to designing parameterized algorithms for planar graphs, aiming for running times that depend subexponentially on the size of the solution or on some other parameter of the input instance \cite{DBLP:conf/focs/MarxPP18,DBLP:journals/siamcomp/ChitnisFHM20,DBLP:conf/focs/FominLMPPS16,DBLP:journals/talg/MarxP22,DBLP:conf/soda/KleinM14,DBLP:conf/icalp/KleinM12,DBLP:conf/icalp/Marx12,DBLP:conf/fsttcs/LokshtanovSW12,DBLP:journals/algorithmica/Verdiere17,DBLP:conf/stoc/Nederlof20a,platypus,DBLP:journals/siamdm/DemaineFHT04,DBLP:conf/soda/DemaineH04,DBLP:journals/jacm/DemaineFHT05,DBLP:journals/cj/DemaineH08,DBLP:journals/ipl/FominLRS11,DBLP:conf/soda/MarxMNT22}. While many of these algorithms employ problem-specific arguments, several useful general techniques and design patterns have also been discovered. In some cases, these techniques generalize also to bounded-genus or $H$-minor-free graphs.

The theory of bidimensionality gives a very clean win/win approach for certain problems, where the answer for an instance with parameter $k$ is trivial if treewidth is $\Omega(\sqrt{k})$. However, there is only a limited set of such problems and often natural generalizations of the problem definition ruins this property; in particular, this is often the case for problems involving terminals and directed graphs. Marx et al.~\cite{DBLP:conf/soda/MarxMNT22} and Bandyapadhyay et al.~\cite{DBLP:conf/soda/BandyapadhyayLL22} presented general techniques for handling cycle hitting problems, where the task is to delete at most $k$ edges/vertices such that some property related to components or 2-connected components of the resulting graph is satisfied (e.g., there is no cycle going through a terminal). However, this framework is inherently about deletion-type problems, and can handle a problem only if the sought-after property of the remaining graph can be described using constraints in a certain way.

Fomin et al.~\cite{DBLP:journals/siamcomp/FominLMPPS22} presented a powerful technique that can be applied to many problems where the solution is a connected pattern of size $k$. For such problems, using the following result as a black box and an algorithm for solving bounded-treewidth instances delivers a subexponential parameterized algorithm. An {\em apex graph} is a graph that can be made planar by deleting a vertex. Note that $K_{3,h}$ is an apex graph and for every $g\ge 0$ there is an $h$ such that graphs of Euler genus at most $g$ are $K_{3,h}$-minor-free; thus the statement applies in particular to bounded-genus~graphs.

\begin{theorem}[Fomin et al.~\cite{DBLP:journals/siamcomp/FominLMPPS22}]\label{thm:patterncoverold}
Let $\mathcal{C}$ be a class of graphs that exclude a fixed apex graph as a minor. Then there exists
a randomized polynomial-time algorithm that, given an $n$-vertex graph $G$ from $\mathcal{C}$ and an integer $k$,
samples a vertex subset $A \subseteq V(G)$ with the following properties:
\begin{itemize}
    \item[(P1)] The induced subgraph $G[A]$ has treewidth $\Oh(\sqrt{k} \log k)$.
    \item[(P2)] For every vertex subset $Z \subseteq V(G)$ with $|Z| \leq k$ that induces a connected subgraph of $G$,
    the probability that $Z$ is covered by $A$, that is, $Z \subseteq A$, is at least 
    \[
    \left(2^{\Oh(\sqrt{k} \log^2 k)} \cdot n^{\Oh(1)}\right)^{-1}.
    \]
\end{itemize}
\end{theorem}
Suppose we investigate a problem for which a solution is a connected set $Z$ of size $k$, that is, $G[Z]$ is connected. As a running example, let us consider the problem of finding a directed cycle of length exactly $k$ in a directed graph. Let us invoke the algorithm of Theorem~\ref{thm:patterncoverold} (on the underlying undirected graph) and let $A$ be the set returned by it.
As $G[A]$ has treewidth $\Oh(\sqrt{k}\log k)$, it is possible to test the existence of a directed $k$-cycle in $G[A]$ in time $2^{\Otilde(\sqrt{k})}n^{\Oh(1)}$ using standard dynamic programming techniques. If we find a $k$-cycle, then we can stop, as this is a $k$-cycle in $G$ as well. Otherwise, we invoke again Theorem~\ref{thm:patterncoverold} to sample a new set $A$, and repeat the process. If there is a directed $k$-cycle $Z$, then (P2) ensures that after $\Oh(  2^{\Oh(\sqrt{k} \log^2 k)} \cdot n^{\Oh(1)})$ repetitions, we encounter a set $A$ with $Z\subseteq A$ with high probability, and hence find a directed $k$-cycle in $G$.

\clearpage
In order to be able to use this approach for other algorithmic problems, the problem needs to satisfy three main properties:
\begin{enumerate}[noitemsep]
\item The problem can be solved in time $2^{\Otilde(t)}n^{\Oh(1)}$ on graphs of treewidth $t$.
\item The solution $Z$ needs to be connected.
\item Validity of the solution $Z$ depends only on $G[Z]$ and is not affected by vertices outside $Z$.
\end{enumerate}
The first point is not very restrictive: designing algorithms on bounded-treewidth graphs is quite well understood, and the standard technique of dynamic programming on a tree decomposition delivers such alogrithms. It is worth pointing out that these dynamic programming techniques are very robust in the sense that different problem variations and extensions (colored graphs, directed graphs, weights, etc.) can be usually achieved with little additional effort. On the other hand, the second and the third points severely limits the range of problems that can be handled by this approach. Requiring connectivity excludes problems such as finding many disjoint small patterns, and defining the problem in terms of some property of $G[Z]$ excludes in particular domination-type problems and problems involving distances.

Another technique for subexponential algorithms for detecting subgraphs in planar graphs was presented by Nederlof~\cite{DBLP:conf/stoc/Nederlof20a}. This work builds upon a key combinatorial result of Fomin et al.~\cite{DBLP:journals/siamcomp/FominLMPPS22}, a certain type of duality between a chain of disjoint separators and almost-disjoint paths (see Theorem~\ref{thm:dual1} below). A crucial new insight of Nederlof~\cite{DBLP:conf/stoc/Nederlof20a} (see~\cite[Lemma~3.7]{DBLP:conf/stoc/Nederlof20a}) is that instead of focusing on finding separators of size $\Otilde(\sqrt{k})$ (which is needed to construct a tree decomposition of width $\Otilde(\sqrt{k})$), it is sufficient for our purposes to find separators of size polynomial in $k$ that intersect the solution in $\Otilde(\sqrt{k})$ vertices. 
(This is what we henceforth call ``pattern-sparse'' in this work; similar ideas appeared e.g. in~\cite{DBLP:journals/siamcomp/BergBKMZ20}.)
Indeed, if a dynamic programming algorithm considers how the solution behaves with respect to such a separator, then considering all $k^{\Otilde(\sqrt{k})}=2^{\Otilde(\sqrt{k})}$ possible intersections of the solution with the separator and all $2^{\Otilde(\sqrt{k})}$ possible states of this intersection is within the allowed running time. The algorithm of Nederlof \cite{DBLP:conf/stoc/Nederlof20a} can find disconnected patterns, it is deterministic, and considerable technical efforts have been made to make it work also for the counting versions of the problems within the same running time. On the negative side, the algorithm is not presented in a way that allows easy reuse for other type of problems. Moreover, the algorithm works only on planar graphs: cycle separators and arguments about the inside and outside of the cycle are inherent in the algorithm. This makes potential generalizations to bounded-genus or $H$-minor-free graphs difficult.

\paragraph{Our results.} Building upon the work of Nederlof~\cite{DBLP:conf/stoc/Nederlof20a}, we present a new general technique that is closer in spirit to Theorem~\ref{thm:patterncoverold}, but has a significantly greater level of applicability. In particular, our result has the following advantages:
\begin{itemize}[noitemsep]
\item It works for $H$-minor-free graphs.
\item It can handle disconnected solutions.
\item It returns a tree decomposition of width polynomial in $k$, and (with some probability) the solution intersects each bag in $\Otilde(\sqrt{k})$ vertices.
  \item For $K_{3,h}$-minor free graphs and a fixed $d$, we can even lower bound the probability that the \emph{distance-$d$ neighborhood} of the solution has $\Otilde(\sqrt{k})$ intersection with each bag. 
\end{itemize}
Therefore, this technique can be used, for example, to find (for some well-structured classes of patterns) a pattern graph $P$ as a subgraph in a given $H$-minor-free graph $G$ in time $2^{\Otilde(\sqrt{|V(P)|})}|V(G)|^{\Oh(1)}$, even if $P$ is disconnected. For $K_{3,h}$-minor-free graphs, our extension to distance-$d$ neighborhoods opens up the possibility to solve a wide range of new problems. For example, we can solve problems involving partial domination, where we have to find a vertex set/pattern maximizing the number of vertices in its neighborhood. Moreover, we can handle problems involving distance constraints, such as, say, finding $k$ triangles that are at distance at least $d$ from each other. In Section~\ref{sec:app} below, we elaborate on some of the potential problems we can solve. 

Formally, our main result for $H$-minor-free graphs is the following.

\begin{theorem}\label{thm:main}
  For every graph $H$, there exists a randomized polynomial-time algorithm that, 
  given an $H$-minor-free graph $G$ and an integer $k \geq 2$
  outputs an induced subgraph $G'$ of $G$ and 
  a tree decomposition $(T,\beta)$ of $G'$ of width $\Oh_H(k\log k)$
  such that for every $Z \subseteq V(G)$ of size at most $k$,
  with probability at least 
  \[ \left(k^{\Oh_H\left(\sqrt{k} \log k\right)} |V(G)|^{\Oh_H(1)}\right)^{-1} \]
  we have $Z \subseteq V(G')$ and every bag of $(T,\beta)$ contains $\Oh_H(\sqrt{k}\log^2k)$ vertices
  of $Z$. 
\end{theorem}

The distance version (stated only for $K_{3,h}$-minor-free graphs) is essentially identical: the difference is that we bound the intersection with $N^d_G[Z]$ and not with $Z$.
\begin{theorem}\label{thm:main-d}
  For every integers $d,h \geq 1$, there exists a randomized
  polynomial-time algorithm that,
  given an $K_{3,h}$-minor-free graph $G$ and an integer $k \geq 2$,
  outputs an induced subgraph $G'$ of $G$ and a tree decomposition $(T,\beta)$ of $G'$
  of width $\Oh_{d,h}(k \log k)$ such that for every $Z \subseteq V(G)$ of size at most $k$,
  with probability at least 
  \[ \left(k^{\Oh_{d,h}\left(\sqrt{k} \log k\right)} |V(G)|^{\Oh_{d,h}(1)}\right)^{-1} \]
  we have $N_G^d[Z] \subseteq V(G')$ and every bag of $(T,\beta)$ contains
  $\Oh_{d,h}(\sqrt{k} \log^6 k)$ vertices of $N_G^d[Z]$.
\end{theorem}

It is conceivable that with some technical effort, Theorem~\ref{thm:main-d} can be generalized to apex-minor-free graphs. However, the statement is certainly not true for $H$-minor-free graphs in general. Consider a planar graph $G_0$ with large treewidth and let $G$ be obtained by introducing a new universal vertex $z$, adjacent to every vertex of $G_0$. Observe that $G$ is $K_6$-minor-free. For $d=1$ and $k=1$, if we set $Z=\{z\}$, then $N^d_{G}[Z]$ contains the whole graph $G$. Thus the induced subgraph $G'$ should be $G$ itself, and it cannot have a tree decomposition whose width is polynomial in $k=1$.

In this work we focus on the decision version, and leave any counting extensions
(which were the main topic of~\cite{DBLP:conf/stoc/Nederlof20a}) for future work.

\subsection{Algorithmic applications}
\label{sec:app}
In this section, we give an overview of some of the subexponential algorithms that follow from our main results. All of these algorithms require solving the given problem on bounded-treewidth graphs; more precisely, given a tree decomposition of width $k^{\Oh(1)}$ with the promise that the $k$-vertex solution interects each bag in $\Otilde(\sqrt{k})$ vertices, we need to be able to solve the problem in time $2^{\Otilde(\sqrt{k})}\cdot n^{\Oh(1)}$. The algorithms for bounded-treewidth algorithms can be designed using the standard technique of dynamic programming on tree decompositions, but this exact setting and running time is typically not considered in the literature. In this overview, we do not go into the details of these dynamic programming algorithms: describing such dynamic programing algorithms is notoriously tedious and these details are completely independent of our main techical contributions.

For those applications where we invoke the $d\ge 1$ case of Theorem~\ref{thm:main-d}, we get algorithms only for $K_{3,h}$-minor-free graphs for fixed $h$, but not for $H$-minor-free graphs for any fixed $H$. It remains an interesting open question if these results can be generalized to $H$-minor-free graphs using some problem-specific approach or, preferably, by some appropriate generalization of the framework. However, one can easily find (artificial) problems where Theorem~\ref{thm:main-d} can be used on $K_{3,h}$-minor-free graphs for any $h$, but the problem is NP-hard for $k=1$ on $K_6$-minor-free graphs. For example, an algorithm that solves the problem  ``find a set $S$ of $k$ vertices, each having degree exactly $d$, such that the closed neighborhood of $S$ induces a 4-colorable graph'' in a $K_6$-minor-free graph with $k=1$ can be used to test if a planar graphs is 3-colorable.\footnote{Given a planar graph $G$, construct the $K_6$-minor-free graph $G'$ by introducing a universal vertex. Invoke the algorithm with $k=1$ and $d=|V(G)|$ on $G'$.}

\paragraph{Subgraph isomorphism}
Theorem~\ref{thm:patterncoverold} was used by Fomin et al.~\cite{DBLP:conf/focs/FominLMPPS16} to decide whether a connected bounded-degree graph $P$ is a subgraph of an apex-minor-free graph $G$ in time $2^{\Otilde(\sqrt{|V(P)|})}|V(G)|^{\Oh(1)}$. The proof uses the fact that if a tree decomposition of $G$ of width $t$ is given, then the problem can be solved for connected, bounded-degree $P$ in time $|V(P)|^{\Otilde(t)}|V(G)|^{\Oh(1)}$. If we are given a tree decomposition of width $t_1$ with the promise is that there is a solution where every bag contains at most $t_2$ vertices of the pattern $P$, then the algorithm can be easily modified to run in time $(|V(P)|+t_1)^{\Otilde(t_2)}\cdot |V(G)|^{\Oh(1)}$. Thus, Theorem~\ref{thm:main} extends this result (and its extensions to colored, directed, weighted graphs) to $H$-minor-free graphs. 

The requirement of being connected and bounded-degree cannot be removed: Bodlaender et al.~\cite{DBLP:conf/icalp/BodlaenderNZ16} showed that a  $2^{o(|V(P)| / \log |V(P)|)}|V(G)|^{\Oh(1)}$-time algorithm for connected, unbounded-degree~$P$, or for disconnected, bounded-degree $P$ would violate the Exponential-Time Hypothesis~\cite{DBLP:journals/jcss/ImpagliazzoPZ01}. However, there is an important disconnected case of interest: where every component of the pattern is the same, or put differently, the goal is find $k$ vertex-disjoint subgraphs isomorphic to the fixed pattern $P$. For a fixed $P$, this problem can be solved in general graphs in time $2^{\Oh_P(k)}|V(G)|^{\Oh_P(1)}$ using, e.g., the technique of color coding \cite{platypus} and, given a tree decomposition of $G$ of width $t$, can be solved in time $t^{\Oh_P(t)}|V(G)|^{\Oh_P(1)}$. As Theorem~\ref{thm:main} no longer requires the pattern to be connected, it can be used to solve this problem in $H$-minor-free graphs in time $2^{\widetilde{\Oh}_{P,H}(\sqrt{k})}|V(G)|^{\Oh_{P,H}(1)}$.

\paragraph{Steiner trees}
 Given a graph $G$ with a set $T\subseteq V(G)$ of terminals, the Steiner Tree problem asks for a minimum-size tree $G_T$ that contains every vertex of $T$. The problem can be naturally extended to edge-weighted or vertex-weighted graphs. If we parameterize by the number $|T|$ of terminals, then, assuming ETH, there is no subexponential-time parameterized algorithm even on planar graphs, i.e., an algorithm with running time $2^{o(|T|)}|V(G)|^{\Oh(1)}$ \cite{DBLP:conf/focs/MarxPP18}. However, if we parameterize by the size of the solution $G_T$, then Theorem~\ref{thm:patterncoverold} can be used to get a $2^{\Otilde(\sqrt{|V(G_T)|})}|V(G)|^{\Oh(1)}$ time algorithm in apex-minor free graphs. Our Theorem~\ref{thm:main} generalizes this result to $H$-minor-free graphs.

As the terminal set $T$ is always a part of the solution tree $G_T$, another natural parameterization of the Steiner Tree problem is by the number of nonterminal vertices, i.e., by $\nu \coloneqq |V(G_T)\setminus T|$. 
We can use Theorem~\ref{thm:main-d} with $d=1$ to obtain subexponential parameterized algorithms using this parameter on $K_{3,h}$-minor-free graphs. Without loss of generality, assume that $T$ is an independent set, as any edge of $G[T]$ can be contracted without changing the answer to the problem.
Then, the key observation is that (if $|T| \geq 2$) if $X$ is the set of nonterminal vertices in the solution, then the whole tree $G_T$ is contained in $N_G[X]$. Thus Theorem~\ref{thm:main-d} with $d=1$ delivers a tree decomposition of width polynomial in $\nu$ such that every bag contains at most $\Otilde(\sqrt{\nu})$ vertices of $N_G[X]$, and hence of $G_T$. This allows us to obtain a $2^{\Otilde(\sqrt{\nu})}\cdot |V(G)|^{\Oh(1)}$-time algorithm. This reproves the result of~\cite{DBLP:journals/talg/PilipczukPSL18,DBLP:journals/algorithmica/Suchy17}.

\paragraph{Disjoint paths}
In the Disjoint Paths problem, a graph $G$ is given with $k$ pairs of terminal vertices $(s_1,t_2)$, $\dots$, $(s_k,t_k)$ and the task is to find $k$ pairwise vertex-disjoint paths $P_1$, $\dots$, $P_k$ such that $P_i$ connects $s_i$ and $t_i$. A celebrated result of Robertson and Seymour \cite{DBLP:journals/jct/RobertsonS95b} shows that this problem is FPT parameterized by the number $k$ of terminals pairs (see also \cite{DBLP:journals/siamcomp/LokshtanovMPSZ25,DBLP:conf/focs/KorhonenPS24,DBLP:journals/jct/KawarabayashiKR12}). However, it is a major open problem if length-constraints versions \cite{DBLP:journals/corr/abs-2505-03353,DBLP:conf/stacs/BentertFG25,DBLP:conf/icalp/AkmalWW24,DBLP:conf/stacs/Schlotter24,DBLP:journals/siamdm/BentertNRZ23,DBLP:conf/soda/Lochet21,DBLP:conf/esa/Berczi017} of the problem are also FPT parameterized by $k$, for example, if we want to minimize the total length $\ell$ of the paths. If we consider the more modest question of parameterization by the total length $\ell$, then standard techniques such as Color Coding \cite{platypus} can be used to obtain FPT algorithms. Moreover, our Theorem~\ref{thm:main} gives $2^{\Otilde(\sqrt{\ell})}|V(G)|^{\Oh(1)}$ time algorithms on $H$-minor-free graphs: the solution consists of $\Oh(\ell)$ vertices and the Disjoint Paths problem with various length constraints can be solved on bounded-treewidth graphs using standard techniques. As the solution does not necessarily induces a connected graph, it is essential that Theorem~\ref{thm:main}, unlike Theorem~\ref{thm:patterncoverold}, does not require connectedness.

\paragraph{Densest subgraph}
Given a graph $G$ and integer $k$, the Densest Sugraph problem asks for a set $S$ of $k$ vertices such that $G[S]$ has the maximum number of edges. On general graphs, the problem is clearly at least as hard as $k$-Clique. As the problem can be solved in time $2^{\Oh(t)}|V(G)|^{\Oh(1)}$ if a tree decomposition of width $t$ is given \cite{DBLP:journals/eor/BourgeoisGLMP17}, Theorem~\ref{thm:main} can be used to obtain a subexponential parameterized algorithm for this problem on $H$-minor-free graphs. 
Note that, even though the problem asks for a ``dense'' subgraph, there is no reason to assume that $G[S]$ is connected for the optimal $S$. Hence it is essential again that Theorem~\ref{thm:main} does not require connectivity.

\paragraph{Partial cover problems}
Partial versions of optimization problems ask for a solution set $S$ of size $k$ that potentially does not satisfy all the constraints, but maximizes the number of satisfied constraints. This generalization can make the problem significantly harder: for example, Partial Vertex Cover (find $k$ vertices that cover at least $\ell$ edges) is W[1]-hard parameterized by~$k$. However, the problem becomes FPT on planar and $H$-minor-free graphs and even admits subexponential parameterized algorithms \cite{DBLP:journals/ipl/FominLRS11,DBLP:journals/jcss/AminiFS11}. We point out that Theorem~\ref{thm:main} can be also used to obtain such algorithms, as the problem can be solved on tree decompositions using standard techniques (for example, note that the number of edges covered by $S$ is the total degree of $S$ minus the number of edges in $G[S]$, thus techniques similar to Densest Subgraph are applicable). This algorithm has the advantage of being easily generalizable to other problem variants, such as vertices/edges having weights or colors etc.

One can generalize the Partial Vertex Cover problem the following way: instead of intersecting the maximum number of edges, we want to intersect the maximum number of cliques of size~$c$; clearly, $c=2$ is the same as Partial Vertex Cover. For $c>2$, the number of $K_c$-subgraphs intersected by $S$ cannot be determined from $G[S]$ only, but it can be determined from $G[N_G[S]]$. Thus on $K_{3,h}$-minor-free graphs, we can use Theorem~\ref{thm:main-d} with $d=1$ to obtain a tree decomposition where every bag contains $\Otilde(\sqrt{k})$ vertices of $N_G[S]$. On such a tree decomposition, we can use standard techniques to solve the problem in time $2^{\Otilde(\sqrt{k})}|V(G)|^{\Oh(1)}$.

Another well-studied partial optimization problem is Partial Dominating Set, where the task is to find a set $S$ of $k$ vertices such that the closed neighborhood $N_G[S]$ has size at least $\ell$ \cite{DBLP:journals/talg/DemaineFHT05,DBLP:journals/ipl/FominLRS11}. A further natural generalization is to maximize the size of the distance-$d$ beighborhood $N^d_G[S]$ for some fixed $d$. On $K_{3,h}$-minor-free graphs for fixed $d$, Theorem~\ref{thm:main-d} can be used to obtain a $2^{\Otilde_{d,h}(\sqrt{k})}|V(G)|^{\Oh(d,h)}$-time algorithm.

Let us remark that the robust methodology of these results allows combining different type of problems. For example, we can consider the problem of finding a cycle of length $k$ dominating at least $\ell$ vertices, or finding $k$-disjoint triangles covering the maximum number of egdes, etc.

\paragraph{Distance constraints}
One of our main contributions is that Theorem~\ref{thm:main-d} for $K_{3,h}$-minor-free graphs allows us to express problems that involve distance constraints. Generalizing (Partial) Dominating Set to its distance-$d$ version was one example. But it is natural to consider distance constraints in the context of other problems as well, for example:
\begin{itemize}
\item Find $k$ copies of a fixed pattern $P$ in $G$ that are pairwise at  distance at least $d$.
\item Find vertex-disjoint paths of total length $k$ connecting given terminal pairs such that the paths are at distance at least $d$ from each other.
\item Find a cycle of length exactly $k$ such that there is no ``shortcut'' of length at most $d$ connecting two vertices of the cycle with edges outside the cycle.
  \end{itemize}
  For such problems, if $S$ is the solution of size $k$, then the validity of the solution cannot be verified by knowing only $G[S]$: we need to know $G[N^d_G[S]]$ to verify the validity of the solution. Therefore, the right way to think about the problem is the following:
  \begin{enumerate}
  \item We want to find two sets: a set $S$ of size $k$, and a set $S^*$ of unbounded size.
\item We want to find an assignment $\delta\colon S^*\to \{0,1,\dots,d\}$ such that $\delta(v)$ is the distance of $v$ from $S$.
  \item We want to verify that $S^*$ is exactly $N^d_G[S]$.
    \item We want to verify that $G[S^*]$ shows that $S$ is indeed a correct solution.
    \end{enumerate}
Theorem~\ref{thm:main-d} gives us a set $Z$ with $S^*\subseteq Z$ and a tree decomposition of $G[Z]$ 
    such that every bag contains $\Otilde(\sqrt{k})$ vertices of $S^*$. If we find the sets $S$ and $S^*$, and the assignment $\delta$, then we can verify if $S^*$ is indeed $N^d_G[S]$ (in particular, we need to verify that if $\delta(v)<d$, then $v$ has no neighbor outside $S^*$). Then it is routine to verify that $S$ satisfies the requirements.
  
\paragraph{Map graphs}
A well-studied generalization of planar graphs are map graphs \cite{DBLP:journals/jacm/ChenGP02}, which are intersection graphs of a finite collection of simply connected and internally disjoint regions in the plane. One way to describe a map graph is to take a planar graph $G$ and a set $N$ of faces; the map graph $M$ has vertex set $N$ and two vertices are adjacent if the corresponding two faces of $N$ share a vertex. Despite the fact that map graphs can have arbitrary large cliques, some of the algorithmic results for planar graphs can be generalized to map graphs \cite{DBLP:conf/latin/ByrkaLMSU20,DBLP:conf/icalp/FominLP0Z19,DBLP:conf/fct/EickmeyerK17,DBLP:journals/talg/DemaineFHT05,DBLP:journals/jacm/ChenGP02,DBLP:journals/jal/Chen01a,DBLP:conf/stacs/FabianskiPST19}.

Given $G$ and $N$ as above, let $G'$ be the bipartite planar graph where one side is $N$, the other side is the set of vertices of $G$, and a face in $N$ is adjacent in $G'$ to all its vertices. Then $x,y\in N$ are adjacent in $M$ if and only if they are distance two in the planar graph $G'$. Using this equivalence, algorithmic problems on $M$ can be rephrased as problems on the planar graph~$G'$. In particular, the graph $M[S]$ can be determined from the graph $G'[N_{G'}[S]]$ (and, more generally, $M[N^d_M[S]]$ can be determined from $G'[N^{2d+1}_{G'}[S]]$). Thus, Theorem~\ref{thm:main-d} allows us to obtain all the algorithmic results described above in map graphs as well.

\section{Overview}
In this section, we present a brief overview of the main results and techniques in the paper. The main technical content of the paper consists of three parts:
\begin{itemize}
\item To handle the $d\ge 1$ case, we need a substantial improvement of the duality result of Fomin et al.~\cite{DBLP:journals/siamcomp/FominLMPPS22}, taking into acount the $d$-neighborhood of the pattern (Section~\ref{sec:duality}).
\item We formulate a somewhat technical separator improvement process, which is used repeatedly in the main algorithm (Section~\ref{sec:improve}).
\item We present the main algorithm first for ``nearly embeddable graphs with apices'' (Section~\ref{ss:nearly}), then extend it to $H$-minor-free graphs for $d=0$ (Section~\ref{sec:lift}) and to $K_{3,h}$-minor-free graphs for $d\ge 1$ (Section~\ref{ss:algo-d}).
\end{itemize}

Even though we prove Theorems~\ref{thm:main} and \ref{thm:main-d} together with the same proof, we explain first the ideas behind the proof of Theorem~\ref{thm:main} and then discuss some of the challenges in extending the proof to the $d\ge 1$ case of Theorem~\ref{thm:main-d}.

\subsection{Duality result}
One of the main tools developed by Fomin~et al.~\cite{DBLP:journals/siamcomp/FominLMPPS22} (and used also by Nederlof~\cite{DBLP:conf/stoc/Nederlof20a}) is the following duality
theorem:
\begin{theorem}[Theorem~9 in~\cite{DBLP:journals/siamcomp/FominLMPPS22}]\label{thm:dual1}
  There is a polynomial-time algorithm that given a connected graph $G$,
  a pair $s,t \in V(G)$ of different vertices, and positive integers $p,q$, outputs
  one of the following structures in~$G$:
  \begin{itemize}
  \item a chain $(C_1,\ldots,C_p)$ of $(s,t)$-separators with $|C_j| \leq 2q$ for each $j \in [p]$; or
  \item a sequence $(P_1,\ldots,P_q)$ of $(s,t)$-paths such that for every $i \in [q]$, at most $4p$ vertices of $P_i$ lie on other paths $P_j$, $j \in [q] \setminus \{i\}$. 
  \end{itemize}
\end{theorem}

In the first case, that $(C_1,\ldots,C_p)$ is a chain of separators means that each $C_i$ separates $\{s\}\cup C_1\cup \ldots\cup C_{i-1}$ from $C_{i+1}\cup \ldots \cup C_p\cup \{t\}$.
In the second case, the vertices of $P_i\setminus \{s,t\}$ that appear on other paths $P_j$ are called 
\emph{public vertices}, whereas all other vertices are \emph{private vertices}. Vertices $s$ and $t$ are treated in a special way, and they are neither private nor public.

A typical usage of \cref{thm:dual1} in~\cite{DBLP:journals/siamcomp/FominLMPPS22,DBLP:conf/stoc/Nederlof20a}, where an unknown pattern graph
on at most $k$ vertices is hidden in the host graph, is to set $p \sim \sqrt{k}$ and $q = \mathrm{poly}(k)$.
In the first case, one of the separators $C_j$ will have small (of size at most $\sqrt{k}$) intersection
with the pattern and, with $|C_j|$ bounded polynomially in $k$, this intersection can be guessed from a
subexponential number of choices.
In the second case, a randomly chosen path $P_i$ intersects the pattern only in public vertices
with high probability;
private vertices can be reduced (e.g., deleted or contracted), decreasing the length of $P_i$
to at most $4p \sim 4\sqrt{k}$. 

Let us rephrase Theorem~\ref{thm:dual1} into the following form that will be consistent
with the distance version developed later.
First, we emphasize the existence of a separator/path with small intersection with the unknown set $Z$. Second, as a technical generalization, we also want to have a small intersection with a known set $Z_0$. The key point here is that the unknown $Z$ contributes $|Z|$ to the lower bound on $q$, while the known $Z_0$ contributes only $|Z_0|/p$.

\begin{corollary}\label{cor:dual1}
  There is a polynomial-time algorithm that given a graph $G$,
  a pair $s,t \in V(G)$ of different vertices, a set $Z_0 \subseteq V(G)$
  of size $k_0$, and positive integers $p,q\ge 2,k$ with
 $q\ge k+k_0/p+1$ outputs
  one of the following structures in~$G$:
  \begin{itemize}
  \item a chain $(C_1,\ldots,C_p)$ of $(s,t)$-separators in $G$ 
  of size at most $2q$ such that for any $Z\subseteq V(G)$ of size at most $k$,
  there is a $j \in [p]$ with $|Z_0 \cap C_j| \leq 2k_0/p$ and $|Z \cap C_j|\le  2k/p$.
  \item a sequence $(P_1,\ldots,P_q)$ of $(s,t)$-paths in $G$ and
  a sequence $(Q_1,\ldots,Q_q)$ of sets of size at most $5p+2$ with
  $Q_i\subseteq V(P_i)$ such that
  for any $Z\subseteq V(G)$ of size at most $k$,
  there is a $j \in [q]$ with 
  $(Z_0 \cup Z) \cap (V(P_j) \setminus Q_j) = \emptyset$.
  \end{itemize}
\end{corollary}
\begin{proof}
Apply Theorem~\ref{thm:dual1} to $G$, $s$, $t$, $p$, and $q$. 
If a chain $(C_1,\ldots,C_p)$ is obtained, just return it as the first outcome:
as they are pairwise disjoint, less than $p/2$
separators have more than $2k_0/p$ elements of $Z_0$ and less
than $p/2$ separators have more than $2k/p$ elements of $Z$.
If a sequence $(P_1,\ldots,P_q)$ is returned, then return it together with
$Q_i$ being equal to $\{s,t\}$ plus the set of public vertices of $P_i$
plus at most $p$ private vertices of $P_i$ that are members of $Z_0$.
Since $q \geq k+k_0/p+1$, there exists $i \in [q]$ such that no private vertex
of $P_i$ is in $Z$ and at most $p$ private vertices of $P_i$
are in $Z_0$, and thus $V(P_i) \cap (Z_0 \cup Z) \subseteq Q_i$, as desired. 
\end{proof}

\subsection{Moving to a sparse separation}

A key step in the algorithm of Nederlof~\cite{DBLP:conf/stoc/Nederlof20a} is that if we have a balanced separator that has a large intersection with an unknown solution
$Z$ of size $k$, then with $\textup{quasipoly}(k)$ guesses we can move to a balanced separator that has only $\Otilde(\sqrt{k})$ intersection with $Z$. Essentially the same argument is used multiple times in the algorithm of Nederlof~\cite{DBLP:conf/stoc/Nederlof20a} in slightly different settings. Our main result is also using a similar kind of an argument. We formulate an abstract statement that can cover all required uses of this argument and prove it for $H$-minor-free graphs and a generalization for $d\ge 1$ on $K_{3,h}$-minor-free graphs (Theorem~\ref{thm:improve}). As the $d\ge 1$ case involves lots of additional technicalities, we present here a simplified version relevant for $d=0$.

Recall that a separation of $G$ is a pair $(A,B)$ of subsets of vertices such that $V(G)=A\cup B$, $A\setminus B\neq \emptyset$, $B\setminus A\neq \emptyset$, and $G$ has no edge between $A\setminus B$ and $B\setminus A$. The order of the separation $(A,B)$ is $|A\cap B|$. As we said earlier, a key goal is to find separators of size $\textup{poly}(k)$ that intersect the solution $Z$ only in $\Otilde(\sqrt{k})$ vertices. But since we are looking for a tree decomposition of an induced subgraph of $G$, it is also beneficial if we find a separator consisting of two sets of vertices: a ``sparse'' part of size $\textup{poly}(k)$ that intersects the solution $Z$ only in $\Otilde(\sqrt{k})$ vertices, and an ``irrelevant'' part that has unbounded size, but disjoint from $Z$. In case we have such a separator and partition, then we can remove the second part from the graph.

\begin{theorem}\label{thm:simplifiedimprove}
  For every fixed integer $h \geq 4$, there exists a randomized polynomial-time algorithm 
  that, given on input a $K_h$-minor-free graph $G$,
  an integer $k \geq 2$,
  a separation $(A^\circ,B^\circ)$ in $G$
  and a spanning forest $F^\circ$ of $G[A^\circ \cap B^\circ]$ 
  with less than $h$ connected components and maximum degree at most $h$, 
  samples a tuple
  $(A,B,C)$ of subsets of $V(G)$
  satisfying the following properties:
  \begin{enumerate}[itemsep=0px]
  \item $(A,B)$ is a separation in $G$,
  \item $C \subseteq A \cap B$,
  \item $|C| \leq |A^\circ \cap B^\circ| + \Oh_{h}(k \log k)$,
  \item For every positive integer $\theta$, set $W \subseteq V(G)$, and a set $Z \subseteq V(G)$
  such that 
   \begin{itemize}[itemsep=0px]
   \item $|W| \geq 4\theta$,
   \item $|A^\circ \cap W| \geq \theta$, $|B^\circ \cap W| \geq \theta$, and
   \item $|Z| \leq k$, 
   \end{itemize}
  with probability 
  $\left(\Oh_{h}(k|A^\circ \cap B^\circ|)\right)^{-\Oh_h(\log k)}$
  we have the following:
  \begin{itemize}[itemsep=0px]
  \item 
  $|A \cap W| \geq \theta$, $|B \cap W| \geq \theta$, 
  \item $A \cap B \cap Z \subseteq C$, 
  \item $|A \cap B \cap Z|  =|C \cap Z|$ is bounded by
  $\Oh_{h}(\sqrt{k} \log k)$,
  \item furthermore, if 
   $ |Z \cap A^\circ \cap B^\circ| > h(h+1) \sqrt{k} $, then $|Z \cap A \cap B| \geq \sqrt{k}$.
  \end{itemize}
  \end{enumerate}
\end{theorem}
Let us try to parse this statement. Nederlof \cite{DBLP:conf/stoc/Nederlof20a} considers only cycle separators in planar graphs, that is, $G[A^\circ\cap B^\circ]$ has a 2-regular spanning subgraph. The existence of this cycle shows, for example, that it is possible to partition the separator into multiple connected parts such that each part contains roughly the same number of vertices of the solution. In our generalization, we replace the requirement of cycle separators with a more abstract property: a spanning forest with bounded number of components and bounded degree. The existence of such a forest shows a similar property: each tree can be partitioned into multiple parts such that the number of solution vertices in the different parts differ only by at most a constant factor (or is one of the components of $F^\circ$, but there is only a constant number of them and do not cause problems).
Such a separator with a suitable spanning tree can be obtained from product structure theory results of~\cite{DBLP:journals/jacm/DujmovicJMMUW20}.

We are given a separation $(A^\circ,B^\circ)$ of $G$ having order $\ell$, and we assume that this separation is balanced with respect to some unknown $W$ of arbitrary size (i.e., at least  $\frac{1}{4}$-fraction of $W$ is in $A^\circ$ and in $B^\circ$). Without knowing $W$ or the solution $Z$ of size at most $k$, we need to sample with probability $(k\ell)^{-\Oh(\log k)}$ a separator $(A,B)$ such that $(A,B)$ still satisfies the same balance requirement with respect to $W$, but has only intersection of size $\Oh(\sqrt{k}\log k)$ with the unknown solution $Z$. More precisely, the order of $(A,B)$ can be arbitrary large, but together with $(A,B)$, we also need to output a guess $C\subseteq A\cap B$ that covers every vertex in $Z\cap A\cap B$; this partitions the set $A\cap B$ into a sparse part $C$ and irrelevant part $(A\cap B)\setminus C$. The size $|C|$ of this guess can be at most the order $\ell$ of $(A^\circ,B^\circ)$ plus $\Oh(k\log k)$. Note that there is no bound on the order of $(A,B)$, but the vertices of $(A\cap B)\setminus C$ are not in the solution, hence they are irrelevant and can be removed from the graph. 

We present a very robust proof for Theorem~\ref{thm:simplifiedimprove} that does not use any topological arguments: our algorithm either samples $(A,B)$ as required or exhibits an $H$-minor in $G$. Suppose for simplicity that the spanning forest of $A^\circ\cap B^\circ$ can be partitioned into $h$ connected sets $J_1$, $\dots$, $J_h$ such that each of them intersects $Z$ in roughly the same number of vertices. We can guess this partition by guessing $h-1$ edges of the spanning tree (that is, $(\ell-1)^{h-1}$ possibilities). We guess which of $|A\cap W|$ and $|B\cap W|$ is larger (two possibilities). Thus in the following, without loss of generality we assume that $|A\cap W|$ is larger, hence it is at least $|W|/2$.

\begin{figure}[tb!]
  \begin{center}
    \begin{tikzpicture}
      \tikzset{
        large circle/.style={
        circle,
        draw,
        fill=black!20,
        minimum size=1cm,
        inner sep=10pt
      }
    }

\draw[fill=black!20] (-60:4)
    arc (-60:60:4)
    -- (1, 1) -- (-1, 1) -- (120:4)
    arc (120:240:4)
    -- (-1, -1) -- (1, -1) -- (-60:4) --cycle;

    \draw (0,0) circle (4);
    \node[large circle,fill=green!20] (J1) at (90:4) {$J_1$};
    \node[large circle] (J2) at (135:4) {$J_2$};
    \node[large circle] (J3) at (45:4) {$J_3$};
    \node[large circle] (J4) at (180:4) {$J_4$};
    \node[large circle] (J5) at (0:4) {$J_5$};
    \node[large circle] (J6) at (225:4) {$J_6$};
    \node[large circle] (J7) at (315:4) {$J_7$};
    \node[large circle,fill=green!20] (J8) at (270:4) {$J_8$};
    \draw (0,0) node {$S$};
    \end{tikzpicture}
  \caption{Separating two sets $J_1$ and $J_h$.}\label{fig:sepJJ}
  \end{center}
\end{figure}
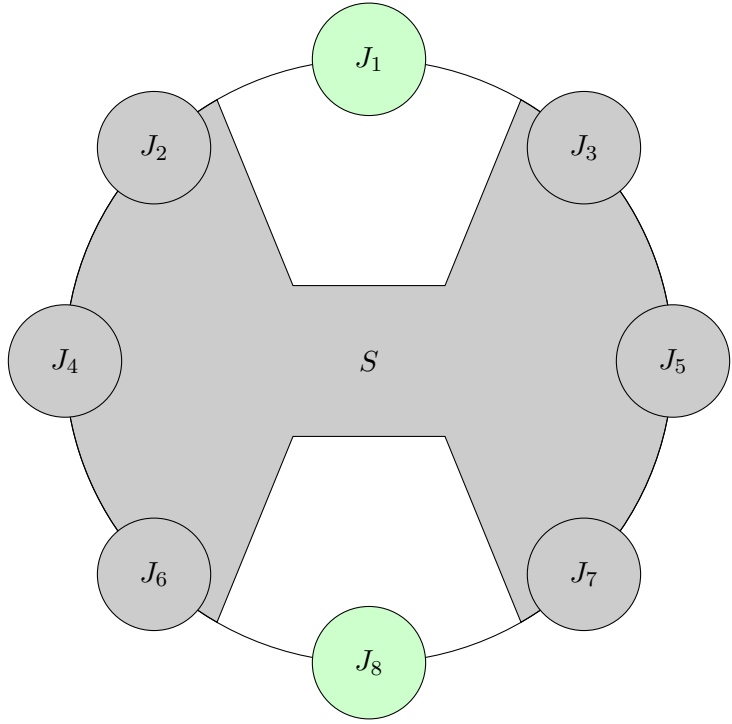

Our goal is to guess a certain separator $S$ that separates two sets $J_i$ and $J_j$ for some $i,j\in[h]$. For notational simplicity, let us try to find such a separator for $i=1$ and $j=h$: we want a set $S$ that is a $J_1-J_h$ separator in $G-(J_2\cup\dots \cup J_{h-1})$ (see Figure~\ref{fig:sepJJ}), has size $\textup{poly}(k)$, and $|S\cap Z|=\Otilde(\sqrt{k})$.  If we have such a set $S$, then this means that there is a separation $(A',B')$ of $G$ with $A'\cap B'=(J_2\cup\dots \cup J_{h-1})\cup S$, $J_1\subseteq A'$, $J_h\subseteq B'$. We guess which of the two sets $J_1$ and $J_h$ is in the part with a larger intersection with $A\cap W$ (two possibilities); without loss of generality, let us assume that $|A'\cap (A\cap W)|\ge |A\cap W|/2\ge |W|/4$. Now we uncross the two separations $(A,B)$ and $(A',B')$ and define the separation $(A'',B'')$ by $A''=A\cap A'$ and $B''=B\cup B'$. Note that we have $A''\cap B''\subseteq (J_1\cup J_2\cup\dots \cup J_{h-1})\cup S$. We observe the following: 
\begin{itemize}
\item $|A''\cap W|=|A'\cap (A\cap W)|\ge |W|/4$, since we assumed that $A'$ contains at least half of the vertices of $A\cap W$.
\item $|B''\cap W| \ge |W|/4$, since we have $B\subseteq B''$ and $|B\cap W|\ge |W/4|$ by assumption.
  \item When moving from $(A,B)$ to $(A'',B'')$, we removed $J_h$ from the separator (which is supposed to contain roughly $1/h$ fraction of the intersection with $Z$), at the cost of introducing $S$ (which is supposed to contain only $\Otilde(\sqrt{k})$ vertices of $Z$). This reduces the intersection with $Z$ by a constant factor, unless this intersetion size is already  $\Otilde(\sqrt{k})$.
  \end{itemize}
  Therefore, we can repeat the procedure for this new separation $(A'',B'')$. After $\Oh(\log k)$ iterations, the intersection size decreases to $\Otilde(\sqrt{k})$. By repeatedly introducing the set $S$, the order of the separation increases by a total of $\textup{poly}(k)$ and the increase of intersection with $Z$ is $\Otilde(\sqrt{k})$.

  But how do we find the required $J_i-J_j$ separator that has $\Otilde(\sqrt{k})$ intersection with $Z$?  Here is the point where the duality result of Corollary~\ref{cor:dual1} is used. The basic idea is that, for every $i,j\in [h]$, we invoke Corollary~\ref{cor:dual1} on the graph obtained from $G$ by contracting set $J_i$ into a single vertex $s$, contracting set $J_j$ into a single vertex $t$, and removing every other sets $J_s$. If for some $i,j\in [h]$, Corollary~\ref{cor:dual1} returns a sequence of disjoint separators, then one of them has intersection of size $\Otilde(\sqrt{k})$ with $Z$, and we guess this separator. If for every $i,j\in [h]$, we get a path $P_{i,j}$, then we try to use these paths to show that the graph has a $K_h$-minor. If the paths are pairwise disjoint, then contracting each $J_i$ into a single vertex and each path into a single edge certainly gives a $K_h$-minor. However, the paths are not necessarily disjoint, so we have to use this argument in a more careful, iterative manner.

  We go through the pairs $(i,j)\in \binom{h}{2}$ in an arbitrary order. If Corollary~\ref{cor:dual1} gives a path $P_{i,j}$, then we remove the internal vertices of the path from the graph, and continue with the next pair $(i,j)$ on the remaining graph. This way, the collected paths will be indeed pairwise disjoint, hence if this iteration goes over every pair $(i,j)$, then the resulting $h$ connected sets $J_1$, $\dots$, $J_h$, and the $\binom{h}{2}$ pairwise disjoint paths show create a $K_h$-minor in the graph. However, if the iteration stops at some pair $(i,j)$ with a $J_i-J_j$ separator $S$, then $S$ is only a separator in the remaining graph. In other words, $S$ is a separator in the original graph only if we extend it with the union $P$ of all the paths $P_{i,j}$ that were iteratively removed. Unfortunately, the size of these paths is unbounded. But recall that each such path $P_{i,j}$ in Corollary~\ref{cor:dual1} come equipped with a set $Q_{i,j}$ of size $\Otilde(\sqrt{k})$, and Corollary~\ref{cor:dual1} not only stated that for one of the paths $P_{i,j}$ the intersection with $Z$ is bounded, but also that $P_{i,j}$ intersects $Z$ only in $Q_{i,j}$. Thus if we selected such a path in each of the at most $\binom{h}{2}$ steps of the iteration, then we can assume that their union $P$ intersects $Z$ only in the union $Q$ of all the $Q_{i,j}$'s, which has size $\Otilde(h^2 \sqrt{k})$. This means that $S\cup P$ is the type of separator that we like: it has a sparse part $S\cup Q$ of size $\textup{poly}(k)$ that contains $\Otilde(\sqrt{k})$ vertices of $Z$, and an unbounded part $P\setminus Q$ that is disjoint from $Z$. Therefore, we can proceed with this $J_i-J_j$ separator $P\cup S$ as described above.

  \subsection{The main algorithm}
  \label{sec:overviewalg}
  
  In this overview, we describe the algorithm of Theorem~\ref{thm:main} for bounded-genus graphs, which already demonstrates the main technical ideas.  Moreover, we assume that the diameter of the graph is $\Oh(k)$: with a standard application of Baker's layering approach, we can make this assumption.
  It is known that if a bounded-genus graph has diameter $\Oh(k)$, then it has treewidth $\Oh(k)$, hence has balanced separations of order $\Oh(k)$. But in order to be able to use Theorem~\ref{thm:simplifiedimprove}, we need to be able to find separators that have a bounded-degree spanning trees consisting of only a bounded number of components. We deduce from (a suitable variant of) the Product Structure Theorem of Dujmovi\'c et al.~\cite{DBLP:journals/jacm/DujmovicJMMUW20} that such balanced separators exists in connected bounded-genus graphs (and, more generally, graphs nearly-embeddable in a surface without apices).

\begin{corollary}\label{cor:XF0}
For every $g\ge 0$, there is a constant $c$ such that the following holds. 
There is a polynomial-time algorithm that, given a graph of genus at most $g$
of diameter $\Delta$ and a weight function $\mu:V(G) \to \mathbb{R}_{\geq 0}$, computes
a balanced separator $X$ of $G$ with respect to $\mu$ and a spanning forest $F$ of $G[X]$
such that 
\begin{enumerate}
\item $|X| \leq c^2 \cdot (\Delta+1)$, 
\item $F$ has at most $c$ connected components, and
\item $F$ has maximum degree at most $3c^2-1$.
\end{enumerate}
\end{corollary}

In Theorem~\ref{thm:main}, we need to construct a tree decomposition of width $\textup{poly}(k)$ of an induced sugraph $G'$ of $G$ such that $G'$ fully contains $Z$ and every bag contains $\Otilde(\sqrt{k})$ vertices of $Z$. Conceptually, it can be easier to think of this goal the following way: we want to find a tree decomposition of the original graph $G$ itself, where there is no bound on the width of the decomposition, but we need to present a partition of each bag into ``useful'' and ``forbidden'' vertices such that no solution vertex is marked forbidden, every bag has only $\textup{poly}(k)$ vertices marked as useful, and every bag contains only $\Otilde(\sqrt{k})$ vertices of the solution. Then we can obtain $G'$ and its tree decomposition by simply removing any vertex that was marked as forbidden in any of the bags.

The basic idea of the main algorithm is to execute a recursive procedure that tries to find a balanced separation $(A,B)$ of the graph $G$, recursively finds a tree decomposition for each of $G[A]$ and $G[B]$, and then puts together these two decompositions to obtain a tree decomposition of $G$. As $(A,B)$ is balanced in terms of the number of vertices of $G$, the depth of the recursion is $\Oh(\log n)$.

In order to be able to join a decomposition of $G[A]$ and a decomposition of $G[B]$ together, algorithms of this type usually solve a more general problem, where a set $R\subseteq V(G)$ is given in the input, and the root bag of the decomposition should contain $R$. We also require that no vertex of $R$ can be made forbidden. If we include $A\cap B$ into $R$ when recursively constructing the decompositions for $G[A]$ and $G[B]$, then the two decompositions can be joined together. However, this way the size of $R$ (and therefore the sizes of the bags) increases in each recursive call: even if $A\cap B$ is only $\textup{poly}(k)$, the set $R$ can end up being of size $\Oh(\textup{poly}(k)\cdot \log n)$.
The natural idea to solve this problem is to include recursive steps where $(A,B)$ is selected to be a balanced separator with respect to $R$, not with respect to $V(G)$. If such steps are included sufficiently often, then $R$ cannot grow arbitrarily large.

As we want to construct a tree decomposition where every bag intersects the solution $Z$ in $\Otilde(\sqrt{k})$ vertices, in particular the set $R$ and therefore the separator $A\cap B$ have to satisfy this property. The main challenge of the algorithm is to ensure that, with high probability, every separation $(A,B)$ we consider has this property. We try to ensure this the following way. First, we use Corollary~\ref{cor:XF0} to obtain a balanced separation $(A,B)$ (with respect to $|V(G)\setminus R|$ or $R$). Then with probability $1-1/k$, we guess that $|(Z \setminus R)\cap A\cap B|\leq \Otilde(\sqrt{k})$, hence we can use the separation $(A,B)$ directly, as described above. Otherwise, with probability $1/k$, we use Theorem~\ref{thm:simplifiedimprove} to 
enhance $(A,B)$ with $k^{-\Oh(\log k)}$ success probability and proceed with the recursion using this enhanced separation. Recall that in Theorem~\ref{thm:simplifiedimprove}, the sampled separation $(A,B)$ comes equipped with a set $C$ of size $\textup{poly}(k)$, which is supposed to cover all the solution vertices in $A\cap B$. Thus we can use $C$ to mark some vertices of $A\cap B$ as useful and we mark the remaining vertices as forbidden.

However, there is one caveat with the above approach. Albeit using a balanced separator
with respect to $R$ once in a while controls the size of $R$ in the recursion, it does not
control sufficiently the size of $Z \cap R$. Even if in a single step
we add $\widetilde{\Oh}(\sqrt{k})$ vertices of $Z$ to $R$, 
at depth $\Theta(\log n)$ they can accumulate to $\log n \cdot \widetilde{\Oh}(\sqrt{k})$,
which is too much. 
We solve this issue as follows. Once in a while, we enter a ``pattern mode'' where,
with probability $\frac{1}{k}$ we seek for a separation $(A,B)$ that separates in a balanced
way an unknown to the algorithm set $Z \cap R$. 
To this end, we observe that, first, we can sample the initial separation $(A,B)$
using the tools of~\cite{DBLP:journals/jacm/DujmovicJMMUW20} from among only $\Oh(|R|)$ choices
and, second, we can enhance it using the fact that Theorem~\ref{thm:simplifiedimprove}
can work without the knowledge of $W$ and $\theta$.
Therefore, depending on the current recursion depth, the algorithm cycles through three modes in a round-robin fashion:
\begin{enumerate}
\item finding a separation balanced with respect the number of vertices (to keep the maximum depth $\Oh(\log n)$),
\item finding a separation balanced with respect to $R$ (to keep the size of $R$ bounded), or
  \item entering pattern mode (to keep the intersection $R\cap Z$ bounded).
  \end{enumerate}
  
Let us try to analyze the probability that this scheme results in the required tree decomposition. The algorithm can make mistakes by incorrectly guessing
whether $|(Z \setminus R)\cap A\cap B|=\Otilde(\sqrt{k})$ holds or not, and incorrectly choosing a separation from the list provided by Theorem~\ref{thm:simplifiedimprove}, i.e., choosing a separation incompatible with $Z$. In order to analyze the probability of making mistakes, we can classify the recursive calls into three types:
  \begin{description}
  \item[Futile:] $V(G)\setminus R$ is disjoint from $Z$.

  \item[Minor:] $V(G)\setminus R$ is not disjoint from $Z$, but contains at most $\sqrt{k}$
  vertices of $Z$.
  \item[Major:] $V(G)\setminus R$ is not disjoint from $Z$ and contains more than $\sqrt{k}$ vertices of $Z$. 
  \end{description}
  The important trick is that in case of a recursive call of type futile, nothing can go wrong: any tree decomposition where each bag has only $\textup{poly}(k)$ useful vertices and the root bag contains $R$ is correct. Therefore, no matter whether we decide to invoke Theorem~\ref{thm:simplifiedimprove} or not, and no matter which separation we choose, the algorithm will be successful with probability 1.

  Let us observe that if a recursive call is futile, then all the descendant recursive calls are futile as well. In other words, the minor and major recursive calls form a prefix of the recursion tree. Furthermore, when using the separation $(A,B)$ for the recursion, the sets $V(G)\setminus R$ in the two recursive call are disjoint: one is a subset of $A\setminus B$, the other is a subset of $B\setminus A$. Therefore, if we look at the leaves of the prefix of the non-futile recursive calls, then the sets $V(G)\setminus R$ are pairwise disjoint. This means that at most $|Z|=k$ of them can have nonempty intersection with $Z$, implying that there are at most $k$ leaves in this prefix. As the depth of the recursion tree is $\Oh(\log n)$, it follows that there are $\Oh(k\log n)$ minor and major recursive calls.
  
  At every minor or major call, we want to correctly guess whether
  $|(Z \setminus R) \cap A \cap B| \leq \widetilde{\Oh}(\sqrt{k})$. 
  Recall that we enter this case with probability $1-1/k$;
  as there are there are $\Oh(k\log n)$ minor and major recursive calls, with probability
  \[
(1-1/k)^{\Oh(k\log n)}=n^{-\Oh(1)}
    \]
    we correctly guess at every minor and major call 
    that $|(Z \setminus R) \cap A \cap B| \leq \widetilde{\Oh}(\sqrt{k})$
    if this is the case. We then proceed with the separation $(A,B)$ given by Corollary~\ref{cor:XF0},  without invoking Theorem~\ref{thm:simplifiedimprove}. Thus such recursive calls behave correctly.

    For the other recursive calls, if 
    $|(Z \setminus R) \cap A \cap B| > \widetilde{\Oh}(\sqrt{k})$, then the call 
    needs to be major.
    In this case, Theorem~\ref{thm:simplifiedimprove} guarantees that $A\cap B$ contains $\Omega(\sqrt{k})$ vertices of $Z \setminus R$. These vertices are moved into $R$, thus such a step ``eats'' $\Omega(\sqrt{k})$ vertices of $Z \setminus R$. That is, if such a ``rich'' recursive call for $(G,R)$ creates two instances $(G_1,R_1)$ and $(G_2,R_2)$, then $V(G_1)\setminus R_1$ and $V(G_2)\setminus R_2$ are disjoint, hence $\sum_{i\in [2]}|(V(G_i)\setminus R_i)\cap Z|\le |(V(G)\setminus R)\cap Z|-\Omega(\sqrt{k})$. It follows that there can be only $\Oh(\sqrt{k})$ such recursive calls.
    Hence, the probability that in each recursive call we select a suitable separation from the list of size $k^{\Oh(\log k)}$ returned by Theorem~\ref{thm:simplifiedimprove} is
      \[
\left(k^{-\Oh(\log k)}\right)^{\Oh(\sqrt{k)}}=2^{-\Otilde(\sqrt{k})},
        \]
        as required. 

        It remains to bound the probability that the ``pattern mode'' behaves correctly.
        Here, again futile calls do not matter. For minor calls, or major calls
        where $|Z \cap R| \leq \Otilde(\sqrt{k})$ we want never to take an
        action in the pattern mode; as there are $\Oh(k \log n)$ minor and major calls
        and we take no action with probability $1-1/k$, we are successful here
        with probability $n^{-\Oh(1)}$. 
        For a major call with large $Z \cap R$, 
        such a call nontrivially splits the pattern $Z$ between its child subcalls;
        this allows us to bound the number of such calls by $\Otilde(\sqrt{k})$,
        allowing to require success from the application of Theorem~\ref{thm:simplifiedimprove}
        in all such calls.

        Considering all types of recursive calls, we can conclude that with probability $2^{-\Otilde(\sqrt{k})}n^{-\Oh(1)}$, every recursive call works correctly, resulting in the required tree decomposition for an induced subgraph of $G$.

     \subsection{Extending to the bounded-distance version ($d\ge 1$)}
     We sketch now the main challenges in proving Theorem~\ref{thm:main-d}, the version where we want to bound the intersection with $N^d_G[Z]$ for some fixed $d\ge 1$. Compared to just proving Theorem~\ref{thm:main}, we need to resolve some fundamental issues. Moreover, additional technical complications appear that require careful handling and some new ideas.

\paragraph{Improved duality result}     As our goal now is to bound the intersection with $N^d_G[Z]$, we need a version of Corollary~\ref{cor:dual1} that not only bounds the intersection with $Z$, but also with $N^d_G[Z]$. However, the averaging argument described in the proof of Corollary~\ref{cor:dual1} breaks down: even if we have $p$ disjoint separators, in principle a single vertex of $Z$ can be adjacent to many vertices of a separator, making the intersection with $N^d_G[Z]$ very large. Similarly, in case of a collection of almost-disjoint paths, in principle a vertex of $Z$ can be a adjacent to a private vertex of each path, hence we cannot say that there is a path whose private vertices are disjoint from $N^d_G[Z]$.

     Nevertheless, in case of $K_{3,h}$-minor-free graphs, we manage to prove a strengthening of Corollary~\ref{cor:dual1} to bound the intersection with $N^d_G[Z]$ (Theorem~\ref{thm:dualdistance} below). In fact, we need an even stronger result: during our algorithm, we often invoke Corollary~\ref{cor:dual1} on an induced subgraph $G-X$ of $G$. In case of $d\ge 1$, it is possible that a vertex of $v$ of $G-X$ is in $N^d_G[Z]$ because of a vertex $z\in X\cap Z$, or because of a vertex $z\in Z\cap (V(G)-X)$ that is at distance more than $d$ from $v$ in $G-X$ (i.e., the shortest $z-v$ path uses vertices of $X$). Therefore, we cannot completely disregard the deleted set $X$, even if we want to find the separators/disjoint paths in $G-X$. We state a version of Corollary~\ref{cor:dual1} that can take into account such a set $X$, as long as it induces a bounded number $\lambda$ of components. In our application, we can bound the number of components of $G[X]$.

There is another aspect of Theorem~\ref{thm:dualdistance} in which it provides a technical strenghtening over Corollary~\ref{cor:dual1}. An easy argument shows that if $Z$ is a {\em known} set of size $k$, then we can find either an $s-t$ path that goes through at most $\sqrt{k}$ vertices of $Z$, or we can find an $s-t$ separator of size~$\sqrt{k}$. This can be considered as an improvement over Corollary~\ref{cor:dual1}, which shows the existence of path that goes through $\Oh(\sqrt{k})$ public vertices or an $s-t$ separator with the weaker property of having size $\Oh(k)$ and intersecting $Z$ at $\Oh(\sqrt{k})$ vertices. Theorem~\ref{thm:dualdistance} considers a situation where there is a {\em known} set $Z_0$ of size $k_0$ and an {\em unknown} set $Z$ of size $k$. We want to 
control both the intersection with $N^d_G[Z_0]$ and with $N^d_G[Z]$ in a way that gives better bounds with respect to $k_0$ than with respect to $k$. 

In our algorithm, the set $Z$ will be the unknown solution remaining in the
current graph $G$, while $Z_0$ will be the
adhesion $R$. Controlling $N^d_G[Z_0]$ will allow us to control 
$N^d_G[x]$ for vertices $x$ of the solution that are outside the scope of
the current recursive call, but that are within distance less than $d$ from the adhesion $R$
and thus $N^d_G[x]$ still intersects the scope of the current recursive call.

\begin{theorem}\label{thm:dualdistance}
  There is a polynomial-time algorithm that given a connected $K_{3,h}$-minor-free graph $G$,
  a pair $s,t \in V(G)$ of different vertices,  positive integers $p,q\ge 2,\lambda,d,k_0,k$ with
 $p\ge 100d\log_2 q$, $q\ge (p+k_0/p+k_1)(10d^2h)^{d+3}\lambda$, a set $Z_0$ of size $k_0$, and a set $X\subseteq V(G)\setminus\{s,t\}$ such that $G[X]$ has at most $\lambda$ components, outputs
  one of the following structures in~$G$:
  \begin{itemize}
  \item a chain $(C_1,\ldots,C_p)$ of $(s,t)$-separators in $G-X$ of size at most $2q$ such that 
  \begin{enumerate}
  \item for every $j \in [p]$ and $v \in V(G)$ it holds that 
    \[ \left|V(C_j) \cap N^d_G[v]\right| \leq (10d^d h)^{2d+1}(1+ \lambda\log_2(4q));\]
  \item
  for any $Z\subseteq V(G)$ of size at most $k$, there is a $j \in [p]$ with 
  \[ \left|N^d_G[Z]\cap C_j\right|\le  ((k_0+k)/p) \cdot \log_2 q  \cdot 50d(10d^dh)^{2d+1}\cdot \lambda. \]
  \end{enumerate}
  \item a sequence $(P_1,\ldots,P_q)$ of $(s,t)$-paths in $G-X$ and a sequence $(Q_1,\ldots,Q_q)$ of sets of size at most $10dp$ with $Q_i\subseteq V(P_i)$ such that for any $Z\subseteq V(G)$ of size at most $k$, there is a $j \in [q]$ with $N^d_G[Z]\cap (V(P_i)\setminus Q_i)=\emptyset$. 
  \end{itemize}
\end{theorem}

For the proof of Theorem~\ref{thm:dualdistance}, we first invoke Theorem~\ref{thm:dual1} to obtain a chain of separators or a set of almost-disjoint paths, and then we modify them to satisfy the requirements. In fact, we need to use a strengthening of Theorem~\ref{thm:dual1}. We observe that, in case of returning a set of almost-disjoint paths, the original proof of Theorem~\ref{thm:dual1} by Fomin et al.~\cite{DBLP:journals/siamcomp/FominLMPPS22} actually proves a slightly stronger statement: the public vertices can be organized into levels and the paths  go through these levels in order (Theorem~\ref{thm:dual0}).

Let us observe that if we have truly disjoint $s-t$ paths and a vertex $z\in Z$ is adjacent to $h$ of these paths (but not on any of the paths), then a contraction creates a $K_{3,h}$-minor, which contradicts our assumption on $G$ (see Figure~\ref{fig:contract2}).  We would like to use a similar argument to show that the $d$-neighborhood of a single vertex cannot have too much influence, but the existence of the public vertices significantly complicates the argument.   
Intuitively, what we show is that if a vertex of $Z$ is close to many paths $P_i$ in the sense that for each path $P_i$ there is a path of length $d$ starting from $z$ that reaches a private vertex $P_i$ without going through the public vertices of $P_i$, then we can construct a $K_{3,h}$-minor. Moreover, we introduce ``shortcuts'' in $P_i$ to obtain a new path $P'_i$ that has only a bounded number of vertices in the $d$-neighborhood of the public vertices of $P_i$. This ensures that one of the modified paths intersects $N^d_G[Z]$ only in a bounded number of vertices.
\begin{figure}
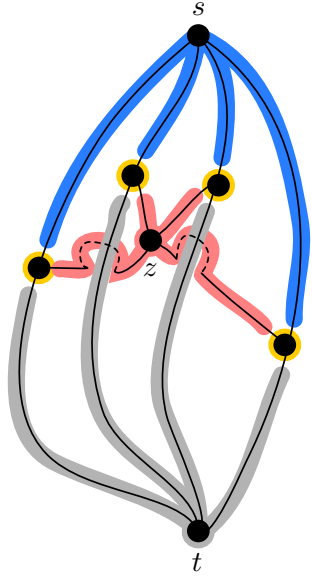

  \svgc{0.25\linewidth}{simplecontract}
  \caption{If vertex $z$ is adjacent to $h$ internally disjoint $s-t$ paths and it is not on these paths, then there is $K_{3,h}$-minor: we contract the blue subpaths into $s$ and the gray subpaths into $t$. More generally, the same is true if $z$ can reach each $s-t$ path without going through any other $s-t$ paths (paths highlighted in red), in which case we contract these paths into $z$.. }
  \label{fig:contract2}
\end{figure}  

In case of a chain of separators, as each separator $C_j$ is minimal, for every $u\in C_j$ there is an $s-t$ path $P_u$ that intersects $C_j$ only in $u$. These paths can of course intersect. Again, if there is a vertex $z$ that is adjacent to $h$ vertices of $C_j$ and is not on any of the paths $P_u$, then we can construct a $K_{3,h}$-minor by contraction. However, now in principle a vertex $z$ can be on many of the paths $P_u$, making it hard to apply this argument. Intuitively, we show that in such a case we can reduce the size of the separator $C_j$ by replacing many vertices from it and adding $z$ to it. Formally, we define a weight function $w_j(v)=2^{\dist_{G-X}(v,C_j)}$ and replace $C_j$ by a minimum-weight $s-t$ separator $C'_j$ under the weight function $w_j$. We prove that this modified separator $C'_j$ now has the property that a vertex $z\in Z$ can be close only to a bounded number of vertices of $C'_j$. We also note that many of these separators are disjoint: if $|i-j|=\Omega(\log k)$, then the exponentially increasing weights and the fact that $|C_i|,|C_j|\le \textup{poly}(k)$ ensure that $C'_i$ and $C'_j$ are disjoint.

\paragraph{Progress in presence of large neighborhoods} Another fundamental difficulty comes from the fact that (as explained in Section~\ref{sec:overviewalg}), in a major recursive call
that invokes Theorem~\ref{thm:simplifiedimprove}, we are making progress because $\Omega(\sqrt{k})$ vertices of $Z$ are put into the set $R$, hence the intersection of $V(G)\setminus R$ and $Z$ significantly decreases. However, in the $d\ge 1$ case, the size of $N^d_G[Z]$ can be unbounded, hence decreasing the intersection $N^d_G[Z]$ by $\Omega(\sqrt{k})$ is not a real progress. For example, it is possible that $R$ contains no vertex of $Z$, but $\Omega(\sqrt{k})$ neighbors of some vertex of~$Z$. This means that it is possible that repeated rich recursive calls are not making any progress besides repeatedly shaving  off $\Omega(\sqrt{k})$ neighbors of the same vertex $z\in Z$.

To resolve this problem, we need a few new techniques. 
First, we observe that Corollary~\ref{cor:XF0} can be strengthened. The argument based on the Product Structure Theorem of Dujmovi\'c et al.~\cite{DBLP:journals/jacm/DujmovicJMMUW20} actually allows us to prove that for any vertex $z$ and integer $d$, the separator contains only $\Oh(d)$ vertices at distance at most $d$ from $Z$ (see Corollary~\ref{cor:XF}). 
Hence, the initial separation $(A,B)$, before any application of
(the $d \geq 1$ variant of) Theorem~\ref{thm:simplifiedimprove}, 
behaves nicely: $|A \cap B \cap N_G^d[Z]|$ is comparable with the number of vertices of $Z$
close to $A \cap B$. 

Second, we extract from Theorem~\ref{thm:dualdistance} that a similar property
can be established for the chain of separators output. Unfortunately,
we cannot establish it for the sequence of paths output. However, here we can
rely on the fact that there are only $\Otilde(\sqrt{k})$ public vertices, as opposed
to $\Otilde(k)$ vertices on separators. This leads to a classification of the vertices
of the separator $A \cap B$ into not two as before, but three categories:
\begin{itemize}[itemsep=0px]
\item a difficult set $\widetilde{C} \subseteq A \cap B$ without any structure
of $N_G^d[Z] \cap \widetilde{C}$, but with a bound $|\widetilde{C}| \leq \Otilde(\sqrt{k})$;
\item a sparse set $C \subseteq A \cap B$ of size $\Otilde(k)$ with 
$|C \cap N_G^d[Z]| \leq \Otilde(\sqrt{k})$ and such that every vertex $v \in V(G)$
satisfies $|C \cap N_G^d[v]| \leq  \Otilde(1)$;
\item an irrelevant set $D \subseteq A \cap B$ of unbounded size
with $D \cap N_G^d[Z] = \emptyset$.
\end{itemize}
The existence of the set $\widetilde{C}$ propagates to the recursion: now we need to maintain
a set $\widetilde{R} \subseteq R$ of size $\Otilde(\sqrt{k})$ such that every
vertex $v \in V(G)$ satisfies $|N_G^d[v] \cap (R \setminus \widetilde{R})| \leq \Otilde(1)$.
To control both bounds in the previous sentence so that no unwanted $\log n$ term appears,
we need to regularly clean
$\widetilde{R}$ and old vertices in $R$, similarly as we control the size of $R$. Therefore, the algorithm has to cycle through not just three, but five different modes.

The tools above allow us to control the size of $R \cap N_G^d[Z]$
via controling the number of vertices of $Z \cap N_G^d[R]$. 
This quantity --- how many vertices of $Z$ are within distance at most $d$
from $R$ --- becomes the analog of the measure $R \cap Z$
from the $d=0$ case. This comes with two issues. First, once in a while we need to enter
the pattern mode to reduce the size of $Z \cap N_G^d[R]$.
Second, we need the following analog of the ``rich'' guess.
If the separation $(A,B)$ brings a large number of new such vertices of $Z$
in the child calls --- that is, there are more than $\Otilde(\sqrt{k})$
vertices of $Z$ that are in $N_G^d[A \cap B]$ but not in $N_G^d[Z]$,
then we can apply the $d \geq 1$ version of \cref{thm:simplifiedimprove} 
to control their number, and charge the probability of the success
of the call to the number of new vertices of $Z$ that join $N_G^d[R]$
in the child calls.

\section{Preliminaries}\label{sec:prelims}

For a positive integer $p$, we write $[p]\coloneqq \{1,\ldots,p\}$.
We also use the notation $[d \geq 1]$ that equals $1$ if $d \geq 1$ and $0$ if $d=0$.

We use standard graph notation. All graphs considered in this paper are undirected and simple, that is, without loops or parallel edges.

For a graph $G$ and a vertex $u$, by $N^G_d[u]$ we denote the {\em{$d$-neighborhood}} of $u$: the set of all vertices at distance at most $d$ from $u$. We extend the notation to sets of vertices by writing $N^G_d[X]\coloneqq \bigcup_{u\in X} N^G_d[u]$. We drop the superscript if the graph $G$ is clear from the context.

An {\em{independent set}} in a graph is a set of pairwise non-adjacent vertices. We extend this notion to directed graphs by considering independent sets in the underlying undirected graphs.

For a graph $G$ and different vertices $s,t \in V(G)$, an \emph{$(s,t)$-path} is a path with endpoint $s$ and $t$, and an
\emph{$(s,t)$-separator} is a set $C \subseteq V(G) \setminus \{s,t\}$ that intersects every $(s,t)$-path.
A sequence $(C_1,C_2,\ldots,C_p)$ of $(s,t)$-separators is a \emph{chain}
if $C_i$ are pairwise disjoint and
for every $1 \leq i < j \leq p$, $C_i$ intersects every path from $s$ to $C_j$
and $C_j$ intersects every path from $C_i$ to $t$.

A \emph{weight function} in a graph $G$ is a function $\mu \colon V(G) \to \mathbb{R}_{\geq 0}$
such that $\mu(G) \coloneqq \sum_{v \in V(G)} \mu(v) > 0$. 
A set $X \subseteq V(G)$ is a \emph{balanced separator} with regards to $\mu$
if every connected component of $G-X$ has total weight at most $\mu(G)/2$.

A tree decomposition of a graph $G$ is a pair $(T,\beta)$ such that
$T$ is a tree and $\beta \colon V(T) \to 2^{V(G)}$ satisfies (1) for every $v \in V(G)$,
the set $\{t \in V(T)~|~v \in \beta(t)\}$ induces a nonempty connected subgraph of $T$, and (2)
for every $uv \in E(G)$ there exists $t \in V(T)$ such that $u,v \in \beta(t)$. 
The set $\beta(t)$ is often called the \emph{bag} of $t$. 
A standard fact about tree decompositions is that for any weight function $\mu$
there exists $t \in V(T)$ such that $\beta(t)$ is a balanced separator for $\mu$; see e.g.~\cite[Lemma~7.19]{platypus}.

Let $(T,\beta)$ be a tree decomposition. For an edge $st \in E(T)$ 
we denote $\sigma(st) \coloneqq \beta(s) \cap \beta(t)$ and call $\sigma(st)$
the \emph{adhesion} of the edge $st$. 
If $T$ is rooted in the context, we abbreviate $\sigma(t) : =\sigma(st)$
where $s$ is the parent of $t$ in $T$, and $\sigma(r) = \emptyset$ for the root $r$ of $T$.
The \emph{torso} of $t \in V(T)$ is the subgraph $G[\beta(t)]$ with every adhesion
$\sigma(st)$ for $s \in N_T(t)$ turned into a clique. 

An \emph{apex graph} is a graph $H$ such that $H-v$ is planar for some $v \in V(H)$. 

A number of our algorithms are recursive procedures. 
The \emph{depth} of a single call to the recursive procedure is the number of 
calls on the recursion stack, including the current call. That is, the initial root call
is at depth $1$. For a parameter $p$, notation $\Oh_p(\cdot)$ hides multiplicative factors that may depend on $p$.

\subsection{Graph minors}

We will use the notion of \emph{nearly-embeddable} graphs. However, we do not need the precise
(complicated and topological) definition that can be found, e.g., in~\cite{DBLP:journals/jacm/DujmovicJMMUW20}.
Instead, we just say that for every tuple of nonnegative integers $g,p,k,a$, there is a 
class of $(g,p,q,a)$-nearly-embeddable graphs, which is closed under taking subgraphs. (For the readers familiar with the notion, $g$
is the Euler genus of the surface, $p$ is the number of vortices, $q$ is the maximum width of a vortex,
and $a$ is the number of apices.)
We abbreviate $(p,p,p,p)$-near-embeddable graphs to $p$-nearly-embeddable graphs
and $(p,p,p,0)$-nearly-embeddable graphs to $p$-nearly-embeddable graphs without apices.
We will need the fact that every $p$-nearly-embeddable graph $G$ contains a set $A$ of
at most $p$ vertices such that $G-A$ is $p$-nearly-embeddable without apices.

With this notation, 
the Structure Theorem for graph classes excluding a minor can be formulated as follows.%
\footnote{Unfortunately, we did not find a clean citation in the literature for this 
(well known in the community) statement.
It can be derived from Theorem~12 in~\cite{DBLP:journals/corr/DvorakT14} as follows.
For $H$ being an apex graph, no vertex in the apex set $A_i$ is a major apex. 
We can merge all vortices into a single vortex by adding a number of handles to the surface
and then add all apices to all bags of the resulting super-vortex.}

\begin{theorem}\label{thm:excl-minor}
  For every graph $H$ there exists a constant $p_H$ and a polynomial-time
  algorithm
  that given an $H$-minor-free graph $G$, computes a tree decomposition $(T,\beta)$ 
  of $G$ with maximum adhesion size at most $p_H$ and a set $A_t \subseteq \beta(t)$
  for every $t\in V(T)$ so that the torso of every node $t$ is $p_H$-nearly-embeddable
  and, upon removal of $A_t$, becomes $p_H$-nearly-embeddable without apices.
  Furthermore, if $H$ is an apex graph, then
  we can have $A_t = \emptyset$, that is, the entire torsos are $p_H$-nearly-embeddable without apices.
\end{theorem}

We will need the following
partial converse of Theorem~\ref{thm:excl-minor}. The first part is from~\cite{DBLP:journals/jct/JoretW13};
the second part is folklore, and can be found explicitly for example as Lemma~B.1 in~\cite{DBLP:journals/corr/abs-2304-07268}.
\begin{theorem}\label{thm:nearly-exclude}
    For every $p$ there exists a graph $H$ such that every $p$-nearly-embeddable graph
    is $H$-minor-free.
    Further, for every $p$ there exists an apex graph $H$ such that every $p$-nearly-embeddable without apices 
    graph is $H$-minor-free.
\end{theorem}

We will also use that in graphs excluding a fixed apex graph as a minor, treewidth is bounded 
linearly in the diameter.

\begin{theorem}[\cite{DBLP:conf/soda/DemaineH04}]\label{thm:local-tw}
  For every apex graph $H$ there exists a constant $c_H$
  such that given an $H$-minor-free graph $G$ of diameter $\Delta$,
  one can in polynomial time compute a tree decomposition of $G$ of width at most~$c_H \Delta$.
\end{theorem}
\begin{corollary}\label{cor:local-tw-nearly}
  For every integer $p$ there exists a constant $c_p$
  such that given a $p$-nearly-embeddable without apices graph $G$ of diameter $\Delta$,
  one can in polynomial time compute a tree decomposition of $G$ of width at most $c_p \Delta$.
\end{corollary}
\begin{proof}
The claim follows from the second statement of Theorem~\ref{thm:nearly-exclude} and Theorem~\ref{thm:local-tw}.
\end{proof}

\subsection{Tools from the product structure theory}

In our proof, we will need a result of Dujmovi\'c et al.~\cite{DBLP:journals/jacm/DujmovicJMMUW20} that apex-minor-free graphs admit a so-called product structure. To formulate this result, we need a few definitions.

Let $G$ be a graph.
A partition $\mathcal{P}$ of $V(G)$ is \emph{connected} if every part $P \in \mathcal{P}$
induces a connected subgraph of $G$. 
A \emph{layering} of $G$ is a sequence $\mathcal{L} = (V_1,V_2,\ldots,V_m)$ of pairwise disjoint subsets of $V(G)$
such that $\bigcup_{i=1}^m V_i = V(G)$ and for every edge $uv \in E(G)$, say with $u \in V_i$ and $v \in V_j$, we have $|i-j| \leq 1$.
The sets $V_i$ are called \emph{layers}. 
For a partition $\mathcal{P}$ of $V(G)$ and a layering $\mathcal{L}$ of $G$,
the \emph{width} of $(\mathcal{P},\mathcal{L})$ is the maximum of $|P \cap V_i|$ where
$P$ ranges over parts of $\mathcal{P}$ and $V_i$ ranges over the layers of $\mathcal{L}$.
For a partition $\mathcal{P}$ of $V(G)$, the \emph{quotient graph} $G/\mathcal{P}$
has vertex set $\mathcal{P}$ and an edge between parts $P$ and $Q$ if and only if
there exist $u\in P$ and $v \in Q$ with $uv \in E(G)$. 

The main result of the product structure theory for graphs excluding an apex graph as a minor can
be phrased as follows. This is essentially \cite[Lemma~24]{DBLP:journals/jacm/DujmovicJMMUW20},
stated as an algorithm.
\begin{theorem}[Lemma~24 of~\cite{DBLP:journals/jacm/DujmovicJMMUW20}, algorithmic version]\label{thm:product}
  For every constant $p$ there exists a constant $c$ and a polynomial-time algorithm that,
  given a connected $p$-nearly-embeddable without apices graph $G$, computes a connected partition
  $\mathcal{P}$ of $G$, a layering $\mathcal{L}$ of $G$ such that $(\mathcal{P},\mathcal{L})$
  has width at most $c$, and a tree decomposition of $G/\mathcal{P}$ of width less than $c$. 
\end{theorem}

To extract balanced separators from a tree decomposition, we need the following standard lemma.
\begin{lemma}\label{lem:td-to-balanced-sep}
  Let $G$ be a graph, $(T,\beta)$ its tree decomposition, and $\mu : V(G) \to \mathbb{R}_{\geq 0}$ 
  a weight function with support contained in $U \subseteq V(G)$. 
  \begin{itemize}
  \item Given $G$, $(T,\beta)$, and $\mu$, one can in polynomial time compute $t \in V(T)$ such that
  the bag of $t$ is a balanced separator in the following sense:
  every connected component $S$ of $T-\{t\}$ satisfies $\mu(\bigcup_{r \in S} \beta(r) \setminus \beta(t)) \leq \mu(V(G))/2$.
  \item Given $G$, $(T,\beta)$, and $U$, 
  one can in polynomial time compute a set $Y$ of at most $2|U|-1$ nodes
  of $T$ such that for every weight function $\mu : V(G) \to \mathbb{R}_{\geq 0}$ with
  support contained in $U$,
  there exists $t \in Y$ that is a balanced separator as in the previous bullet.
  \end{itemize}
\end{lemma}
\begin{proof}
Root $T$ in an arbitrary node. 

For every $u \in U$, let $t_u$ be the topmost node of $T$ whose bag contains $u$.
Let $X = \{t_u~|~u \in U\}$ and let $Y$ consist of $X$ and lowest common ancestors of any pair of nodes in $X$. 
Clearly, $|X| \leq |U|$ and $|Y| \leq 2|X|-1 \leq 2|U|-1$. Note that the definition of $Y$ depends on $U$
and not on the precise weight function $\mu$.

We say that a node $t \in V(T)$ is \emph{good} if for every connected component $S$ of $T-\{t\}$, 
we have $\mu(\bigcup_{r \in S} \beta(r) \setminus \beta(t)) \leq \mu(V(G))/2$. 

For every edge $st \in E(T)$ with $t$ being the parent of $s$, look at the two connected components $T_s$ and $T_t$ 
of $T-\{st\}$, compare $\mu(\bigcup_{r \in V(T_s))} \beta(r))$ with $\mu(\bigcup_{r \in V(T_t)} \beta(r))$
and orient the edge $st$ from smaller weight to the larger weight; in case of equal weights, orient the edge
upwards (from $s$ to $t$). Let $t_0$ be a node of out-degree $0$. As every edge incident with $t_0$ is oriented
towards $t_0$, $t_0$ is good. This proves that there exists a good node (and we can return such a node in the first bullet).

For the second bullet, it suffices to prove the following: if $t$ is good but $t \notin Y$, then
$t$ is not the root of $T$ and the parent of $t$ is good, too. This proves that the topmost good node
is in $Y$ and we are done with the second bullet. 

To this end, assume that $t$ is good but $t \notin Y$. This implies that for at least one component $S$
of $T-\{t\}$ that is a child subtree of $t$, we have 
$U \cap (\bigcup_{r \in S} \beta(r) \setminus \beta(t)) \neq \emptyset$.
However, as $\mu(\bigcup_{r \in S} \beta(r) \setminus \beta(t)) \leq \mu(V(G))/2$, any $t \notin X$, 
there must be a component $S'$ of $T-\{t\}$ that is not a child subtree of $t$; in particular, $t$ is not the root of $T$.
Let $s$ be the parent of $t$. 
As $t \notin X$, we have $\beta(s) \cap U \supseteq  \beta(t) \cap U$. This implies that $s$ is good, too:
every connected component of $T-\{s\}$ either is contained in $S'$ or its bag contain the same vertices of $U$
as $S$. 
\end{proof}

\begin{corollary}\label{cor:XF}
    Fix a constant $p$ and let $c$ be the constant for $p$ coming from Theorem~\ref{thm:product}.
    Then, there exists
    a polynomial-time algorithm that, given a $p$-nearly-embeddable without apices graph $G$
of diameter $\Delta$,
computes a tree decomposition $(T,\beta)$ of $G$ and, for every $t \in V(T)$,
 a spanning forest $F_t$ of $G[\beta(t)]$, such that the following holds for every $t \in V(T)$:
\begin{enumerate}
\item $|\beta(t)| \leq c^2 \cdot (\Delta+1)$, 
\item $F_t$ has at most $c$ connected components,
\item $F_t$ has maximum degree at most $3c^2-1$, and
\item for every integer $d \geq 0$ and every $v \in V(G)$, there are at most $(2d+1)c^2$ vertices
of $\beta(t)$ within distance at most $d$ from $v$.
\end{enumerate}
In particular, we have the following routines.
\begin{description}
\item[(Explicit weight function usage.)]
Given additionally a weight function $\mu:V(G) \to \mathbb{R}_{\geq 0}$, 
we can compute in polynomial time 
a separation $(A,B)$ of $G$ such that $|A \cap B| \leq c^2 (\Delta+1)$, 
$\mu(B \setminus A), \mu(A \setminus B) \leq \frac{2}{3} \mu(V(G))$, 
and a spanning forest $F$ of $G[A \cap B]$ satisfying properties above 
with $A \cap B$ playing the role of $\beta(t)$;
\item[(Explicit support usage.)]
Given additionally a set $U \subseteq V(G)$,
we can compute in polynomial time a family $\mathcal{X}$ of size at most $2|U|-1$
of tuples $((A,B),F)$ as in the previous bullet such that for every weight function $\mu:V(G) \to \mathbb{R}_{\geq 0}$
with support contained in $U$, at least one element $((A,B),F) \in \mathcal{X}$ 
satisfies the conditions for $\mu$ as in the previous bullet.
\end{description}
\end{corollary}

\begin{proof}
  Apply Theorem~\ref{thm:product} for the constant $p$ to the graph $G$, obtaining
  $\mathcal{P}$, $\mathcal{L}$, and a tree decomposition $(T,\beta')$ of $G' \coloneqq G/\mathcal{P}$.
  For every $t \in V(T)$, let $\beta(t) \coloneqq \bigcup_{P \in \beta'(t)} P$. 
  Clearly, $(T,\beta)$ is a tree decomposition of $G$. 

  Fix $t \in V(T)$.
  As $\beta(t)$ consists of $|\beta'(t)| \leq c$ parts of $\mathcal{P}$, and each part $P \in \mathcal{P}$
  contains
  at most $c$ vertices of each layer of $\mathcal{L}$, 
  $\beta(t)$ contains at most $c^2$ vertices of a single layer of $\mathcal{L}$.
  Since $G$ has diameter at most $\Delta$, at most $\Delta+1$ layers of $\mathcal{L}$ are nonempty. 
  This implies $|\beta(t)| \leq c^2 (\Delta+1)$. 
  
  Since $\beta(t)$ is the union of a collection of at most $c$ connected parts,
  $G[\beta(t)]$ has at most $c$ connected components.
  Furthermore, since every vertex of $G$ has neighbors only in at most three layers of $\mathcal{L}$,
  the maximum degree in $G[\beta(t)]$ is bounded by $3 \cdot |\beta'(t)| \cdot c - 1 \leq 3c^2-1$. 
  Hence, any spanning forest $F$ of $G[\beta(t)]$ has at most $c$ connected components and maximum degree
  at most $3c^2-1$. 

  To see the last property, observe that if a vertex $v \in V(G)$ lies in the $i$-th layer of $\mathcal{L}$
  and a vertex $u \in \beta(t)$ is within distance at most $d$ from $v$, then $u$ lies between layers $(i-d)$ and $(i+d)$. 
  Since $\beta(t)$ contains at most $c^2$ vertices of each layer of $\mathcal{L}$, there are at most $(2d+1)c^2$ vertices of $\beta(t)$
  within distance at most $d$ from $v$, as desired. 

  As for the claims in the two bullets, first assume without loss
  of generality that $T$ has maximum degree at most $3$ (it can be easily obtained by simple modifications
  that duplicate some nodes). 
  Then, use the appropriate bullet of Lemma~\ref{lem:td-to-balanced-sep}
  on the tree decomposition $(T,\beta)$ and $\mu$ or $U$.
  Finally, observe that for a node $t$ that is a balanced separator for $\mu$ in the sense of Lemma~\ref{lem:td-to-balanced-sep}, there is a 
  partition $(\mathcal{A},\mathcal{B})$ of the connected components of $T-\{t\}$
  such that the separation $(A\coloneqq \beta(t) \cup \bigcup_{r \in \mathcal{A}} \beta(r), B\coloneqq \beta(t) \cup \bigcup_{r \in \mathcal{B}}\beta(r))$
  satisfies the requirements with the forest $F_t$ as $A \cap B = \beta(t)$. 
  As $T$ is of maximum degree at most $3$, there is only a constant number of partitions $(\mathcal{A},\mathcal{B})$
  to check (in case $\mu$ is given) or enumerate (in case $U$ is given). 
\end{proof}

\subsection{Clusters and cluster families}

In the proof of Theorem~\ref{thm:main-d}, it will be convenient to switch at some point
the narrative from balls to just subgraphs of bounded (strong) diameter. 
Here we introduce some notation and basic facts. 

Fix an integer $d$ and let $G$ be a graph.
A \emph{$d$-cluster} in $G$ is a nonempty set $A \subseteq V(G)$ such that $G[A]$ is connected and has diameter
at most $d$. Note that $0$-clusters are singleton vertex~sets.
We observe the following.
\begin{lemma}\label{lem:clusters-disjoint}
  Let $d \geq 0$ be an integer, $G$ be a graph, and $\mathcal{Z}$ a family of $d$-clusters in $G$.
  Then, there exists a family $\mathcal{Z}'$ of pairwise disjoint $(2d)$-clusters in $G$
  such that $|\mathcal{Z}'| \leq |\mathcal{Z}|$ and $\bigcup \mathcal{Z} = \bigcup \mathcal{Z}'$. 
\end{lemma}
\begin{proof}
  Let $U = \bigcup \mathcal{Z}$.
  Let $A \subseteq U$ be an inclusion-wide maximal set of vertices of $U$ pairwise at  distance more than $d$
  in the graph $G[U]$. Since every $K \in \mathcal{Z}$ can contain at most one vertex of~$A$, we have
  $|A| \leq |\mathcal{Z}|$.
  Construct a partition $(K_u)_{u \in A}$ of $U$ as follows: every $v \in U$ assign to $K_u$
  where $u$ is the closest to $v$ vertex of $A$ in the graph $G[U]$, breaking ties arbitrarily
  (i.e., $K_u$s are Voronoi cells of the centers $A$ in the graph $G[U]$). 

  By construction, for every $u \in A$ and $v \in K_u$, a shortest path in $G[U]$ from $v$ to $u$ is contained
  in $K_u$. This implies that $G[K_u]$ is connected and, by maximality of $A$, this path is of length at most 
  $d$. Hence, $G[K_u]$ has diameter at most $2d$. 
  We infer that $\mathcal{Z}' = \{K_u~|~u \in A\}$ satisfies the desired properties.
\end{proof}
Lemma~\ref{lem:clusters-disjoint} allows us to switch from a family of radius-$d$ balls around vertices
in $Z$ to a family of at most $k$ pairwise disjoint $(4d)$-clusters.

\section{Flow/cut duality}\label{sec:duality}

In this section we will develop a variant
of Theorem~\ref{thm:dual1} that handles distance-$d$ neighborhoods of patterns, but works only in $K_{3,t}$-minor-free graphs. 
For the proof, we will need to look under the hood of the proof
of Theorem~\ref{thm:dual1} to derive some additional properties. The following statement
can be extracted from the proof in~\cite{DBLP:journals/siamcomp/FominLMPPS22}
and, for completeness, we reproduce the proof in Appendix~\ref{app:dual0}

\begin{definition}[$(p,q)$-structure]
  Let $G$ be an undirected graph, $s,t \in V(G)$ be two
  distinct vertices, and $p,q \geq 0$ be integers.
  A \emph{$(p,q)$-structure} in $(G,s,t)$ consists of
  a chain $(C_1,\ldots,C_p)$ of $(s,t)$-separators
  and a sequence $(P_1,\ldots,P_q)$ of $(s,t)$-paths
  such that the following holds:
  \begin{enumerate}
  \item For every $1 \leq i \leq q$ and $1 \leq j \leq p$, we have $|V(P_i) \cap C_j| = 1$.
  Note that this implies that for a path $P_i$, the unique vertices
  lying in the intersection of $V(P_i)$ and separators $C_1,C_2,\ldots,C_p$
  lie exactly in this order if we traverse $P_i$ from $s$ to $t$. 
  \item For every $1 \leq i \neq i' \leq q$, every vertex
  of $V(P_i) \cap V(P_{i'})\setminus \{s,t\}$ lies on one of the separators~$C_j$.
  Note that this implies that $|V(P_i) \cap \bigcup_{i' \in [q] \setminus \{i\}} V(P_{i'})| \leq p$.
  \item For every $1 \leq j \leq p$, every vertex of $C_j$ lies on some path $P_i$. 
  Note that this implies $|C_j| \leq q$ for every $1 \leq j \leq p$. 
  \end{enumerate}
\end{definition}

\begin{theorem}\label{thm:dual0}
  There exists a polynomial-time algorithm that,
  given a graph $G$, distinct vertices $s,t \in V(G)$,
  and an integer $q \geq 0$, outputs an integer $p \geq 0$
  and a $(p,q)$-structure in $(G,s,t)$. 
\end{theorem}

Note that Theorem~\ref{thm:dual0} can output $p=0$;
in this case, the chain is empty,
but the second property implies that 
the paths $P_i$ are vertex-disjoint, except for the endpoints.

\subsection{Distance version}

The proof of Theorem~\ref{thm:dualdistance} starts with first invoking Theorem~\ref{thm:dual0} to find a $(p_0,q_0)$-structure with $q=q_0$. Then we consider two cases. If $p_0$ is small, then we modify the paths in the $(p_0,q_0)$-structure to obtain the required set of paths.
If $p_0$ is large, then we modify the separators in the $(p_0,q_0)$-structure to obtain the required set of separators. Lemmas~\ref{lem:almostdisjoint-distance} and \ref{lem:sepchain-distance}, proposed in this section, handle these two cases, respectively.

In both cases, we need the following folklore observation that a directed graph with small maximum outdegree contains a large independent set.

\begin{lemma}\label{lem:outdeg}
If directed graph $H$ has outdegree at most $d$, then it contains an independent set of size at least $\ceil{|V(H)|/(2d+1)}$.
\end{lemma}
\begin{proof}
  Let $H'$ be the underlying undirected graph of $H$. 
  Let us try to find an independent set of $H'$ by greedily selecting a vertex $v$ with lowest degree, putting $V$ into the independent set, and removing its neighborhood. 
  Every induced subgraph of $H$ has outdegree at most $d$, thus every induced subgraph of $H'$ has average degree at most $2d$. This means that the greedy process can always find a vertex of degree at most $2d$ and hence discards at most $2d$ additional vertices when adding a new vertex $v$ into the independent set. Thus the algorithm finds an independent set of size $\ceil{|V(H)/(2d+1)}$.
\end{proof}

\paragraph{Handling almost-disjoint paths.}

We first handle the case when we have a $(p,q)$-structure with small $p$, that is, a large set of almost-disjoint paths that have only $p$ public vertices each. The main argument is showing that if the $d$-neighborhood of a vertex of $Z$ can influence too many vertices on some path of the $(p,q)$-structure, then we can locate a $K_{h,3}$-minor in the graph. However, this is true only after modifying the paths (using shortcuts): there can be lots of such vertices on a path if they are close to a public vertex of the path. 

Before presenting the proof, we establish a basic combinatorial result. First, we need the Sunflower Lemma of Erd\H os and Rado.

\begin{lemma}[Sunflower Lemma \cite{erdos1960intersection}]\label{lem:sunflower}
Let $\mathcal{F}$ be a family of subsets of size $r$ of some universe. If  $|\mathcal{F}| > r! \cdot (k - 1)^r$, then  there exist $k$ distinct sets $S_1, S_2, \dots, S_k \in \mathcal{F}$ and a (possible empty) subset $C$ of the universe such that for all $1 \le i < j \le k$, we have $S_i \cap S_j = C$.
\end{lemma}

We can use the Sunflower Lemma to bound the size of cross-intersecting pairs. 
\begin{lemma}\label{lem:A-B-disjoint}
Let $A_1$, $\dots$, $A_n$, $B_1$, $\dots$, $B_n$ be subsets of a universe such that $|A_i|\le a$, $|B_i|\le b$, $A_i\cap B_i=\emptyset$ for every $i\in [n]$, and $B_i\neq B_j$ for $i,j\in [n]$, $i\neq j$. If $n\ge b!(t(1+2a)-1)^b$, then there is a subset $I\subseteq [n]$ of size $t$ such that $A_i\cap B_j=\emptyset$ for every $i,j\in I$.
\end{lemma}
\begin{proof}
Let us apply Lemma~\ref{lem:sunflower} on the distinct sets $B_i$ and $k=t(1+2a)$; let $I_1\subseteq [n]$ correspond to the resulting $k$ sets $B_i$, $i\in I_1$. Let us observe that, for every $i\in I_1$, there are most $a$ values $j\in I_1$ such that $A_i\cap B_j\neq\emptyset$: a vertex of the common intersection $C=\bigcap_{j\in I_1}B_j$ cannot be in $A_i$, and every other vertex is in at most one $B_j$ for $j\in I_1$. Let us construct a directed graph $H$ on $I_1$ where there is an edge $\overrightarrow{ij}$ if $A_i\cap B_j\neq \emptyset$. As every vertex of $H$ has outdegree at most $a$, Lemma~\ref{lem:outdeg} implies that there is an independent set $I_2$ of size at least $|I_1|/(1+2a)=t$.
\end{proof}

The following lemma takes a $(p,q)$-structure with large $q$ (compared to $p$ and $k$), and modifies it in a way such that, for some unknown set $Z$ of size at most $k$, one of the paths has small intersection not only with $Z$, but also with its $d$-neighborhood.  
\begin{lemma}\label{lem:almostdisjoint-distance}
There is a polynomial-time algorithm that given a connected $K_{3,h}$-minor-free graph $G$, positive integers $p,q,r,d,\lambda, k_0, k_1$ satisfying $r\ge p$ and $q\ge (p+k_0/r+k_1)(10d^2h)^{d+3}\lambda$, a pair $s,t \in V(G)$ of different vertices, a set $Z_0\subseteq V(G)$ of size at most $k_0$, a set $X\subseteq V(G)\setminus \{s,t\}$ such that $G[X]$ has at most $\lambda$ components, and a $(p,q)$-structure of $(G-X,s,t)$,
outputs, for some $q'\le q$,  a sequence $(P'_1,\ldots,P'_{q'})$ of $(s,t)$-paths in $G-X$ and a set $Q_i\subseteq V(P'_i)$ for every $i\in [q']$ satisfying the following conditions:
\begin{itemize}
  \item $|Q_i|\le 10dr$ for every $i\in [q']$.
\item For every $Z_1\subseteq V(G)$ of size at most $k_1$, there is an $i\in [q']$ such that $V(P'_i)\setminus Q_i$ is disjoint from $N^d_G[Z_0\cup Z_1]$.
\end{itemize}
\end{lemma}
\begin{proof}
  Let $(C_1,\ldots,C_p),(P_1,\ldots,P_q)$ be the given $(p,q)$-structure of $(G,s,t)$.
  Let $\Pub_i$ be the public vertices of $P_i$, i.e., those vertices of $P_i\setminus \{s,t\}$ that are on some other path $P_j$, $i\neq j$, as well. Let $\Priv_i=V(P_i)\setminus \Pub_i\setminus \{s,t\}$. We define the layering numbers for the vertices in $G-X$ in a natural way, as follows. For $1\le i \le p$, every vertex of the separator $C_i$ in the $(p,q)$-structure is in layer $i$, as well as the vertices between separators $C_i$ and $C_{i+1}$. Vertex $s$ and all vertices between $s$ and $C_1$ are in layer 0, while vertex $t$ and all vertices between $C_p$ and $t$ are in layer $p$. Observe that the layer numbers of adjacent vertices can differ only by at most 1 and every path $P_i$ has exactly one public vertex in each layer $\ell\in [p]$.

  The following construction will be useful for handling the set $X$.
   Let $G'$ be the graph obtained from $G$ by contracting each vertex of $G[X]$ into a single vertex; let $X'$ be the set of these contracted vertices. Note that if $v\not\in X'$, then $G'-(X'\setminus \{v\})$ is simply $G'-X'$, and if $v\in X'$, then $v$ is the only vertex of $G'-(X'\setminus \{v\})$ outside of $G'-X'$.

For any vertex $v$ in $G-X$, we want to bound the number paths $P_i$ such that a private vertex of $P_i$ can be reached from $v$ by path of length at most $d$ in $G-X$ that does not go through the public vertices of $P_i$. Similarly, we want to bound the number of such paths $P_i$ where a private vertex can be reached from $X$ by a path of length at most $d$ in $G$ that does not go through the public vertices of $P_i$. The following claim presents a unified bound for both settings, under the assumption that the reached private vertices are on the same layer.
   
  \begin{claim}\label{cl:handlepaths}
    Let $v$ be a vertex of $G'\setminus\{s,t\}$ and $I\subseteq [q]$ be a subset of indices. For every $i\in I$, let $u_i\in \Priv_i$ be a vertex and $W_i$ be a $v-u_i$ path of length at most $d$ in $G'-(X\setminus \{v\})$. Assume furthermore that\medskip

    \begin{itemize}[nosep]
    \item for every $i\in I$, vertex $u_i$ is in the same layer $\ell$, and
    \item for every $i\in I$, we have $P_i\cap W_i=\{u_i\}$.
      \end{itemize}
\medskip
      
     \noindent Then $|I|< (10d^2h)^{d+2}$.
     \end{claim}
     \begin{claimproof}
       Assume that $|I|\ge (10d^2h)^{d+2}$. Let us construct a directed graph $H$ on $I$ where there is an edge $\overrightarrow{i_1i_2}$ for $i_1,i_2\in I$ if $W_{i_1}\cap \Priv_{i_2}\neq\emptyset$. As $W_{i_1}$ has at most $d+1$ vertices and the sets $\Priv_{i_2}$ are disjoint for different values of $i_2$, we have that the every vertex of $H$ has outdegree at most $d+1$. Thus by Lemma~\ref{lem:outdeg}, in $H$ there is an independent set $I'$ of size at least $|I|/(2(d+1)+1)\ge (10d^2h)^{d+1}$. That $I'$ is an independent set in particular implies that $v\not\in \Priv_i$ for every $i\in I'$, as this would imply that $W_j$ intersects $\Priv_i$ for any $j\in I'$, $j\neq i$ (here we use that $|I'|\ge 2$). Therefore, we can assume $v\neq u_i$ for every $i\in I'$. Moreover, $s,t\not\in W_i$, as otherwise this vertex would be a vertex of $P_i\cap W_i$ different from $u_i$.
       
Consider some $i,j\in I'$, $i\neq j$ such that $P_i$ and $W_j$ intersect.
 Note that $W_j$ can have only a single vertex not in $G'-X'$: its endpoint $v$.
 As $u_j$ is in layer $\ell$, this means that every vertex of $W_j\setminus X'$ is between layers $\ell-d$ and $\ell+d$. Thus if $P_i$ and $W_j$ intersect in a vertex $w$, then by the definition of $I'$, we have $w\not \in \Priv_i$ and hence $w$ has to be one of the at most $2d+1$ public vertices of $P_i$ between layers $\ell-d$ and $\ell+d$. Let $A_i$ be this subset of $\Pub_i$ and let $B_i=W_i$; observe that $A_i\cap B_i=\emptyset$, as $W_i$ is disjoint from $\Pub_i$. Let us apply  Lemma~\ref{lem:A-B-disjoint} with $a=2d+1$, $b=d+1$, and $k=h$: note that
$|I'|\ge (10d^2h)^{d+1}\ge (d+1)!(h(2d+2)-1)^{d+1}=b!(h(1+2a)-1)^b$ holds. Therefore, there is a subset $I''\subseteq I'$ of size $h$ such that $A_i\cap B_j=\emptyset$ for all $i,j\in I''$.

  Using that the paths $P_i$ and $W_j$ do not intersect each other for any $i,j\in I''$, $i\neq j$ and that every $u_i$ is in the same layer, performing the following contractions for every $i\in I'$ creates a $K_{3,h}$-subgraph with $\{s,t,v\}$ on one side and $\{u_i \mid i\in I''\}$ on the other side:
\medskip

\begin{itemize}[nosep]
  \item Contract $W_i\setminus\{u_i\}$ into the vertex $v$.
  \item Contract the subpath of $P_i$ from $s$ up to (but not including) $u_i$ into the vertex $s$.
  \item Contract the subpath of $P_i$ from the vertex immediately after $u_i$ to $t$ into the vertex $t$.
\end{itemize}   \medskip

We can verify that vertices contracted to different vertices ($s$, $t$, or $v$) are disjoint. For this, we observe the following.
\begin{itemize}
\item For any $i,j\in I'$, path $W_i\setminus \{u_i\}$ is disjoint from every $P_j$.
\item Vertices of $P_i$ contracted to $s$ can coincide with vertices of $P_j$ contracted to $t$ only in public vertices of $P_i$ and $P_j$. However, for any $i,j\in I'$, the public vertices on the subpath  of $P_i$ from $s$ up to $u_i$ have layer number at most $\ell$, while the public vertices on the subpath of $P_j$ from the vertex immediately after $u_j$ to $t$ have layer number greater than $\ell$.
\end{itemize}
Therefore, this sequence of contractions create a $K_{3,h}$-minor in $G'-(X\setminus \{v\})$. This contradicts the assumption that $G$ (and hence its minor $G'-(X\setminus \{v\})$) is $K_{3,h}$-minor free.
\end{claimproof}  
  The following two claims use a simple Pigeonhole Principle argument to remove the layer requirement from Claim~\ref{cl:handlepaths}. 
  \begin{claim}\label{cl:distance-d-inside}
For every vertex $v\in V(G)\setminus X$, there are at most $(10d^2h)^{d+3}$ paths $P_i$ with $\dist_{G-(X\cup \Pub_i\cup \{s,t\})}(\Priv_i,v)\le d$.
\end{claim}
\begin{claimproof}
  Suppose that there are more than $(10d^2h)^{d+3}$ such paths; let $I$ contain the indices of these paths. For every $i\in I$, let us select a shortest path $W_i$ from $v$ to $\Priv_i$ in $G-(X\cup \Pub_i\cup \{s,t\})$; let $u_i$ be its endpoint in $\Priv_i$. Clearly, $W_i$ intersects $\Priv_i$ only in $u_i$.

  As every $u_i$ is at distance at most $d$ from $v$ in $G-X$, the layer of $u_i$ differs from the layer of $v$ by at most $d$. By the Pigeonhole Principle, there is a subset $I'\subseteq I$ of size more than  $|I|/(2d+1)\ge (10d^2h)^{d+2}$ such that $u_i$ is in the same layer for every $i\in I'$. However, Claim~\ref{cl:handlepaths} shows that $|I'|< (10d^2h)^{d+2}$, a contradiction.
\end{claimproof}  

\begin{claim}\label{cl:distance-d-outside}
There are at most $p\cdot (10d^2h)^{d+2}\lambda$ paths $P_i$ such that $\dist_{G-\Pub_i\cup\{s,t\}}(\Priv_i,X)\le d$.
\end{claim}
\begin{claimproof}
  Assume that there are more than $p\cdot (10d^2h)^{d+3}\lambda$ such paths;  let $I$ contain the indices of these paths.
Observe that $\dist_{G-\Pub_i\cup\{s,t\}}(\Priv_i,X)\le d$ can be also written as $\dist_{G'-\Pub_i\cup\{s,t\}}(\Priv_i,X')\le d$.
  For every $i\in I$, let us select a shortest path $W_i$ from $X$ to $\Priv_i$ in $G-(X\cup \Pub_i\cup \{s,t\})$; let $u_i$ be its endpoint in $\Priv_i$. Clearly, $W_i$ intersects $\Priv_i$ only in $u_i$.
  By the Pigeonhole Principle, there is a subset $I'\subseteq I$ of size more than $|I|/p\ge (10d^2h)^{d+2}\lambda$ such that every $u_i$ is in the same layer. Again by the Pigeonhole Principle, there is a subset $I''\subseteq I'$ of size more than $|I'|/\lambda\ge (10d^2h)^{d+1}$ such that $W_i$ starts at the same vertex $x'\in X'$  for every $i\in I'$.  However, Claim~\ref{cl:handlepaths}, shows that $|I'|<(10d^2h)^{d+2}$, a contradiction.
 \end{claimproof}

For $i\in [q]$, let $Z^i_0\subseteq Z_0\setminus X$ be the set comprised of all those vertices $z\in Z_0\setminus X$ for which $\dist_{G-(X\cup \Pub_i\cup \{s,t\})}(z,\Priv_i)\le d$ (as $Z_0$ is given in the input, we can compute this set). Let us consider the graph $G_i$ induced by the vertex set
\[
\Priv_i\cup\bigcup_{v\in \Pub_i\cup Z^i_0\cup\{s,t\}}N^{d}_{G-X}[v].
\]
Let $P'_i$ be a shortest $s-t$ path in $G_i$.
\begin{claim}\label{cl:newvertices}
$|V(P'_i)\cap N^{d}_{G-X}[v]|\le 2d+1$ for every $v\in \Pub_i\cup Z^i_0\cup \{s,t\}$.
\end{claim}
\begin{claimproof}
  For some $v\in \Pub_i\cup Z^i_0\cup \{s,t\}$, let $u_1$ (resp., $u_2$) be the first (resp., last) vertex of $P'_i$ in $N^{d}_{G-X}[v]$. If $P'_i$ has at least $2d+2$ vertices in $N^{d}_{G-X}[v]$, then the $u_1-u_2$ subpath of $P'_i$ contains all such vertices and hence has length at least $2d+1$. However, in graph $G_i$, there is a $u_1-u_2$ path of length at most $2d$ via $v$. This contradicts the minimality of the path $P'_i$.
\end{claimproof}
Let us define $Q_i=V(P'_i)\cap \bigcup_{v\in \Pub_i\cup Z^i_0\cup\{s,t\}}N^{d}_{G-X}[v]$.
From Claim~\ref{cl:newvertices}, it is clear that $|Q_i|\le (2d+1)(p+|Z^i_0|+2)$. Moreover, we the following claim shows that $|Z^i_0|$ (and hence $Q_i$) cannot be too large for too many of the paths $P'_i$.

\begin{claim}\label{claim:Qbound}
There are at most $(10d^2h)^{d+3} k_0/r$ values of $i\in [q]$ such that $|Q_i|\ge 10dr$.
\end{claim}
\begin{claimproof}
By Claim~\ref{cl:distance-d-inside}, each vertex $Z\in Z_0\setminus X$ can be in $Z^i_0$ for at most $(10d^2h)^{d+3}$ values of $i$. Therefore, by an averaging argument, there are at most $(10d^2h)^{d+3} k_0/r$ values of $i$ such that $|Z^i_0|> r$. Suppose that $|Z^i_0|\le r$ for some $i\in [q]$. For such an $i$, as $|\Pub_i\cup Z^i_0\cup \{s,t\}|\le p+r+2$, Claim~\ref{cl:newvertices} implies $Q_i=|V(P'_i)\cap \bigcup_{v\in \Pub_i\cup Z^i_0\cup\{s,t\}}N^{d}_{G-X}[v]|\le (p+r+2)(2d+1)\le 10dr$. 
\end{claimproof}  

Suppose now that $Z_1\subseteq V(G)$ is of size at most $k_1$ and let $v\in V(P'_i)\setminus Q_i\subseteq \Priv_i$ be a vertex in $N^d_G[Z_0\cup Z_1]$. Let $W$ be a $v-(Z_0\cup Z_1)$ path of length at most $d$ in $G$. Following $W$ from $v$, let $w$ be the first vertex of $W$ in $(\Pub_i\cup \{s,t\})\cup X\cup Z_0\cup Z_1$. 
    We consider four possible cases and show that each of them can happen only for a bounded number of different paths $P'_i$, hence there is an $i\in [q]$ such that $V(P'_i)\setminus Q_i$ is disjoint from $N^d_G[Z_0\cup Z_1]$.

\begin{enumerate}
  \item $w\in \Pub_i\cup \{s,t\}$. As no earlier vertex of $W$ was in $X$, we have that $v\in N^d_{G-X}[w]\subseteq Q_i$, a contradiction.
    \item $w\in Z_0\setminus X$. The case $w\in \Pub_i\cup \{s,t\}$ was handled above and no earlier vertex of $W$ was in $X$ or in $\Pub_i\cup \{s,t\}$, we have that  $\dist_{G-(X\cup \Pub_i\cup \{s,t\})}(v,X)\le d$. This implies $w\in Z^i_0\subseteq Q_i$, a contradiction.
  \item $w\in X$. As no earlier vertex of $W$ was in $X$ or in $\Pub_i\cup \{s,t\}$, we conclude that   $\dist_{G-(X\cup \Pub_i\cup \{s,t\})}(v,X)\le d$. By Claim~\ref{cl:distance-d-outside}, this can happen for at most $p(10d^2h)^{d+2}\lambda$ values of $i$.
  \item $w\in Z_1$. The case $w\in X\cup \Pub_i\cup \{s,t\}$ was handled above and no earlier vertex of $W$ was in $X\cup \Pub_i\cup \{s,t\}$, thus we have that $\dist_{G-(X\cup \Pub_i\cup \{s,t\})}(v,w)\le d$. For any $w\in Z_1$, Claim~\ref{cl:distance-d-inside} shows that this can happen only for at most $(10d^2h)^{d+3}$ different paths $P_i$, thus in total this case can happen for at most $(10d^2h)^{d+1}k_1$ different $i$.
  \end{enumerate}
  We say that a path $P'_i$ is {\em bad} if either $|Q_i|>10dr$ or $V(P'_i)\setminus Q_i$ is not disjoint from $N^d_G[Z_0\cup Z_1]$.
The number of bad paths can be bounded by using the bound in Claim~\ref{claim:Qbound} and the bounds in the last two cases above.  
    As $q$ is larger than the sum $(10d^2h)^{d+3} k_0/r+p(10d^2h)^{d+2}\lambda+(10d^2h)^{d+3}k_1$ of these three bounds, there is an $i$ such that $V(P'_i)\setminus Q_i$ is disjoint from $N^d_G[Z_0\cup Z_1]$ and $|Q_i|\le 10dr$. In other words, if we let $I\subseteq [q]$ contain $i$ if $|Q_i|\le 10dr$, then returning the $q'=|I|$ paths $P'_i$, $i\in I$ and the sets $Q_i$ satisfies the requirements of the lemma.
  \end{proof}

  \paragraph{Handling a sequence of separators.}
Next we show how a chain of separators can be modified so that the distance-$d$ neighborhood of each vertex of $Z$ can influence only a bounded number of vertices. The assumption that the graph is $K_{3,h}$-minor-free will be again used to limit the influence of every vertex on the separators in the chain.

\begin{lemma}\label{lem:sepchain-distance}
There is a polynomial-time algorithm that given a connected $K_{3,h}$-minor-free graph $G$,
positive integers $p,q,d,\lambda$ satisfying $q\ge 2$ and $p\ge 100d\log_2q$, 
a pair $s,t \in V(G)$ of different vertices, a set $X\subseteq V(G)\setminus \{s,t\}$
such that $G[X]$ has at most $\lambda$ components,
and a chain $(C_1,\ldots,C_p)$ of $(s,t)$-separators in $G-X$ of size at most $2q$,
outputs for some $p'\le p$ a sequence $(C'_1,\ldots, C'_{p'})$ of $(s,t)$-separators
in $G-X$ of size at most $2q$ so that the following holds: 
\begin{enumerate}
\item for every $j \in [p']$ and $v \in V(G)$, 
\[ \left|V(C'_j) \cap N^d_G[v]\right| \leq (10d^d h)^{2d+1}(1+ \lambda\log_2(4q));\textrm{and}\]
\item 
for every set $Z\subseteq V(G)$ of size at most $k$, there is an index $j\in [p']$ such that \[\left|V(C'_j)\cap N^d_G[Z]\right|\le      (k/p) \cdot \log_2 q  \cdot 50d(10d^dh)^{2d+1}\cdot \lambda.\]
\end{enumerate}
\end{lemma}  

\begin{proof}
  Let \[z\coloneqq 2\log_2(2q+1)+2d.\] Note that $z\le 10d\log_2q\le p/10$, because $d\ge 1$ and $q\ge 2$.
We define the separators $C'_1$, $C'_2$, $C'_3$, $\dots$ based on the separators $C_{z+d+1}$, $C_{2z+d+1}$, $C_{3z+d+1}$, $\dots$ the following way. 
Let \[J=[\floor{(p-2(d+1))/z}-2]\qquad \textrm{and}\qquad \iota(j)=jz+d+1.\] Observe that for $j\in J$, we have $d+1+z\le \iota(j) \le p-(z+d+1)$. Also, from the bounds on $p$, $q$, $d$, and $z$, we have that \begin{eqnarray*}|J| & \ge &  (p-2(d+1))/z-3\ge (96/100)(p/z)-3  \ge (96/100)(3/10)(p/z)\\
	& \ge & (96/100)(3/10)(p/(10d\log_2 q))>p/(50d\log_2 q).\end{eqnarray*}

We define the separator $C'_j$ based on $C_{\iota(j)}$ as follows.  Consider the weight function on $G-X$ defined as \[w_j(v)\coloneqq 2^{\dist_{G-X}(v,C_{\iota(j)})}.\] Let $C'_{j}$ be a minimum weight $(s,t)$-separator under this weight function.  Observe that $w(C_{\iota(j)})=|C_{\iota(j)}|=2q$, hence $w(C'_{j})\le 2q$ as well; in particular, $|C'_{j}|\le 2q$. We remark that, by definition, $s,t\not\in C'_j$, hence their weight is not relevant when  choosing $C'_j$. However, shortest paths may go through $\{s,t\}$, hence the vertices $s$ and $t$ can influence the weight of other vertices.

First we show that, the definition of $C'_j$ ensures that these sets are pairwise disjoint. In fact, no vertex of $G-X$ can be close to two such sets. 

  \begin{claim}\label{cl:sepdisjoint}
    For every vertex $v$ of $G-X$, there is at most one index $j\in J$ such that $\dist_{G-X}(v,C'_{j})\le d$. Furthermore, for $v\in \{s,t\}$, there is no such $j\in J$.
\end{claim}    
\begin{claimproof}
  Suppose that there are two such values $j_1\neq j_2$, and that $u_1\in C'_{j_1}$ and $u_2\in C'_{j_2}$ are at distance at most $d$ from $v$ in $G-X$.
Since $(C_1,\ldots,C_p)$ form a chain of $(s,t)$-separators in $G-X$, sets $C_{\iota(j_1)}$ and $C_{\iota(j_2)}$ are at distance at least $z$ from each other in $G-X$. Therefore, without loss of generality, we can assume that $\dist_{G-X}(v,C_{\iota(j_1)})\ge z/2$. This implies  $\dist_{G-X}(u_1,C_{\iota(j_1)})\ge z/2-d \ge \log_2(2q+1)$. But then $w(u_1)\ge 2q+1$, contradicting the fact that $w(C'_{j_1})\le 2q$.

For the second statement, suppose that $\dist_{G-X}(s,u)\le d$ for some $u\in C'_j$. As $\iota(j)\ge z+d+1$, we have $\dist_{G-X}(s,C_{\iota(j)})\ge z+d+1$, hence $\dist_{G-X}(u,C_{\iota(j)})\ge z+1$. Therefore, we have that $w_j(v)\ge 2^{z+1}>2q$, a contradicting the fact that $w(C'_j)\le 2q$.
\end{claimproof}  

 Let $G'$ be the graph obtained from $G$ by contracting each component of $G[X]$ into a single vertex; let $X'$ be the set of these contracted vertices.
If a vertex $u\in C'_j$ is in $N^d_G[Z]$, then either $v$ is at distance at most $d$ from some vertex of $Z\setminus X$ in $G-X$, or $u$ is at distance at most $d$ from one of the at most $\lambda$ components of $G[X]$. In other words, either $u$ is at distance at most $d$ from some $v\in Z\setminus X$ in $G-X$, or $u$ is at distance at most $d$ from some $x'\in X'$ in $G-(X\setminus \{x'\})$ (i.e., without using any other contracted vertex). The following claim will be used to give a unified argument bounding the number of both types of vertices in $C'_j$; however, the claim in this form only bounds the number of vertices with the same weight under the weight function~$w_j$. Note that if $v\not\in X'$, then $G'-(X'\setminus \{v\})$ is simply $G'-X'$, and if $v\in X'$, then $v$ is the only vertex of $G'-(X'\setminus \{v\})$ outside $G'-X'$.

\begin{claim}\label{cl:handlepathssep}
    Consider any $j\in J$. Let $v$ be a vertex in $V(G')$ and let $U\subseteq C'_j$ be a subset such that $w_j(u)$ is the same for every $u\in U$ and $\dist_{G'-(X'\setminus \{v\})}(v,u)\le d$ for every $u\in U$. Then we have $|U|\le (10d^dh)^{2d}$.
 \end{claim}
 \begin{claimproof}
   Suppose that $|U|>(10d^dh)^{2d}$. If $v\in U$, then let us remove this one vertex from $U$; in the following, we assume $v\not\in  U$ and $|U|\ge (10d^dh)^{2d}$.

   Consider a shortest path tree from $v$ to $U$ in $G'-(X'\setminus \{v\})$; this tree has depth at most~$d$. Thus, if every node had less than $10dh^22^d$ children, then the tree would have less than $d(10dh^22^d)^d<|U|$ nodes; contradiction. This means that we can assume that there is a node $v'$ with at least $10dh^22^d$ children. Let $U'\subseteq U$ be obtained by selecting, for each child $v''$ of $v'$, a vertex of $U$ that can be reached from $v'$ via $v''$ in the shortest path tree. For every $u'\in U$, let $W'_u$ be the $v'-u'$ path in the shortest path tree. Note that $v'$ is the only vertex shared by these paths. Because $v\not\in C'_j$, path $W'_u$ is of length at least 1.

  Let us observe two technicalities. First, we claim that only one vertex of $W'_u$ can be possibly in $X'$: its endpoint $v'$, and only if $v=v'$. Indeed, if $v\not\in X'$, then  $W'_u$ is fully contained in $G'-(X'\setminus \{v\})=G'-X'$. If $v\in X'$, then either $v'=v$ and hence $v\in X'$ is the endpoint of $W'_u$, or $v\neq v'$, but then $v$ is not contained in any of the paths $W'_u$. In particular, this means that $W'_u\setminus X'$ is connected. Second, $W'_u$ cannot contain $s$ or $t$: by Claim~\ref{cl:sepdisjoint}, these vertices are at distance more than $d$ from $C'_j$ in $G'-X'$ (and only the endpoint of $W'_u$ can be outside $G'-X'$). In particular, $v'$ is not in $\{s,t\}$.

  Let us consider the graph $G^*=G-(X\cup (C'_{j}\setminus U'))$. Note that $U'$ is an $(s,t)$-separator in~$G^*$.  Let us give infinite capacity to $s$ and $t$, and capacity $w_j(v)$ to every other vertex. If there is an $(s,t)$-separator $S$ of weight less than $w_j(U')$ in $G^*$, then $(C'_{j}\setminus U')\cup S$ is an $(s,t)$-separator in $G-X$ having weight strictly less than $w_j(C'_{j})$, contradicting the choice of $C'_{j}$. Thus by flow-cut duality, we can assume that in $G^*$ there is a collection $\mathcal{P}$ of $s-t$ paths such that  $|\mathcal{P}|=w_j(U')$ and every vertex $v$ of $G^*$ is contained in at most $w_j(v)$ of the paths of $\mathcal{P}$. As $U'$ is an $(s,t)$-separator of weight $w_j(U')$ in $G^*$, it also follows that every path in $\mathcal{P}$ contains exactly one vertex of $U'$ and every vertex $u\in U'$ is contained in exactly $w_j(u)$ such paths; let $\mathcal{P}_u\subseteq \mathcal{P}$ be the set of these paths. Note that $\{\mathcal{P}_u\mid u\in U'\}$ is a partition of $\mathcal{P}$. 

We build a directed graph $H$ on the vertex set $U'$ the following way. Let $m$ be the common value of $w_j(u)$ for every $u\in U'$. We introduce an edge $\overrightarrow{u_1u_2}$ if $W'_{u_1}$ intersects at least $m/(h-1)$ of the paths in $\mathcal{P}_{u_2}$. We have observed that $W'_{u_1}\setminus X'$ is connected, hence every vertex of  $W'_{u_1}\setminus X'$ is at distance at most $d$ from $u_1$ in $G'-X'=G-X$. This means that $w_j(v)\le 2^d\cdot m$ for every $v\in W'_{u_1}\setminus X'$. It follows that $W'_{u_1}$ can intersect at most $2^d\cdot dm$ paths. In particular, there are at most $2^d\cdot d(h-1)$ vertices $u_2\in U$ such that $W'_{u_1}$ intersects at least $m/(h-1)$ paths in $\mathcal{P}_{u_2}$. Thus the outdegree of $u_1$ is at most  $2^d\cdot d(h-1)$.

By Lemma~\ref{lem:outdeg}, the directed graph $H$ has an independent set $U''$ of size bounded as follows: \[|U''|\geq |U'|/(1+2^d\cdot 2d(h-1))\ge h.\]
For every $u\in U''$, each of the $t-1$ sets $W'_{u'}$ for $u'\in U''\setminus \{u\}$ intersects less than $m/(h-1)$ of the $m$ paths in $\mathcal{P}_u$, hence we can select a path $P_u\in\mathcal{P}_u$ that is disjoint from $W'_{u'}$ for every $u'\in U''\setminus \{u\}$.
Note that $v'\not\in P_u$ for any $u\in U''$: otherwise, $W'_{u'}\cap P_u$ would contain $v'$ for any $u'\in U''\setminus \{u\}$ (here we use $|U''|\ge 2$).

Intuitively, we would like to obtain a $K_{3,h}$-minor by the following sequence of contractions: (1) contracting the $s-u$ part of $P_u$ (without $u$) into $s$, (2) contacting the $u-t$ part of $P_u$ (without $u$) into $t$, and (3) contracting every $W'_u$ (without $u$) into $v'$. However, because $W'_u$ can intersect $P_u$, these contractions have to be performed in a careful way (see Figure~\ref{fig:contract}).
We claim that performing the following contractions in $G'-(X'\setminus v)$ for every $u\in U''$ creates a $K_{3,h}$-subgraph with $\{s,t,v\}$ on one side and the vertices of $U''$ on the other side:
 \begin{itemize}
 \item First, starting from $v'$, let vertex $u'$ be the first vertex of $W'_u$ that is in $P_u$ (note that $W'_u$ contains $u\in P_u$). Let us contract the $u'-u$ subpath of $W'_u$ into $u$. Note that $v'\not\in P_{u}$ implies $u'\neq v'$. Let us contract the subpath of $W'_u$ from $v'$ to the vertex preceding $u'$ into~$v'$. 
   \item After this contraction, $P_u-u$ consists of at least two components (possibly more if the contracted subpath of $W'_u$ intersects $P_u$ multiple times). Let us contract the component of the subgraph $P_u-u$ containing $s$ into vertex $s$, and the subpath of $P_u-u$ containing $t$ into vertex $t$.
   \end{itemize}
   Let us verify that the sets contracted to different vertices (i.e., to some $u\in U''$, $s$, $t$, or $v'$) are disjoint. The contractions to different vertices $u\in U''$ are disjoint because the paths $W'_u-v'$ are disjoint for different $u\in U''$. As $P_{u_1}$ and $W'_{u_2}$ are disjoint if $u_1\neq u_2$, no vertex of $W'_{u_2}$ is contracted into $s$ or $t$.  Since $U'$ is an $(s,t)$-separator in $G^*$ and each path $P_u$ intersects $U'$ in a single vertex, the $s-u_1$ subpath of $P_{u_1}$ cannot intersect the $t-u_2$ subpath of $P_{u_2}$ for any $u_1,u_2\in U''$. Thus indeed the contractions result in a $K_{3,h}$ subgraph in $G'-(X'\setminus v)$. As  $G'-(X'\setminus v)$ is a minor of $G$, this contradicts the assumption that $G$ is $K_{3,h}$-minor free.
   \begin{figure}
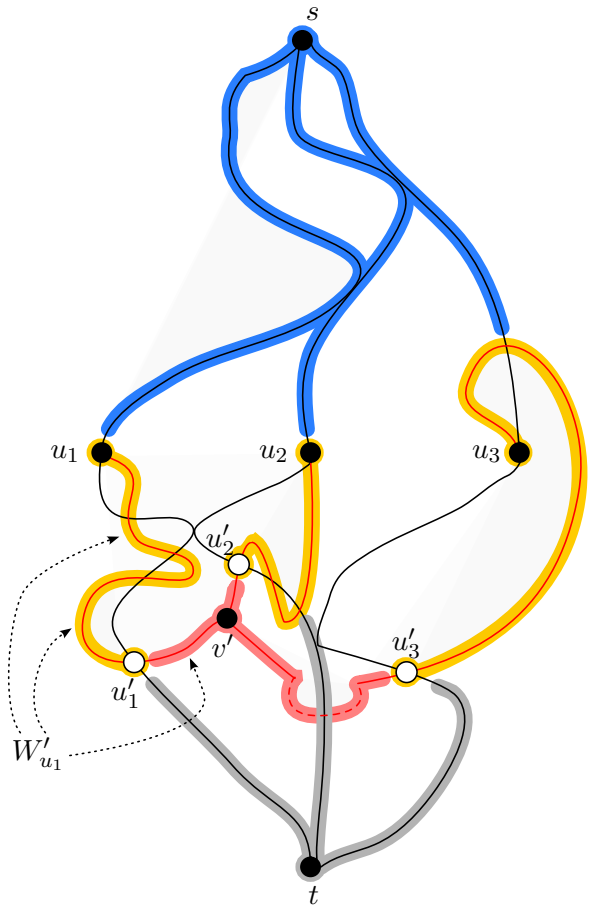

     \svgc{0.5\linewidth}{contract}
     \caption{Contractions to obtain a $K_{3,3}$-minor. The paths highlighted in yellow are contracted into $u_1$, $u_2$, and $u_3$. The pahs highlighted in blue, gray, and red are contracted into $s$, $t$, and $v'$, respectively.}
     \label{fig:contract}
   \end{figure}
 \end{claimproof}

Using Claim~\ref{cl:handlepathssep}, it is easy to bound the number of vertices contained in a separator $C'_j$ that can be close to some fixed vertex $v$ in $G-X$. 
\begin{claim}\label{cl:sepinside}
 For every vertex $v$ of $G-X$ and every $j\in J$, there are at most $(10d^dh)^{2d+1}$ vertices $u\in C'_{j}$ with $\dist_{G-X}(v,u)\le d$.
\end{claim}
\begin{claimproof}
  Let $U_1\subseteq C'_{j}$ be the set of these vertices and assume that there are more than $(10d^dh)^{2d+1}$ such vertices. As every vertex of $U_1$ is at distance at most $d$ from $v$ in $G-X$, they are at distance at most $2d$ from each other in $G-X$. This implies that
   there are at most
$2d+1$ different values of $w_j$ on $U_1$ and hence there is a subset $U_2\subseteq U_1$ of size larger than $|U_1|/(2d+1)\ge (10d^dh)^{2d}$ such that $w_j(u)$ is the same for every $u\in U_2$.
However, Claim~\ref{cl:handlepathssep} bounds the size of $U_2$ by $(10d^dh)^{2d}$, a contradiction.
\end{claimproof}

Similarly, Claim~\ref{cl:handlepathssep} allows us to bound the number of vertices of $C_j'$ close to $X$ in~$G$.

 \begin{claim}\label{cl:sepoutside}
 For every $j\in J$, there are at most $\lambda (10d^dh)^{2d} \log_2(4q)$ vertices $u\in C'_{j}$ satisfying $\dist_{G}(u,X)\le d$.
\end{claim}
\begin{claimproof}
  Let $U_1\subseteq C'_{j}$ be the set of these vertices and suppose that $|U_1|>\lambda (10d^dh)^{2d} \log_2(4q)$.
  Recall that $w_j(u)\le 2q$ for every $u\in C'_{j}$, hence there are at most
  $\log_2(4q)$ different values of $w_j(u)$ for $u\in U_1$. Therefore, there is a subset $U_2\subseteq U_1$ of size at least $|U_1|/\log_2(4q)\ge \lambda (10d^dh)^{2d}$  such that $w_j(u)$ is the same for every $u\in U_1$.

  For any $u\in C'_j$, the condition $\dist_{G}(u,X)\le d$ can be also written as $\dist_{G'}(u,X')\le d$: closeness to $X$ in $G$ means closeness to a contracted vertex of $X'$ in $G$.  As there are at most $\lambda$ vertices in $X'$, this means that there is an $x'\in X'$ and a subset $U_3\subseteq U_2$ of size more than $|U_2|/\lambda=(10d^dh)^{2d}$ such that $\dist_{G'-(X\setminus \{x'\})}(u,x')\le d$ for every $u\in U_3$. However, Lemma~\ref{cl:handlepathssep} bounds the size of $U_3$ by $(10d^dh)^{2d}$, a contradiction.
\end{claimproof}

Claims~\ref{cl:sepinside} and~\ref{cl:sepoutside} allow us to bound $|C_j' \cap N_G^d[v]$ for any
$v \in V(G)$. 
\begin{claim}\label{cl:septotal}
  For every $j \in J$ and $v \in V(G)$ there are at most $(10d^d h)^{2d+1}(1+ \lambda\log_2(4q))$ vertices
  $u \in C'_j$ satisfying $\dist_G(v,u) \leq d$. 
\end{claim}
\begin{claimproof}
  Let $u \in C'_j$ be such a vertex and let $Q$ be a shortest path from $u$ to $v$.
  If $V(Q) \cap X = \emptyset$, then $\dist_{G-X}(u,v) \leq d$ and
  there are at most $(10d^d h)^{2d+1}$ such vertices $u$ by Claim~\ref{cl:sepinside}.
  Otherwise, $\dist_G(u,X) \leq d$ and there are at most 
  $\lambda (10d^d h)^{2d} \log_2(4q)$ such vertices $u$ by Claim~\ref{cl:sepoutside}. 
  The claim follows.
\end{claimproof}

Claim~\ref{cl:septotal} establishes the first point of the lemma statement. For the second point,
let $C^*\coloneqq \bigcup_{j\in J} C'_{j}$ (note that this is a disjoint union by Claim~\ref{cl:sepdisjoint}) and let 
 $Z\subseteq V(G)$ be a set of at most $k$ vertices. We want to prove an upper bound on $|C^*\cap N^d_G[Z]|$. Let $u\in C^*\cap N^d_G[Z]$ and assume that $u\in C'_{j}$. At least one of the two cases holds for $u$:
 \begin{enumerate}
 \item $\dist_{G-X}(u,z)\le d$ for some $z\in Z$. By Claim~\ref{cl:sepdisjoint}, vertex $z$ can be at distance at most $d$ from $C'_j$ in $G-X$ only for at most one $j\in J$. For this value of $j$, by Claim~\ref{cl:sepinside}, vertex $z$ can be at distance at most $d$ from at most $(10d^dh)^{2d+1}$ vertices of $C'_j$. Thus in total there are at most $k \cdot (10d^dh)^{2d+1}$ vertices $u\in C^*$ satisfying this case.
   \item $\dist_{G}(u,X)\le d$. By Claim~\ref{cl:sepoutside}, there are at most $\lambda (10d^dh)^{2d} \log_2(4q)$ such vertices for each $j\in J$, hence there are at most $|J|\cdot \lambda (10d^dh)^{2d} \log_2(4q)$ such vertices in total.
   \end{enumerate}
   Thus we can bound  $|C^*\cap N^d_G[Z]|$ by $k \cdot (10d^dh)^{2d+1}+|J|\cdot \lambda (10d^dh)^{2d} \log_2(4q)$. By the disjointness of the sets $C'_{j}$ and the bound $|J|\ge p/(50d\log_2 q)$ established earlier, there is a $j\in J$ such that
   \begin{multline*}
|C'_{j}\cap N^d_G[Z]|\le k(10d^dh)^{2d+1} /|J|+ \lambda (10d^dh)^{2d} \log_2(4q)\\
     \le (k/p) \cdot 50d(10d^dh)^{2d+1}\log_2 q+\lambda (10d^dh)^{2d} \log_2 {4q}\\\le
     (k/p) \cdot \log_2 q \cdot 50d(10d^dh)^{2d+1}\cdot \lambda.
   \end{multline*}
 This finishes the proof.
\end{proof}  

\paragraph{Proof of Theorem~\ref{thm:dualdistance}.} We are now in position to complete the proof of Theorem~\ref{thm:dualdistance}. In a nutshell, if we have a $(p,q)$-structure with a large $p$, then Lemma~\ref{lem:sepchain-distance} gives us the required separators, while for a small $p$, Lemma~\ref{lem:almostdisjoint-distance} gives the required set of almost disjoint paths.

\begin{proof}[Proof of Theorem~\ref{thm:dualdistance}]
  Let $p$ and $q$ be given. Let us use the algorithm of Theorem~\ref{thm:dual0} to find a $(p_0,q_0)$-structure in $(G-X,s,t)$ for $q_0=q$. If $p_0 \ge p$, then we obtain a chain $(C_1,\dots,C_p)$ of $(s,t)$-separators in $G-X$ of size at most $q$ each. Invoking Lemma~\ref{lem:sepchain-distance} on this chain gives a sequence of separators such that one of them has the required bound on the intersection with~$N^d_G[Z]$.

  If $p_0<p$, then we invoke Lemma~\ref{lem:almostdisjoint-distance} on the $(p_0,q_0)$-structure with $r=p>p_0$. Note that  $q_0=q\ge (p+k_0/p+k_1)(10d^2h)^{d+3}\lambda>(p_0+k_0/r+k_1)(10d^2h)^{d+3}\lambda$ holds. Therefore, we get the required set of paths $P_i$ and sets $Q_i$ with $|Q_i|\le 10dr=10dp$ for every $i$.
\end{proof}

\newcommand{\sC}{\widetilde{C}}
\section{Separator improvement process}\label{sec:improve}

In this section we provide a generalization of the ``separator improvement'' process of Nederlof~\cite{DBLP:conf/stoc/Nederlof20a}
to graphs excluding a fixed minor. 
(Recall the notation $[d \geq 1]$ that equals $1$ if $d \geq 1$ and $0$ otherwise.)

\begin{theorem}\label{thm:improve}
  For every fixed integers $c,d \geq 0$ and $h \geq 4$,
  there exists an algorithm with the following specification.
  
  Denote $H=K_h$ if $d=0$ and $H=K_{3,h-3}$ if $d \geq 1$. 
  On input, the algorithm receives an $H$-minor-free graph $G$,
  an integer $k \geq 2$,
  a set $Z_0 \subseteq V(G)$,
  a separation $(A^\circ,B^\circ)$ in $G$
  such that $|N_G^d[v] \cap A^\circ \cap B^\circ| \leq c$ for every $v \in V(G)$,
  and a spanning forest $F^\circ$ of $G[A^\circ \cap B^\circ]$ 
  with less than $h$ connected components and maximum degree at most $h$. 
  We denote 
  $k_0 \coloneqq |Z_0|$ and 
  $\ell \coloneqq |A^\circ \cap B^\circ|$.

  The algorithm computes a family $\mathcal{S}$ of tuples $(A,B,C,\sC)$ of subsets of $V(G)$
  satisfying the following properties:
  \begin{enumerate}
  \item For every $(A,B,C,\sC) \in \mathcal{S}$, we have that
  \begin{itemize}
  \item $(A,B)$ is a separation in $G$,
  \item $\sC \subseteq C \subseteq A \cap B$,
  \item $|C| \leq \ell + \Oh_{d,h}((k+k_0/\sqrt{k}) \log k)$,
  \item $A \cap B \cap N^{d}[Z_0] \subseteq C$,
  \item $|\sC| \leq \Oh_{d,h}((\sqrt{k}+ \log^{[d \geq 1]}(k+k_0)) \log k)$, and
  \item for every $v \in V(G)$, $|(C \setminus \sC) \cap N_G^d[v]| = \Oh_{d,h}(c \log(k+k_0) \log k)$.%
\footnote{Note that this condition is trivially satisfied for $d=0$.}
  \end{itemize}
  \item For every positive integer $\theta$, set $W \subseteq V(G)$, and a family $\mathcal{Z}$
  of pairwise disjoint $d$-clusters in $G$
  such that 
   \begin{itemize}
   \item $|W| \geq 4\theta$,
   \item $|A^\circ \cap W| \geq \theta$, $|B^\circ \cap W| \geq \theta$, and
   \item $|\mathcal{Z}| \leq k$, 
   \end{itemize}
  there exists a tuple $(A,B,C,\sC) \in \mathcal{S}$ such that
  \begin{itemize}
  \item 
  $|A \cap W| \geq \theta$, $|B \cap W| \geq \theta$, 
  \item $A \cap B \cap \bigcup \mathcal{Z} \subseteq C$, 
  \item $|A \cap B \cap \bigcup \zZ|  =|C \cap \bigcup \mathcal{Z}|$ is bounded by
  \[  \Oh_{d,h}\left(c \left(\sqrt{k}+k_0/\sqrt{k}\right) \log^{[d \geq 1]} (k+k_0) \log k\right)  \]
  \item furthermore, if 
   \[ |\{K \in \zZ~|~K \cap A^\circ \cap B^\circ \neq \emptyset\}| > h(h+1) \sqrt{k} \]
  then \[ |\{K \in \zZ~|~K \cap C \neq \emptyset\}| \geq \sqrt{k}/c.\]
  \end{itemize}
  \item The size of the family $\mathcal{S}$ is bounded by 
   \[ \left(\Oh_{d,h}((k+k_0)\ell)\right)^{\Oh_h(\log k)}. \]
  \item The running time bound is bounded 
  polynomially in the size of $G$ and the upper bound on the size of $\mathcal{S}$
  from the previous point.
  \end{enumerate}

  Furthermore, the algorithm can be modified to a polynomial-time randomized algorithm
  that, given the same input, samples a tuple  $(A,B,C,\sC) \in \mathcal{S}$
  such that every tuple is sampled with probability at least
  \[ \left(\Oh_{d,h}((k+k_0)\ell)\right)^{-\Oh_h(\log k)}. \]
  If the algorithm is additionally given on input a set $W$ and $\theta$ as above, 
  it samples only among tuples $(A,B,C,\sC)$ satisfying $|A \cap W| \geq \theta$, $|B \cap W| \geq \theta$
  (that is, these inequalities hold with probability $1$).
\end{theorem}
Observe that for the case $d=0$ we can always take $c=1$ and use $\sC = \emptyset$. 
In other words, the constant $c$ and the element $\sC$ is only relevant for the distance $d \geq 1$ version of the theorem.

We will need the following toolbox to manipulate the spanning forest of $G[A^\circ \cap B^\circ]$.

Fix an integer $\delta \geq 1$. Assume we are given a graph $G$
and there is some unknown family $\zZ$ of at most $k$ pairwise disjoint $d$-clusters in $G$. 
An \emph{annotated forest} of a pair $(F,\zeta)$ where $F$ is a forest of maximum degree
at most $\delta$ and $\zeta$ assigns to every connected component $J$ of $F$ a guess
$\zeta(J) \in \{0,1,\ldots,ck\}$
of the actual value of $|J \cap \bigcup \zZ|$.
An annotated forest is \emph{consistent with $\zZ$} if indeed
$\zeta(J) = |J \cap \bigcup \zZ|$ for every
connected component $J$ of $F$.
We will sometimes refer to $\zeta(J)$ as to the \emph{weight} of the component $J$, 
and the \emph{weight} of $(F,\zeta)$ is the sum of the weights of all the connected components of $F$.

We need the following observation.
\begin{lemma}\label{lem:cut-tree}
  If $T$ is a tree of maximum degree at most $\delta$, and $Z \subseteq V(T)$ is of size at least $2$,
  then there exists an edge $e \in E(T)$ such that both connected components of $T\setminus\{e\}$ contain
  at least $\frac{|Z|}{\delta+1}$ vertices of $Z$.
\end{lemma}
\begin{proof}
  First orient every edge $e$ of $T$ from the component
  of $T\setminus \{e\}$ that contains fewer vertices of $Z$ to the one that contains
  more (breaking ties arbitrarily).
  Since $|E(T)| = |V(T)| - 1$, there is a vertex $v$ of outdegree $0$. We then have that
  every connected component of $T-\{z\}$ contains at most $|Z|/2$ vertices of $Z$. 
  Then, the edge $e$ between $z$ and the component of $T-\{z\}$ that contains the most vertices
  of $Z$ satisfies the requirements of the claim. 
  Here, we use the assumption $|Z| \geq 2$ to ensure that this component contains at least one vertex
  of $Z$, i.e., $Z \not\subseteq \{v\}$. 
\end{proof}

With an annotated forest $(F,\zeta)$, we define the (branching) operation of a $\zZ$-\emph{split}  as follows:
we pick $J$ to be the connected component of $F$ with maximum $\zeta(J)$,
randomly choose an edge $e$ of $J$ and a guess of $|J_1 \cap \bigcup \zZ|$ and $|J_2 \cap \bigcup \zZ|$ 
for the two connected components $J_1,J_2$ of $J \setminus \{e\}$, 
considering only values not smaller than $\frac{\zeta(J)}{\delta+1}$.
More precisely, we randomly sample a spanning forest $(F',\zeta')$ where $F' = F \setminus \{e\}$
and $e$ is a randomly chosen edge of $J$, 
$\zeta'(J') = \zeta(J')$ for connected components $J'$ of $F$ distinct than $J$, 
and for the two connected components $J_1,J_2$ of $J \setminus \{e\}$, the value
$\zeta'(J_1)$ is chosen uniformly at random from 
$\{ \lceil \frac{\zeta(J)}{\delta+1} \rceil, \ldots, \zeta(J) - \lceil \frac{\zeta(J)}{\delta+1} \rceil\}$
and $\zeta'(J_2)$ is set to $\zeta(J) - \zeta'(J_1)$. 
Te operation above is defined only when $\zeta(J) > 1$, as otherwise 
there may be no choices for $\zeta'(J_1)$ satisfying the conditions above.
Lemma~\ref{lem:cut-tree} for $Z \coloneqq J \cap \bigcup \zZ$ ensures that if $(F,\zeta)$ is consistent
with $\zZ$, then there is always a choice of $e$ and $\zeta'(J_1)$ that leads to $(F',\zeta')$
consistent with $\zZ$.

Note that a $\zZ$-split has 
  at most $|E(J)| \cdot (1+\zeta(J)) \leq (ck+1)|E(F)|$ choices for $(F',\zeta')$.
  We will only use forests $F$ that are subgraphs of the input forest $F^\circ$, and 
  hence $|E(F)| < \ell$ in this case, giving at most $(ck+1)(\ell-1)$ choices.

  We summarize the split operation with the following immediate lemma.
  \begin{lemma}\label{lem:split}
  Assume $(F,\zeta)$ is consistent with $\zZ$.

  If there exists a component $J$ of $F$ with $\zeta(J) > 1$, then the output $(F',\zeta')$
  of the $\zZ$-split operation on $(F,\zeta)$ is also consistent with $\zZ$ with probability at least 
  \[ \left((ck+1)(\ell-1)\right)^{-1}. \]
  \end{lemma}

We now proceed to the proof of \cref{thm:improve}.

\begin{proof}[Proof of \cref{thm:improve}.]
  We describe below the randomized polynomial-time algorithm. 
  The enumeration algorithm is a straightforward modification that replaces every random guess
  with branching. 

  Let $\delta \leq h$ be the maximum degree of $F^\circ$.
  Recall that $\ell = |A^\circ \cap B^\circ| = |V(F^\circ)|$.

  The algorithm will be described as a recursive algorithm that at every step
  takes polynomial time and 
  randomly guesses one out of polynomial-in-$(k+\ell)$ number of choices. We will ensure that the 
  depth of the recursion is bounded by $3h(\delta+1)\log k$. 
  This will give the promised
  bounds on the running time and the success probability.
  For the sake of the description, we fix $\theta$, $W$, and $\zZ$ as in the theorem statement.
  All guessing steps will be in fact guesses about some properties of $W$ and $\zZ$ and we ensure
  that the guess is correct in at least one of the choices.
  If $W$ and $\theta$ are known to the algorithm, in the guessing
  steps concerning the properties of $W$ and $\theta$ we know which case to choose
  and we choose it deterministically.

  \paragraph{Initial guesswork.}
  Recall that $F^\circ$ is the spanning forest of $G[A^\circ \cap B^\circ]$ given on input
  and $F^\circ$ has $\lambda^\circ < h$ connected components. 
  Before we start recursion, for every connected component of $F^\circ$,
  we guess how many vertices of $\bigcup \zZ$ are contained in the component.
  By our assumption, every $K \in \zZ$ has at most $c$ vertices in $A^\circ \cap B^\circ$;
  hence, $|A^\circ \cap B^\circ \cap \bigcup \zZ| \leq ck$.
  This step turns $F^\circ$ into an annotated forest $(F^\circ, \zeta^\circ)$
  that is consistent with $\zZ$ with probability at least $(ck+1)^{-\lambda^{\circ}}$.
  In what follows, we assume that this is the case.

  We say that a component $J$ of an annotated forest $(F,\zeta)$ is \emph{$\zZ$-heavy}
  if $\zeta(J) > (\delta+1)\sqrt{k}$
  and \emph{$\zZ$-significant}
  if $\zeta(J) > \sqrt{k}$.
  An annotated forest is \emph{$\zZ$-heavy}
  if it contains a $\zZ$-heavy
  component.

  We now perform the following loop: While $(F^\circ, \zeta^\circ)$ contains less than $h$ connected components
  and at least one $\zZ$-heavy component, perform the split operation. 
  Note that this operation strictly increases the number of components of $F^\circ$ and the two
  new components resulting from the split operation are $\zZ$-significant. 
  Hence, the loop runs for at most $h-\lambda^\circ$ iterations.
  Furthermore, as $\lambda^\circ < h$, at least one iteration is run unless
  $(F^\circ, \zeta^\circ)$ is not $\zZ$-heavy to start with.
  
  Let $(F^1, \zeta^1)$ be the final output. Note that $(F^1, \zeta^1)$ has at most $h$ connected components and
  if it is $\zZ$-heavy, then it contains exactly $h$ connected components
  and at least two $\zZ$-significant components (as at least one iteration of the loop was run). 
  Finally, observe that the probability that $(F^1,\zeta^1)$ is consistent with $\zZ$ is at least
  \[ (ck+1)^{-\lambda^\circ} \cdot \left((ck+1)(\ell-1)\right)^{-(h-\lambda^\circ)} \geq \left((ck+1)(\ell-1)\right)^{-h}. \]
  
  If $(F^1, \zeta^1)$ is not $\zZ$-heavy, then we output $(A^\circ, B^\circ, C \coloneqq A^\circ \cap B^\circ, \sC \coloneqq \emptyset)$ and terminate.
  That is,
  we output the initial separation with $C$ being the entire cutset $A^\circ \cap B^\circ$ and $\sC = \emptyset$. 
  Note that, if $(F^1,\zeta^1)$ is consistent with $\zZ$, then 
  \[ \left|A^\circ \cap B^\circ \cap \bigcup \zZ\right| = \left|V(F^1) \cap \bigcup \zZ\right| \leq h(\delta+1)\sqrt{k} = \Oh_h(\sqrt{k}). \]
  Hence, then the output satisfies all the required properties. 

  We observe that the loop does not terminate prematurely if the additional condition on $A^\circ \cap B^\circ \cap \bigcup \zZ$
  is satisfied.
  \begin{claim}\label{cl:large2large}
    If 
    \[ |\{K \in \zZ~|~K \cap A^\circ \cap B^\circ \neq \emptyset\}| > h(h+1) \sqrt{k}, \]
    then $(F^1,\zeta^1)$ is $\zZ$-heavy.
  \end{claim}
  \begin{claimproof}
  The above conditions imply $|A^\circ \cap B^\circ \cap \bigcup \zZ| > h(\delta+1) \sqrt{k}$, 
  which in turn implies that at every step of the above loop, one of the (at most $h$) components $J$ of the current forest
  has weight greater than $(\delta+1) \sqrt{k}$.
  \end{claimproof}   
  
  We proceed with the assumption that $(F^1, \zeta^1)$ is $\zZ$-heavy.
  In particular, it contains exactly $h$ connected components and at least two $\zZ$-significant components.

  \paragraph{Main procedure.}
  We are now ready to proceed to the recursion; see Figure~\ref{fig:improve-diagram} for a diagram
  of the control flow. 
  One recursive call will receive on input a separation $(A,B)$ of $G$,
  a tuple $(C,D,\sC)$ of subsets of $A \cap B$ with $\sC \subseteq C$, $C \cap D = \emptyset$, 
  $D \cap N_G^d[Z_0] = \emptyset$,
  and an annotated forest $(F,\zeta)$ where $F$ is a spanning
  forest of $(A \cap B) \setminus (C \cup D)$ satisfying:
  \begin{itemize}
    \item $F$ has exactly $h$ connected components;
    \item at least two components of $F$ are $\zZ$-significant;
    \item $F$ is a subgraph of $F^1$;
    \item every component of $F$ is either a connected component of $F^1$ or is $\zZ$-significant.
  \end{itemize} 

  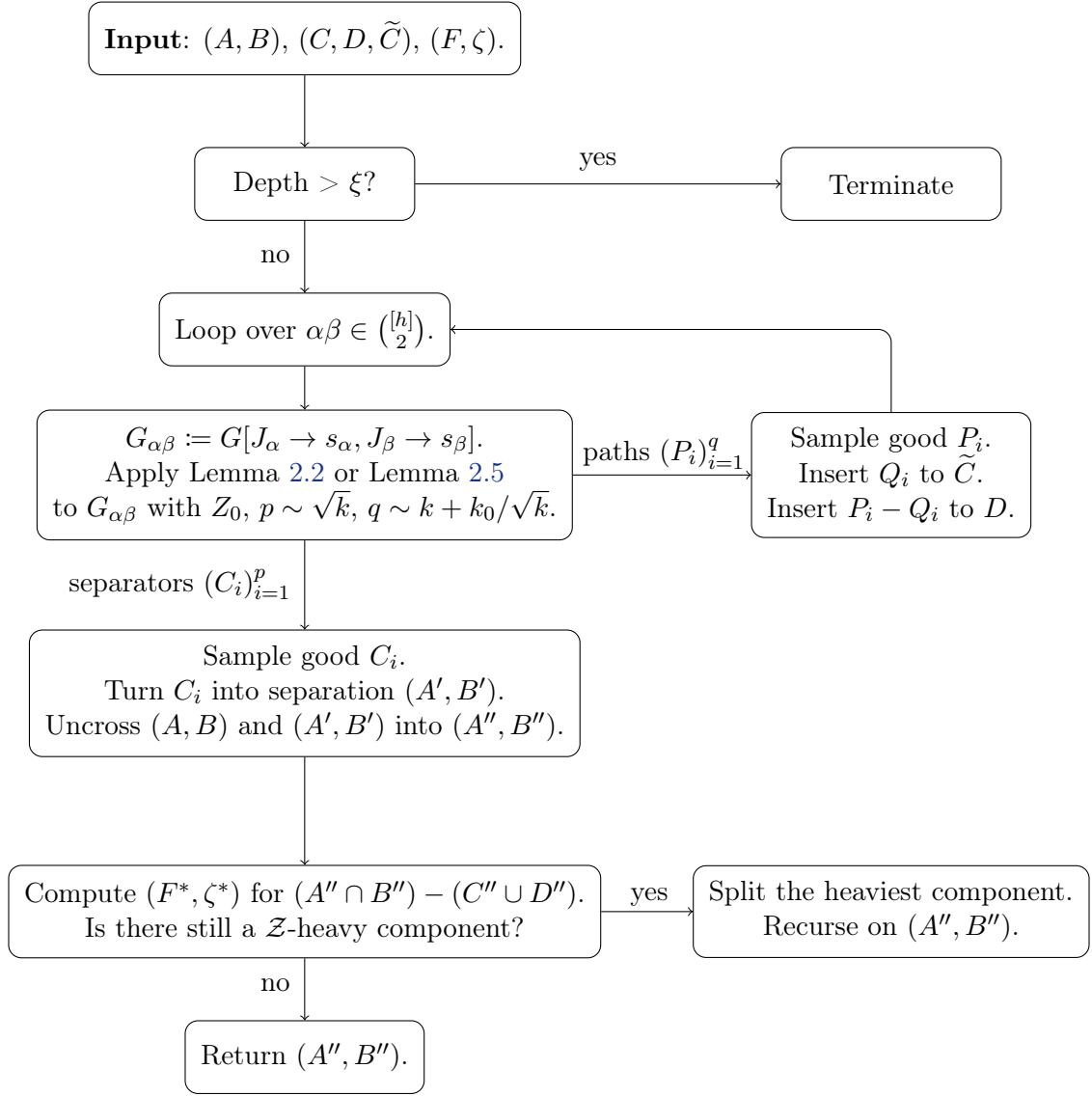
\begin{figure}[tb!]
    \begin{center}
    \begin{tikzpicture}
      \tikzset{
    rounded box/.style={
        rectangle,
        draw=black,
        rounded corners=5pt,
        minimum width=3cm,
        minimum height=1cm,
        text centered, align=center,
        inner sep=6pt,
    }
}
\node[rounded box] (start) at (0, 0) {\textbf{Input}: $(A,B)$, $(C,D,\sC)$, $(F,\zeta)$.};
\node[rounded box] (depthcheck) at (0, -2) {Depth > $\xi$?};
\node[rounded box] (toodeep) at (8, -2) {Terminate};
\draw[->] (start) -- (depthcheck);
\draw[->] (depthcheck) -- (toodeep);
\draw (4, -1.7) node {yes};
\node[rounded box] (loop) at (0, -4) {Loop over $\alpha\beta \in \binom{[h]}{2}$.};
\node[rounded box] (loopstep) at (0, -6) {$G_{\alpha\beta} \coloneqq G[J_\alpha \to s_\alpha, J_\beta \to s_\beta]$.\\
Apply \cref{cor:dual1} or \cref{thm:dualdistance}\\to $G_{\alpha\beta}$ with $Z_0$, $p \sim \sqrt{k}$, $q \sim k + k_0/\sqrt{k}$.};
\draw[->] (loop) -- (loopstep);
\draw[->] (depthcheck) -- (loop);
\draw (-0.4, -3) node {no};
\node[rounded box] (paths) at (8, -6) {Sample good $P_i$.\\
Insert $Q_i$ to $\sC$.\\
Insert $P_i \setminus Q_i$ to $D$.};
\draw[->] (loopstep) -- (paths);
\draw (4.9, -5.7) node {paths $(P_i)_{i=1}^q$};
\draw[rounded corners=5pt,->] (paths) -- (8, -4) -- (loop);

\node[rounded box] (cycles) at (0, -9) {Sample good $C_i$.\\Turn $C_i$ into separation $(A', B')$.\\
Uncross $(A,B)$ and $(A',B')$ into $(A'',B'')$.};
\draw[->](loopstep)  -- (cycles);
\draw(-1.7, -7.5) node {separators $(C_i)_{i=1}^p$};
\node[rounded box] (finish) at (0, -12) {Compute $(F^\ast, \zeta^\ast)$ for $(A'' \cap B'') \setminus (C'' \cup D'')$.\\
 Is there still a $\mathcal{Z}$-heavy component?};
 \draw[->] (cycles) -- (finish);
 \node[rounded box] (recurse) at (8, -12) {Split the heaviest component.\\Recurse on $(A'', B'')$.};
 \draw[->] (finish) -- (recurse);
 \draw (4.7, -11.8) node {yes};
 \node[rounded box] (stop) at (0, -14) {Return $(A'', B'')$.};
 \draw[->] (finish) -- (stop);
 \draw (-0.4, -13) node {no};
    \end{tikzpicture}
    
    \caption{Diagram of the control flow of the main procedure.}\label{fig:improve-diagram}
    \end{center}
  \end{figure}

  In the root call, we use $(A,B) = (A^\circ, B^\circ)$, $C=D=\sC=\emptyset$, 
  and $(F,\zeta)$ being $(F^1,\zeta^1)$.

  A recursive call terminates without any work if its depth is larger
  than $\xi \coloneqq 3h(\delta+1) \log k$. 

  We will ensure that a call at depth $j$ satisfies
  \begin{itemize}
  \item $|C| \leq \ell + j \cdot \Oh_{h,d}(k+k_0/\sqrt{k})$;
  \item $|\sC| \leq j \cdot \Oh_{h,d}(\sqrt{k} \log k)$;
  \item for every $v \in V(G)$, we have $|N_G^d[v] \cap (C \setminus \sC)| \leq j \cdot \Oh_{h,d}(\log (k+k_0))$.
  \end{itemize}
This is clearly satisfied for the root call.

  We say that a call at depth $j$ is \emph{consistent} if
  \begin{itemize}
  \item $|A \cap W|, |B \cap W| \geq \theta$;
  \item $|C \cap \bigcup \zZ|$ is bounded by 
   \[ \begin{cases} 
      j \cdot \Oh_{h,d}(\sqrt{k} + k_0/\sqrt{k}), & \mathrm{if\ } d = 0, \\
      j \cdot \Oh_{h,d}((\sqrt{k} + k_0/\sqrt{k})\log(k+k_0)), & \mathrm{if\ } d \geq 1, 
  \end{cases} \]
  \item $D \cap \bigcup \zZ = \emptyset$; and
  \item $(F,\zeta)$ is consistent with $\zZ$.
  \end{itemize}
  Note that from the assumptions that the guesses before the recursion are correct
  it follows that the root call is consistent.

  Consider now a recursive call with input $((A,B),(C,D,\sC),(F,\zeta))$
  and assume that this call is consistent.
  First, we guess whether $|A \cap W|$ of $|B \cap W|$ is larger.
  (This is a step that is deterministic if we know $\theta$ and $W$.)
  In the following description, we assume $|A \cap W| \geq |B \cap W|$, so in particular 
  $|A \cap W| \geq |W|/2\geq 2\theta$; the other case is symmetric. 

  \paragraph{Internal loop.}
  Let $J_1,\ldots,J_h$ be the connected components of $F$.
  We initialize $\sC' \coloneqq \emptyset$ and $D' \coloneqq \emptyset$;
  We iterate over all pairs $\alpha\beta \in \binom{[h]}{2}$ in an arbitrarily chosen order.
  We will maintain an invariant that $\sC'$ and $D'$ are disjoint subsets of $V(G) \setminus V(F)$
  and $G[\sC' \cup D']$ has at most $\iota$ connected components, where $\iota$ is the number of iterations
  performed so far. 
  For a fixed pair $\alpha\beta$, proceed as follows.

  Let $G_{\alpha\beta}$ be the graph $G$ with 
  $J_\alpha$ and $J_\beta$ contracted into single vertices, which we denote by $s_\alpha$
  and $s_\beta$. Note that $G_{\alpha\beta}$ is a minor of $G$. 
  Let $X = \sC' \cup D' \cup (V(F) \setminus (J_\alpha \cup J_\beta))$. 
  For $d=0$, we apply~\cref{cor:dual1}
  to $G_{\alpha\beta}-X$ with $(s,t) = (s_\alpha,s_\beta)$, $Z_0$,
  $p = \lceil \sqrt{k} \rceil$ and $q = k + \lceil k_0/\sqrt{k} \rceil + 1$. 
  For $d \geq 1$, we apply Theorem~\ref{thm:dualdistance}
  to $G_{\alpha\beta}$, $(s,t) = (s_\alpha,s_\beta)$, $X$, $Z_0$,
  \begin{align*}
  p &= \lceil \sqrt{k} \rceil + \Theta_{d,h}(\log(k+k_0)),\\
  q &= \Theta_{d,h}(k+k_0/\sqrt{k})
  \end{align*}
   equal to the minimum values the theorem allows
  for $Z_0$, the set $Z$ of size at most $k$ and $\lambda \coloneqq \binom{h+1}{2}$.
  (Note that $G[X]$ has at most $\iota + h < \binom{h+1}{2}$ connected components.)

  Assume first that a sequence $(P_1,\ldots,P_q)$ of $(s_\alpha,s_\beta)$ paths is returned
  with sets $(Q_1,\ldots,Q_q)$. 
  By the promise of \cref{cor:dual1} and \cref{thm:dualdistance},
  there exists $i \in [q]$ such that no vertex of $V(P_i) \setminus Q_i$ belongs to $\bigcup \zZ$
  or $N_G^d[Z_0]$.
  Randomly guess the index $i \in [q]$ among those indices $i$ for which $V(P_i) \setminus Q_i$ is disjoint with $N_G^d[Z_0]$.
  Denote $P^{\alpha\beta} \coloneqq P_i$.
  Insert $Q_i \setminus \{s_\alpha,s_\beta\}$ into $\sC'$
  and $V(P_i) \setminus (Q_i \cup \{s_\alpha,s_\beta\})$ into $D'$.
  Here, the endpoints $s_\alpha$ and $s_\beta$ are put neither to $\sC'$ nor to $D'$ --- they do not exist in $G$
  --- and therefore we maintain the invariant that $\sC'$ and $D'$ are disjoint subsets of $V(G) \setminus V(F)$.
  Furthermore, the new vertices inserted into $\sC' \cup D'$ form the path $P_i \setminus \{s_\alpha,s_\beta\}$,
  and hence the number of connected components of $G[\sC' \cup D']$ increased by at most one. 
  Note that in this step we put $\Oh_{d,h}(\sqrt{k} + \log^{[d \geq 1]}(k+k_0))$ vertices into $\sC'$.
  We maintain $D' \cap N_G^d[Z_0] = \emptyset$ and, furthermore, 
  if the guess of $i$ is correct, 
  we maintain $D' \cap \bigcup \zZ = \emptyset$
  and $\Oh_{d,h}(\sqrt{k} + \log^{[d \geq 1]}(k+k_0))$ vertices of $\bigcup \zZ$ are inserted into $\sC'$. 

  Before we proceed to the description of the second case --- when a chain of separators is returned ---
  we take note that this case must eventually happen. 
  \begin{claim}\label{cl:chain-is-found}
    For some $\alpha\beta \in \binom{[h]}{2}$,
  a chain $(C_1,\ldots,C_p)$ is found.
  \end{claim}
  \begin{claimproof}
  By contradiction, assume that the iteration continued for all $\alpha\beta \in \binom{[h]}{2}$
  and in each step a sequence of paths is found and the path $P^{\alpha\beta}$ is defined.
  Since in the $G^{\alpha\beta}$ we find paths excluding $X = (\sC' \cup D') \cup (V(F) \setminus (J_\alpha \cup J_\beta))$
  and $J_\alpha$, $J_\beta$ are contracted onto $s_\alpha$ and $s_\beta$, respectively, 
  and the internal vertices of $P^{\alpha\beta}$
  are inserted into $\sC' \cup D'$, the sets $K_i$ for $1 \leq i \leq h$ and paths $P^{\alpha\beta}$
  for $\alpha\beta \in \binom{[h]}{2}$ form a minor model of $K_h$ in $G$. This
  is a contradiction. 
  \end{claimproof}

  Let us perform an analysis of the loop until a chain of separators is found. 
  As for every $\alpha\beta$ we add $\Oh_{d,h}(\sqrt{k}+ \log^{[d \geq 1]}(k+k_0))$ vertices to $\sC'$ whenever a sequence
  of paths $(P_1,\ldots,P_q)$ is found, we have
  \begin{equation}\label{eq:Csize0}
   |\sC'| \leq \left(\binom{h}{2}-1\right) \Oh_{d,h}\left(\sqrt{k}+ \log^{[d \geq 1]}(k+k_0)\right) \leq \Oh_{d,h}\left(\sqrt{k}+ \log^{[d \geq 1]}(k+k_0)\right). 
  \end{equation}
  Furthermore, $D' \cap N_G^d[Z_0] = \emptyset$ and, if the guesses are correct, $D' \cap \bigcup \zZ = \emptyset$.
  
  Aggregating over all guesses so far in this recursive call, the guesses are correct
  with probability at least
  \begin{equation}\label{eq:prob0}
    \frac{1}{2} \cdot \left(\Oh_{d,h}(k+k_0)\right)^{-\left(\binom{h}{2}-1\right)} \geq 
   \left(\Oh_{d,h}(k+k_0)\right)^{-\binom{h}{2}+1}.
  \end{equation}
  
  \paragraph{Extracting a good separation from a chain of separators.}
  We now proceed to the description of the algorithm in the case when 
  a chain $(C_1,\ldots,C_p)$ of $(s_\alpha,s_\beta)$-separators is returned.
  We terminate the iteration over the pairs $\alpha\beta\in \binom{[h]}{2}$ and continue as follows.

  By the promise of \cref{cor:dual1} and \cref{thm:dualdistance},
  there exists $i \in [p]$ with $|C_i \cap \bigcup \zZ|$ bounded by
  \[   
    \Oh_{d,h}\left(\frac{k+k_0}{p} \log^{[d \geq 1]} (k+k_0)\right) \leq \Oh_{d,h}\left((\sqrt{k}+k_0/\sqrt{k}) \log^{[d \geq 1]} (k+k_0)\right).
   \]
  We randomly guess the index $i \in [p]$
  and henceforth assume that this guess is correct.
  
  Recall $X = \sC' \cup D' \cup (V(F) \setminus (J_\alpha \cup J_\beta))$.
  Construct a separation $(A',B')$ of $G$
  corresponding to $C_i$ as follows:
  set $A' \cap B' = C_i \cup X$, $A' \setminus B'$ to be $J_\alpha$
  and all vertices of $G_{\alpha\beta} \setminus (C_i \cup X)$ reachable from $s_\alpha$,
  and $B' \setminus A'$ to be $J_\beta$ and the remaining vertices of $G_{\alpha\beta} \setminus (C_i \cup X)$. 
   Since $|A \cap W| \geq 2\theta$, we have
  \begin{equation}\label{eq:theta-cut}
    |W \cap (A' \cap A)| \geq \theta \quad\mathrm{or}\quad |W \cap (B' \cap A)| \geq \theta.
  \end{equation}
  We randomly guess which inequality of~\eqref{eq:theta-cut} is satisfied.
  (This is deterministic if $W$ and $\theta$ are known to the algorithm.)
  In the following description, without loss of generality we assume that
  $|W \cap (A' \cap A)| \geq \theta$; the other case is symmetric.

  Let $Y = V(F) \setminus J_\beta$. 
  We uncross $(A',B')$ and $(A,B)$ to a separation $(A'',B'')$ of $G$ as follows:
  \begin{align*}
   A'' &\coloneqq Y \cup (A' \cap A),\ \mathrm{and}\\
   B'' &\coloneqq Y \cup (B' \cup B). 
  \end{align*}
  That is, we uncross $(A',B')$ and $(A,B)$ in a standard way to $(A' \cap A, B' \cup B)$,
  but artificially retain $Y$ in both sides of the separation. 
  
  Clearly, $(A'',B'')$ is a separation of $G$. 
  Furthermore, as $J_\beta \subseteq B' \setminus A'$, we have 
  $J_\beta \subseteq B'' \setminus A''$ and 
  \[ A'' \cap B'' \subseteq ((A \cap B) \cup (A' \cap B')) \setminus J_\beta. \]
  
  We have $|A'' \cap W| \geq \theta$ by \eqref{eq:theta-cut}
  and $|B'' \cap W| \geq \theta$ as $|B \cap W| \geq \theta$ and $B \subseteq B''$.

  Set 
  \begin{align*}
  \sC'' &\coloneqq (A'' \cap B'') \cap (\sC \cup \sC'),\\
  C'' &\coloneqq (A'' \cap B'') \cap (C \cup \sC' \cup C_i),\ \mathrm{and}\\
  D'' &\coloneqq (A'' \cap B'') \cap (D \cup D'). 
  \end{align*}
  It is important to note now that 
  \begin{equation}\label{eq:ABCDxx}
     (A'' \cap B'') \setminus (C'' \cup D'') = Y = V(F) \setminus J_\beta. 
  \end{equation}

  We now analyse the tuple $((A'',B''),(C'',D'',\sC''))$. 
  As $\sC \subseteq C$, we have $\sC'' \subseteq C''$.

  By~\eqref{eq:Csize0} and as $|C_i| = \Oh_{d,h}(k+k_0/\sqrt{k})$, we have
  \begin{equation}\label{eq:Csize1}
     |C''| \leq |C| +  \Oh_{d,h}(k+k_0/\sqrt{k})
  \end{equation}
  and
  \begin{equation}\label{eq:sCsize1}
    |\sC''| \leq |\sC| + \Oh_{d,h}(\sqrt{k}+ \log^{[d \geq 1]}(k+k_0)).
  \end{equation}
  Since $C'' \setminus (C \cup \sC'') \subseteq C_i$, for every $v \in V(G)$ it holds that
  \begin{equation}\label{eq:Cvsize1}
    \left|N_G^d[v] \cap (C'' \setminus (C \cup \sC''))\right| \leq \left|N_G^d[v] \cap C_i\right| = \Oh_{h,d}(\log (k+k_0)).
  \end{equation}
  We have $D'' \cap N_G^d[Z_0] = \emptyset$ and, 
  if every guess is correct,
  we have $D'' \cap \bigcup \zZ = \emptyset$ and
  \begin{equation}\label{eq:CZ-1}
    |(C'' \setminus C) \cap \bigcup \zZ| = 
    \Oh_{d,h}\left((\sqrt{k} + k_0/\sqrt{k}) \log^{[d \geq 1]} (k+k_0)\right)
  \end{equation}

  Aggregating over all guesses so far in this recursive call, from~\eqref{eq:prob0}
  it follows that the guesses are correct with probability at least
  \begin{equation}\label{eq:prob1}
   \left(\Oh_{d,h}(k+k_0)\right)^{-\binom{h}{2}+1} \cdot \left(\Oh_{d,h}\left(\sqrt{k}+ \log^{[d \geq 1]}(k+k_0)\right)\right)^{-1} \cdot \frac{1}{2} \geq 
   \left(\Oh_{d,h}(k+k_0)\right)^{-\binom{h}{2}}.
  \end{equation}

  \paragraph{Using the separation $(A'',B'')$: the recursion case.}
  Let $(F^\ast,\zeta^\ast)$ be the annotated forest $(F,\zeta)$ with the component $J_\beta$ removed.
  Note that $F^\ast$ is a spanning forest of $(A'' \cap B'') \setminus (C'' \cup D'')$
  and it a subgraph of $F^1$.

  We first consider a case where some component $J^\ast$ of $F^\ast$ is $\zZ$-heavy, that is, 
  it has weight more than $(\delta+1)\sqrt{k}$. 
  We apply the splitting operation to $(F^\ast,\zeta^\ast)$, obtaining $(F'',\zeta'')$,
  and recurse on the tuple $((A'',B''), (C'', D'',\sC''), (F'', \zeta''))$. 
  If all guesses so far are correct, $(F^\ast,\zeta^\ast)$ is consistent with $\zZ$ and 
  the splitting operation with probability at least $((ck+1)(\ell-1))^{-1}$
  outputs $(F'',\zeta'')$ consistent with $\zZ$.
  Furthermore, $F''$ has exactly $h$ connected components and 
  the two new connected components of $F''$ are $\zZ$-significant
  as $(F^\ast, \zeta^\ast)$ is $\zZ$-heavy.

  Let us check the assumptions on the child call. 
  The bound on $|C|$ follows from~\eqref{eq:Csize1}.
  The bound on $|N_G^d[v] \cap (C \setminus \sC)|$ follows from~\eqref{eq:Cvsize1}.
  If all the guesses are correct, the child call is consistent
  thanks to~\eqref{eq:CZ-1}.

  Finally, from~\eqref{eq:prob1} and the analysis in the previous paragraph it follows
  that, aggregating over all guesses in this recursive call, the child recursive call is consistent again  
  with probability at least 
  \[ \left(\Oh_{d,h}(k+k_0)\right)^{-\binom{h}{2}} \cdot \frac{1}{(ck+1)(\ell-1)} \geq 
   \left(\Oh_{d,h}(k+k_0)\right)^{-h^2} \cdot \ell^{-1}. \]

  \paragraph{Using the separation $(A'',B'')$: the leaf case.}
  We are left with the case when $(F^\ast, \zeta^\ast)$ is not $\zZ$-heavy, that is, 
  every connected component $J^\ast$ of $F^\ast$ satisfies
  $\zeta^\ast(J^\ast) \leq (\delta+1)\sqrt{k}$.
  This is a leaf of the recursion: we output $(A'',B'',C'' \cup Y,\sC'')$ and terminate.
  (Recall $Y = V(F) \setminus J_\beta$.)

  Let us analyze this output. We claim that it satisfies the requirements.
  Recall that the depth of the recursion is artificially cut at $\xi = 3h(\delta+1)\log k = \Oh_{h}(\log k)$.
  Hence, from~\eqref{eq:Csize1} it follows that
  \[ |C'' \cup Y| \leq \ell + \Oh_{d,h}((k+k_0/\sqrt{k}) \log k). \]
  Similarly, from~\eqref{eq:sCsize1} it follows that
  \[ |\sC''| \leq \Oh_{d,h}\left(\left(\sqrt{k}+ \log^{[d \geq 1]}(k+k_0)\right) \log k\right). \]
  From~\eqref{eq:Cvsize1} we have that for every $v \in V(G)$ it holds that
  \[ \left|N_G^d[v] \cap (C'' \setminus \sC'')\right| = \Oh_{h,d}(\log(k+k_0) \cdot \log k). \]
  Recall here that $|N_G^d[v] \cap Y| \leq c$ as $Y \subseteq V(F^\circ)$. 

  Since at least two components of $(F,\zeta)$ are $\zZ$-significant, we have $\zeta(Y) > \sqrt{k}$.
  Since $Y \subseteq V(F^\circ)$ and $|K \cap V(F^\circ)| \leq c$ for every $K \in \zZ$, we obtain
  \[ |\{K \in \zZ~|~K \cap Y \neq \emptyset\}| \geq \sqrt{k}/c. \]
  Also, as $(F^\ast, \zeta^\ast)$ is not $\zZ$-heavy, if $(F^\ast, \zeta^\ast)$ is consistent with $\zZ$, then
  \[ \left|Y \cap \bigcup \zZ\right| \leq h(\delta+1)\sqrt{k} = \Oh_h(\sqrt{k}). \]
  If every guess so far is correct, 
  from~\eqref{eq:CZ-1} it follows that $A'' \cap B'' \cap \bigcup \zZ \subseteq C'' \cup Y$ and
  \[ \left|C'' \cap \bigcup \zZ\right| \leq 
  \Oh_{d,h}\left((\sqrt{k}+k_0/\sqrt{k}) \log^{[d \geq 1]}(k+k_0) \log k\right).
  \]
  This completes the proof that the output satisfies the promised properties and,
  if the call is consistent and all guesses are correct, it also satisfies the promised
  properties with regards to $\zZ$.

  We also observe that, thanks to~\eqref{eq:prob1},
  the probability that all guesses are correct until depth $\xi = \Oh_{h}(\log k)$ is at least
  
  \[ \left(\Oh_{d,h}((k+k_0) \ell)\right)^{-\Oh_h(\log k)}. \]

  \paragraph{Bound on the recursion depth.}
  It remains to argue that, if all guesses are correct, 
  the recursion reaches a leaf no later than at depth $\xi = 3h(\delta+1)\log{k}$.
  That is, all recursive calls beyond this depth have made some incorrect guess in their history
  and can be terminated.
  
  We start with the following observation.
  \begin{claim}\label{cl:Fcomp}
    In a recursive call $((A,B),(C,D,\sC),(F,\zeta))$, every connected component $J$
    of $F$ is either one of the connected components of $F^1$
    or its weight is at least
    a $\frac{1}{h(\delta+1)}$ fraction of the total weight of $F$. 
  \end{claim}
  \begin{claimproof}
    If $J$ is not a connected component of the annotated forest at the root of the recursion
    tree, then it is created by a split of a component $J^\ast$ of an annotated forest $(F^\ast,\zeta^\ast)$
    at some ancestor of the considered call.
    Let $(\bar{F},\bar{\zeta})$ be the annotated forest in this ancestor call;
    recall that $F^\ast$ equals $\bar{F}$ with one component removed.
    Since in the splitting operation we take the component of maximum value of $\zeta^\ast$, 
    and all components of $F$ are subsets of components of $F^\ast$, $\zeta^\ast(J^\ast) \geq \zeta(J'')$
    for every connected component $J''$ of $F$. 
    Hence, $\zeta(J) \geq \frac{1}{\delta+1} \zeta(J'')$ for every connected component $J''$ of $F$.
    As $F$ has $h$ connected components, the weight of $J$ is at least a $\frac{1}{h(\delta+1)}$ fraction
    of the weight of $F$.
  \end{claimproof}

  Consider a recursive call $((A,B),(C,D,\sC),(F,\zeta))$ and its child call
  $((A'',B''), (C'',D'',\sC''), (F'',\zeta''))$. 
  $F''$ is created from $F$ by first discarding one connected component (denoted $J_\beta$ in the description
  above) and then a split operation on a remaining component of maximum value of $\zeta$.
  We say that a child call is \emph{heavy} if the weight of the discarded component
  is at least a $\frac{1}{h(\delta+1)}$ of the weight of $F$ and \emph{light} otherwise.
  \cref{cl:Fcomp} ensures that in an event
  that all guesses are correct and all calls are consistent,
  there are at most $h$ light calls and $\log k / \log \left(1-\frac{1}{h(\delta+1)}\right)$
  heavy calls. Consequently, the depth of such call is bounded by 
  \[ 1+ h + \left\lceil \left(\log \left(1-\frac{1}{h(\delta+1)}\right)\right)^{-1}\log k  \right \rceil
  \leq 3h(\delta+1)\log k. \]
  This finishes the proof of \cref{thm:improve}.
\end{proof}

\newcommand{\RetSet}{\ensuremath{\overline{M}}}
\newcommand{\cR}{\ensuremath{\widetilde{R}}}
\newcommand{\oldR}{\ensuremath{R_\mathrm{old}}}
\newcommand{\kconst}{\kappa}
\newcommand{\hconst}{h}

\section{The algorithms}\label{sec:algo}

We are now ready to prove \cref{thm:main,thm:main-d}.
The main part of this section is devoted to the proof of the statement when $G$ is nearly-embeddable without
apices and of diameter bounded linearly in $k$; this part is presented in \cref{ss:nearly}.
Lifting this case to the full statement of \cref{thm:main}
is done in \cref{sec:lift}
and to the full statement of \cref{thm:main-d} in \cref{ss:algo-d}.

\subsection{Nearly embeddable graphs without apices}\label{ss:nearly}

In this section we focus on proving the following statement,
which deals with the case when the graph is nearly-embeddable without apices
and we may use the diameter in the bounds.

We remark that for Theorem~\ref{thm:main}, the theorem below is only needed with $d=0$; 
a reader not interested in Theorem~\ref{thm:main-d} can focus only on this case.
Also, it is then safe to ignore the set $Y_0$, as for $d=0$
the properties concerning the set $Y_0$ are easy to satisfy via post-processing:
just add $Y_0$ to $V(G')$ and to every bag of the constructed tree decomposition.
We also recall the notation $[d \geq 1]$ which equals $1$ for $d \geq 1$ and $0$ for $d=0$.

\begin{theorem}\label{thm:main2}
  For every nonnegative integers $p,d$, there exists a constant $c_{d,p} > 0$
  and randomized polynomial-time algorithm that, 
  given an integer $k \geq 2$, a graph $G$
  that is 
  \begin{enumerate}[itemsep=0pt]
  \item of diameter $\Delta \geq 1$,
  \item $p$-nearly-embeddable without apices, and
  \item additionally, if $d \geq 1$, $K_{3,p}$-minor-free,
  \end{enumerate}
  and a set $Y_0 \subseteq V(G)$ of size at most $p$,
  outputs an induced subgraph $G'$ of $G$ and 
  a tree decomposition $(T,\beta)$ of $G'$ with maximum bag size bounded by
  \[ c_{d,p} (k \log k+\Delta)\]
  such that 
  $N^d[Y_0] \subseteq V(G')$, every bag of $(T,\beta)$
  contains at most 
  \[ c_{d,p} \left(\sqrt{k}+ \log (k+\Delta)\right) \log(k+\Delta) \]
  vertices of $N^d[Y_0]$ and, 
  for every 
  family $\zZ$ of pairwise disjoint $d$-clusters
  with $\sqrt{k} \leq |\zZ| \leq k$,
  with probability at least 
  \[ \left(\left(k\Delta\right)^{\frac{c_{d,p} |\zZ| \log k}{\sqrt{k}} } |V(G)|^{\frac{c_{d,p}|\zZ|}{k}}\right)^{-1} \]
  we have $\bigcup \zZ \subseteq V(G')$ and every bag of $(T,\beta)$ contains at most
  \[ d_{d,p}\left(\sqrt{k} \log k + \frac{\Delta}{\sqrt{k}}\right) \log^{1+3[d \geq 1]} k \cdot \log^{[d \geq 1]}(k + \Delta).\]
   vertices of $\bigcup \zZ$. 
\end{theorem}
\begin{proof}
Let $G$, $Y_0$, and $k$ be given on input and let $\Delta$ be the diameter of $G$. 
For the sake of analysis, fix $\zZ$ as in the statement.

We assume that $k > \kconst_{d,p}$ for some constant $\kconst_{d,p}$ depending on $d$ and $p$
to be fixed later. 
Indeed, if $k \leq \kconst_{d,p}$, 
it suffices to return $G'=G$ and a tree decomposition of $G$
of \cref{cor:XF}.

\paragraph{Definition of a recursive procedure.}
We give a recursive procedure with the following specification.
Each recursive call $\mathfrak{C}$ is given sets $M, R, \cR, \oldR$ with
$\cR,\oldR \subseteq R \subseteq M \subseteq V(G)$ and $Y_0 \cap M \subseteq R$,
where
\begin{align*}
|R| &= \Oh_{d,p}\left(k \log k+\Delta\right),\\
|\cR| &= \Oh_{d,p}\left(\left(\sqrt{k}+ \log^{[d \geq 1]}(k+\Delta)\right) \log k\right).
\end{align*}
We will ensure the following property: for every
$v \in V(G)$ we have
\begin{equation}\label{eq:Nv-property}
 \left| N_G^d[v] \cap (R \setminus \cR) \right| = \Oh_{d,p}(\log^2k \log(k+\Delta)). 
\end{equation}
We remark that for $d=0$ the properties regarding $\cR$
(including the above) will be trivially satisfied with $\cR=\emptyset$;
the set $\cR$ only plays an important role for the case $d \geq 1$.
Also, the set $\oldR$ is only used as a tool to maintain~\eqref{eq:Nv-property},
hence is only important for $d \geq 1$.

A recursive call returns a subset $\RetSet \subseteq M$ 
and a rooted tree decomposition $(T,\beta)$
of $G[\RetSet]$ with the following properties:
\begin{enumerate}[itemsep=0pt]
\item $M \cap N_G^d[R] \subseteq \RetSet \subseteq M$;
(Note that the first inclusion in particular implies $R \subseteq \RetSet$ and
for $d=0$ it is actually equivalent to $R \subseteq \RetSet$.)
\item $R$ is contained in the root bag of $(T,\beta)$;
\item every bag of $(T,\beta)$ is of size
\[ \Oh_{d,p}(k \log k + \Delta); \]
\item for every $t \in V(T)$, there exists a set 
 $\tilde{\beta}(t) \subseteq \beta(t)$ of size $\Oh_{d,p}((\sqrt{k}+ \log^{[d \geq 1]}(k+\Delta)) \log k)$ 
 such that for every $v \in V(G)$, 
 \[ \left|N_G^d[v] \cap (\beta(t) \setminus \tilde{\beta}(t)) \right| = \Oh_{d,p}(\log^2 k \log(k + \Delta)). \]
(Similarly as with $\cR$, this is trivially satisfied for $d=0$ with $\tilde{\beta}(t) = \emptyset$.)
\end{enumerate}
We say that the output $(\RetSet,(T,\beta))$ of a recursive call with input $M$, $R$, $\cR$, $\oldR$
is \emph{consistent} with $\zZ$ if the following holds:
\begin{enumerate}
  \item $M \cap \bigcup \zZ \subseteq \RetSet$, and
  \item every bag of $(T,\beta)$ contains 
    \[ \Oh_{d,p}\left(\left(\sqrt{k} \log k + \frac{\Delta}{\sqrt{k}}\right) \log^{1+3[d \geq 1]} k \cdot \log^{[d \geq 1]}(k + \Delta) \right).\]
     vertices of $M \cap \bigcup \zZ$.
\end{enumerate}

We will use simple notation $M, R, \cR, \oldR$ when we discuss a single recursive call.
In reasonings where we consider multiple recursive calls, 
we will use notation $M(\mathfrak{C})$, $R(\mathfrak{C})$, etc., for the corresponding set in the call $\mathfrak{C}$.

\paragraph{Initial call and final output.}
The initial call is made to $\oldR = \cR = R = Y_0$,
and $M= V(G)$.
The output $\RetSet$ and $(T,\beta)$ becomes the output of the whole
algorithm with $G' = G[\RetSet]$ and its tree decomposition $(T,\beta)$.
It is straightforward to verify that the promised properties of the output
$(\RetSet,(T,\beta))$ imply the promised properties of the final output of the algorithm.
In particular, the promise with regards to $N_G^d[Y_0]$ is guaranteed by the promised property regarding $\tilde{\beta}(t)$:
for every $t \in V(T)$ we have 
\begin{align*}
  \left|N_G^d[Y_0] \cap \beta(t)\right| &\leq |\tilde{\beta}(t)| + \bigcup_{v \in Y_0} \left|N_G^d[v] \cap (\beta(t) \setminus \tilde{\beta}(t)) \right|\\
&= \Oh_{d,p}\left(\left(\sqrt{k}+ \log^{[d \geq 1]}(k+\Delta)\right) \log k\right) + p \cdot \Oh_{d,p}(\log^2 k \log(k+\Delta)) \\&= \Oh_{d,p}\left(\left(\sqrt{k}+ \log^{[d \geq 1]}(k+\Delta)\right) \log(k+\Delta)\right). 
\end{align*}
Furthermore, observe that if the output $(\RetSet,(T,\beta))$ is consistent with $\zZ$, 
then the final output of the algorithm satisfies the promised properties concerning the set $\zZ$.
We will ensure that the output of the root call to the recursion is consistent with $\zZ$
with sufficiently high probability.

\paragraph{Leaves of the recursion.}
At a leaf of the recursion, if $|M \setminus R| \leq 4\sqrt{k}$,
we just return $\RetSet = M$ and a trivial tree decomposition with a single bag $M$.
The promised properties are satisfied trivially.
(We can use $\tilde{\beta}(t) = \cR \cup (M \setminus R)$ for the single bag created
in this call.)

\paragraph{Notation for a recursive call.}
Before we describe a non-leaf recursive call, let us set up some notation.

For a call $\mathfrak{C}$ and fixed $\zZ$, we denote
\begin{align*}
 \zZ_\mathrm{border}(\mathfrak{C}) &= \{K \in \zZ~|~K \cap R(\mathfrak{C}) \neq \emptyset\},\\
 \zZ_\mathrm{inside}(\mathfrak{C}) &= \{K \in \zZ~|~K \subseteq M(\mathfrak{C}) \setminus R(\mathfrak{C})\},\\
 \zZ_\mathrm{all}(\mathfrak{C}) &= \{K \in \zZ~|~K \cap M(\mathfrak{C}) \neq \emptyset\}.
\end{align*} 
Note that 
$\zZ_\mathrm{all}(\mathfrak{C}) = \zZ_\mathrm{border}(\mathfrak{C}) \uplus \zZ_\mathrm{inside}(\mathfrak{C})$ is a partition.
We will drop the argument $\mathfrak{C}$ if the call is clear from the context.

Let $a_F$ and $\delta_F$ be the bounds on the number of connected components
and the maximum degree of the forest $F$ given by \cref{cor:XF},
and $c_F$ be such that $|A \cap B| \leq c_F(\Delta + 1)$
and every vertex $v \in V(G)$ satisfies $|N_G^{d'}[v] \cap A \cap B| \leq c_F (2d'+1)$
for every $d' \geq 0$ for a separation $(A,B)$ given by~\cref{cor:XF}.
Note that $a_F$, $c_F$, and $\delta_F$ are constants depending on $p$ only. 

Apply \cref{thm:nearly-exclude} to the class
of $p$-nearly-embeddable without apices graphs,
obtaining a graph $H$ such that every $p$-nearly-embeddable without apices graph is $H$-minor-free.
Set $\hconst \coloneqq \max(p+4,a_F+1,\delta_F,|V(H)|)$. Note that $\hconst$ is a constant depending on $p$ only.

\paragraph{Non-leaf recursive call.}
The recursion alternates two types of recursive steps:
the regular mode or the pattern mode in a round-robin fashion.
Furthermore, the regular mode comes in four flavors, depending on how we initialize a variable $W$.
The mode of the recursive
call depends on the remainder mod 5 of the depth of the recursive call:
if the depth equals 0, 1, 2, or 3 mod 5, then we enter 
the regular mode with one of the four flavors, while for depth equal to 4 mod 5
we enter the pattern mode.
The root call is at depth $0$.

In some cases, we do nothing and just pass the same arguments to a recursive call
of depth one larger (so that it has a different mode). In all other cases, we will 
identify a set $W \subseteq M$, set $\theta := \lfloor |W|/4 \rfloor$,
and find a separation $(A^\circ, B^\circ)$ of $G$ that is balanced with respect to $W$
in the following sense: $|A^\circ \cap W| \geq \theta$ and $|B^\circ \cap W| \geq \theta$.
We will refine $(A^\circ, B^\circ)$ into a separation $(A,B)$ (often using Theorem~\ref{thm:improve})
in order to control the intersection of $\zZ$ with $A \cap B$. 
Based on $(A,B)$, we will invoke two recursive child subcalls and compute
the output $(\RetSet, (T,\beta))$ based on the output of the recursive subcalls. 

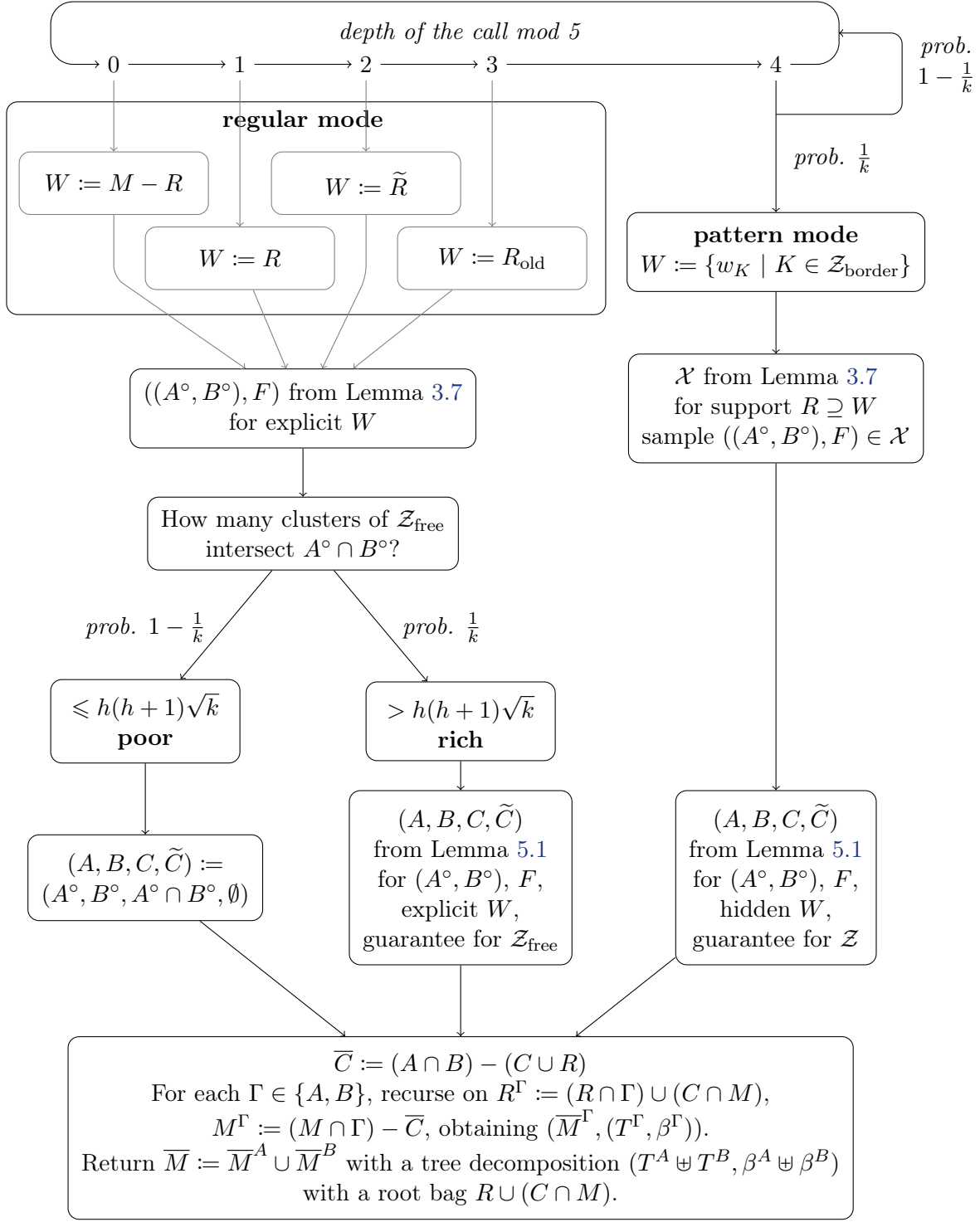
\begin{figure}[tb!]
\begin{center}
\begin{tikzpicture}
\tikzset{
    rounded box/.style={
        rectangle,
        draw=black,
        rounded corners=5pt,
        minimum width=3cm,
        minimum height=1cm,
        text centered, align=center,
        inner sep=6pt,
    }
}
  \node (d0) at (-0.5, 0) {0};
  \node (d1) at (1.5,0) {1};
  \node (d2) at (3.5, 0) {2};
  \node (d3) at (5.5, 0) {3};
  \node (d4) at (10, 0) {4};
  \draw[->] (d0) -- (d1);
  \draw[->] (d1) -- (d2);
  \draw[->] (d2) -- (d3);
  \draw[->] (d3) -- (d4);
  \draw[rounded corners=10pt,->] (d4) -- ($ (d4) + (1, 0) $) -- ($ (d4) + (1, 1) $) -- ($ (d0) + (-1, 1) $) -- ($ (d0) + (-1, 0) $) -- (d0);
  \draw (5, 0.5) node {\textit{depth of the call mod 5}};

  \draw[rounded corners=5pt, black] (-2.2, -4) rectangle (7.3, -0.6);
  \draw (2.5, -0.9) node {\textbf{regular mode}};
  \node[rounded box,draw=gray] (W0) at (-0.5, -1.9) {$W \coloneqq M \setminus R$};
  \node[rounded box,draw=gray] (W1) at (1.5, -3.1) {$W \coloneqq R$};
  \node[rounded box,draw=gray] (W2) at (3.5, -1.9) {$W \coloneqq \cR$};
  \node[rounded box,draw=gray] (W3) at (5.5, -3.1) {$W \coloneqq \oldR$};
  \node[rounded box] (W4) at (10, -3) {\textbf{pattern mode}\\$W \coloneqq \{w_K~|~K \in \zZ_\mathrm{border}\}$};
  \draw[gray,->] (d0) -- (W0);
  \draw[gray,->] (d1) -- (W1);
  \draw[gray,->] (d2) -- (W2);
  \draw[gray,->] (d3) -- (W3);
  \draw[->] (d4) -- (W4);
  \draw[rounded corners=5pt,->] (10, -0.8) -- (12, -0.8) -- (12, 0.5) -- (11, 0.5);
  \draw (10.9, -1.5) node {\it prob. $\frac{1}{k}$};
  \draw (12.7, 0) node[align=center] {\it prob.\\$1-\frac{1}{k}$};

  \node[rounded box] (Ao01) at (2.5, -5.5) {$((A^\circ, B^\circ), F)$ from \cref{cor:XF}\\for explicit $W$};
  \node[rounded box] (Ao2) at (10, -5.5) {$\mathcal{X}$ from \cref{cor:XF}\\for support $R \supseteq W$\\sample $((A^\circ, B^\circ), F) \in \mathcal{X}$};
  \draw[gray,rounded corners=5pt,->] (W0) -- ($ (W0) + (0, -1.5) $) -- (Ao01);
  \draw[gray,->] (W1) -- (Ao01);
  \draw[gray,rounded corners=5pt,->] (W2) -- ($ (W2) + (0, -1.5) $) -- (Ao01);
  \draw[gray,->] (W3) -- (Ao01);
  \draw[->] (W4) -- (Ao2);

  \node[rounded box] (q01) at (2.5, -7.5) {How many clusters of $\zZ_\mathrm{free}$\\intersect $A^\circ \cap B^\circ$?};
  \draw[->] (Ao01) -- (q01);
  \node[rounded box] (rich) at (5, -10.5) {$> \hconst(\hconst+1)\sqrt{k}$\\\textbf{rich}};
  \node[rounded box] (poor) at (0, -10.5) {$\leq \hconst(\hconst+1)\sqrt{k}$\\\textbf{poor}};
  \draw[->] (q01) -- (rich);
  \draw[->] (q01) -- (poor);
  \draw (0, -9) node {\it prob. $1-\frac{1}{k}$};
  \draw (4.7, -9) node {\it prob. $\frac{1}{k}$};
  \node[rounded box] (poor2) at (0, -13) {$(A,B,C,\sC) \coloneqq$\\$(A^\circ, B^\circ, A^\circ \cap B^\circ, \emptyset)$};
  \node[rounded box] (rich2) at (5, -13) {$(A,B,C,\sC)$\\from \cref{thm:improve}\\for $(A^\circ, B^\circ)$, $F$,\\explicit $W$,\\guarantee for $\zZ_\mathrm{free}$};
  \node[rounded box] (abc2) at (10, -13) {$(A,B,C,\sC)$\\from \cref{thm:improve}\\for $(A^\circ, B^\circ)$, $F$,\\hidden $W$,\\guarantee for $\zZ$};
  \draw[->] (rich) -- (rich2);
  \draw[->] (poor) -- (poor2);
  \draw[->] (Ao2) -- (abc2);
  \node[rounded box] (rec) at (5, -17) {$\overline{C} \coloneqq (A \cap B) \setminus (C \cup R)$\\
  For each $\Gamma \in \{A, B\}$, recurse on 
 $R^\Gamma \coloneqq (R \cap \Gamma) \cup (C \cap M)$,\\
 $M^\Gamma \coloneqq (M \cap \Gamma) \setminus \overline{C}$, obtaining $(\RetSet^\Gamma, (T^\Gamma, \beta^\Gamma))$.\\
 Return $\RetSet \coloneqq \RetSet^A \cup \RetSet^B$ with a tree decomposition $(T^A \uplus T^B, \beta^A \uplus \beta^B)$\\with a root bag $R \cup (C \cap M)$.};
   \draw[->] (rich2) -- (rec);
   \draw[->] (poor2) -- (rec);
   \draw[->] (abc2) -- (rec);
\end{tikzpicture}
\caption{Diagram of the control flow of the non-leaf recursive call.}\label{fig:algo-flow}
\end{center}
\end{figure}

We now proceed to the precise description of the modes. See Figure~\ref{fig:algo-flow}
for a diagram.
\begin{description}
\item[Regular mode.] 
First, we define $W \subseteq M$
depending on the remainder of the depth mod 5: for remainders 0, 1, 2, 3, we
set $W := M \setminus R$, $W := R$, $W := \cR$, and $W := \oldR$, respectively.
Furthermore, in the last case (i.e., remainder 3), if $|\oldR| \leq 3$, then we
reset $\oldR \coloneqq R$ before setting $W := \oldR$. 

Given $W$, apply \cref{cor:XF} with the explicit weight function usage to $G$ and a uniform measure on $W$, obtaining $((A^\circ, B^\circ),F)$.
(It is important that we do this in the whole graph $G$;
we cannot apply \cref{cor:XF} to just $G[M]$ here, as it may have a huge diameter.)
Set $\theta = \lfloor |W|/4 \rfloor \geq 1$. We thus have $|W| \geq 4\theta$,
$|A^\circ \cap W| \geq \theta$, and $|B^\circ \cap W| \geq \theta$.

We say that $A^\circ \cap B^\circ$ is \emph{rich} if $A^\circ \cap B^\circ$ intersects more than
$\hconst(\hconst+1) \sqrt{k}$ sets of $\zZ_\mathrm{inside}$, and \emph{poor} otherwise.
We randomly guess whether $A^\circ \cap B^\circ$ is rich or poor with a biased probability.
More precisely, we guess that $A^\circ \cap B^\circ$ is rich with probability
$\frac{1}{k}$
and with the remaining probability $1-\frac{1}{k}$ we assume $A^\circ \cap B^\circ$ is poor.

\begin{itemize}
\item 
If we have guessed ``rich'', then 
we invoke the randomized polynomial-time version of
\cref{thm:improve} for the constants $c \coloneqq c_F(2d+1)$, $d$, and $\hconst$ to $G$, $Z_0 := R$, $k$, $W$, $\theta$,
$(A^\circ,B^\circ)$, and the forest $F$.
We denote the obtained tuple by $(A,B,C,\sC)$.
\item
If we have guessed ``poor'', then we set
$(A,B,C,\sC) \coloneqq (A^\circ,B^\circ, A^\circ \cap B^\circ,\emptyset)$. 
\end{itemize}

\item[Pattern mode.] This mode is somewhat different than the previous as here
the set $W$ is defined based on $\zZ$, which is unknown to the algorithm.

We start by flipping a biased coin: with probability $1-\frac{1}{k}$ we do nothing
(i.e., we pass the same input to a single child recursive call with depth increased by one,
so it will enter the regular mode) and with probability $\frac{1}{k}$ we proceed further.

Let us define a set $W$ (unknown to the algorithm) in the following way.
For every $K \in \zZ_\mathrm{border}$, pick one $w_K \in K \cap R$
and let $W = \{w_K~|~K \in \zZ_\mathrm{border}\}$ and $\theta = \lfloor |W|/4 \rfloor$. 
Note that $W$ and $\theta$ are unknown to the algorithm, albeit we have $W \subseteq R$. 
We apply \cref{cor:XF} with the explicit support usage
to $G$ and $U := R$, obtaining a family $\mathcal{X}$, and select $((A^\circ,B^\circ),F) \in \mathcal{X}$ uniformly at random. 
With probability at least $(2|R|)^{-1}$, the selected $((A^\circ, B^\circ), F)$ satisfies 
$|A^\circ \cap W| \geq \theta$, and $|B^\circ \cap W| \geq \theta$.

Then, 
we invoke the randomized polynomial-time version of
\cref{thm:improve} for the constants $c \coloneqq c_F(2d+1)$, $d$, and $\hconst$ to $G$, $Z_0 := R$, $k$, 
$(A^\circ,B^\circ)$, and the forest $F$.
The obtained tuple is again denoted by $(A,B,C,\sC)$.
\end{description}

\begin{figure}[tb!]
  \begin{center}
    \begin{tikzpicture}
    \draw[rounded corners=5pt] (-5, -3) rectangle (5,1); 
    \draw[red,very thick,fill=white!80!red,rounded corners=5pt] (-0.5, -2) rectangle (0.5, 2);
    \draw[rounded corners=5pt, fill=white!80!black] (-3, 0) rectangle (3, 1);
    \draw[dashed,blue,rounded corners=5pt] (-5, -3) rectangle (0.5, 3);
    \draw[dashed,magenta,rounded corners=5pt] (5, -3) rectangle (-0.5, 3);
    \draw[dashed, rounded corners=5pt, blue, pattern={Lines[angle=-45,distance={3pt}]},pattern color=white!50!blue] 
      (-5, -1) -- (-5, 1) -- (-3, 1) -- (0.5, 1) -- (0.5, -2) -- (-0.5, -2) -- (-0.5, -3) -- (-5, -3) -- (-5, -2);
    \draw[dashed, rounded corners=5pt, blue, pattern={Lines[angle=45,distance={3pt}]},pattern color=blue]
      (-1, 1) -- (0.5, 1) -- (0.5, -2) -- (-0.5, -2) -- (-0.5, 0) -- (-3, 0) -- (-3, 1) -- (-1, 1);
    \draw (4, 0.5) node {$M$};    
    \draw (2, 0.5) node {$R$};
    \draw[blue] (-4, 2) node {$A$};
    \draw[magenta] (4, 2) node {$B$};
    \draw[red] (0, 1.5) node {$C$};
    \draw[blue] (-4, -2.5) node {$M^A$};
    \draw[blue] (0, 0.5) node {$R^A$};
    \end{tikzpicture}
  \caption{Sets used for the child recursive calls}\label{fig:sets-recursion}.
  \end{center}
\end{figure}
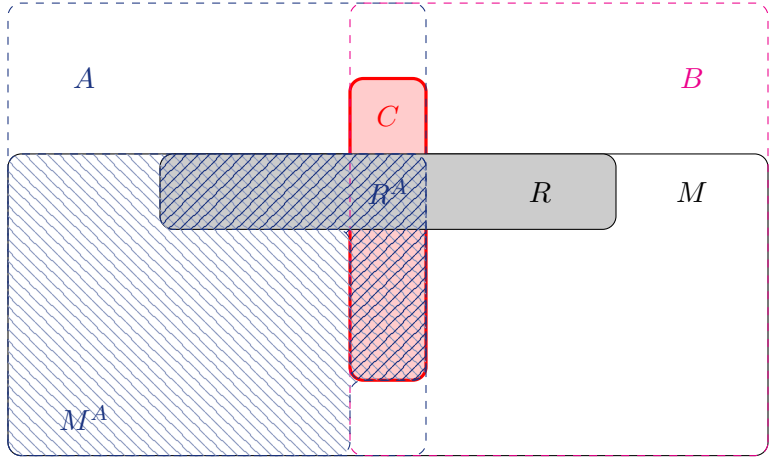

Let
\[ \overline{C} \coloneqq (A \cap B) \setminus (C \cup R). \]
For each $\Gamma \in \{A, B\}$, recurse
on 
\begin{align*}
 R^\Gamma &\coloneqq (R \cap \Gamma) \cup (C \cap M),\\
 \cR^\Gamma &\coloneqq (\cR \cap (\Gamma \setminus C)) \cup (\sC \cap M),\\
 \oldR^\Gamma &\coloneqq \oldR \cap (\Gamma \setminus C),\\
 M^\Gamma &\coloneqq (M \cap \Gamma) \setminus \overline{C},
\end{align*}
obtaining $\RetSet^\Gamma$ and $(T^\Gamma, \beta^\Gamma)$. (See also Figure~\ref{fig:sets-recursion}.)
Denote the child calls by $\mathfrak{C}^\Gamma$ for $\Gamma \in \{A,B\}$.
We set $\RetSet \coloneqq \RetSet^{A} \cup \RetSet^{B}$
and we construct the tree decomposition $(T,\beta)$ 
from the disjoint union of $(T^A,\beta^A)$ and $(T^B,\beta^B)$ by 
creating a root bag 
 \[ R \cup (C \cap M) \]
and attaching the roots of $(T^\Gamma, \beta^\Gamma)$ for $\Gamma \in \{A,B\}$
as children.
It is straightforward to verify that indeed $(T,\beta)$ is a tree decomposition of $G[\RetSet]$.
(Note that $\overline{C}$ is disjoint from $M^A \cup M^B$,  hence $\overline{C} \cap \RetSet = \emptyset$.)
This finishes the description of the recursive algorithm.

\paragraph{Analysis: Basic observations.}
We now proceed to the analysis.
We start with a few immediate observations implied by the definitions of
$R^\Gamma$, $\cR^\Gamma$, $\oldR^\Gamma$, and $M^\Gamma$. Fix $\Gamma \in \{A,B\}$.
First, note that $\oldR^\Gamma \subseteq \oldR$. 
Second, observe that $R^\Gamma \setminus R \subseteq C$
and $\cR^\Gamma \setminus \cR \subseteq \sC$,
and thus in particular
$|R^\Gamma \setminus R| \leq |C|$ and $|\cR^\Gamma \setminus \cR| \leq |\sC|$.
Third,
note that $M^\Gamma \setminus R^\Gamma \subseteq M \setminus R$
and $(M^A \setminus R^A) \cap (M^B \setminus R^B) = \emptyset$.
This implies $\zZ_\mathrm{inside}(\mathfrak{C}^\Gamma) \subseteq \zZ_\mathrm{inside}(\mathfrak{C})$.

Let 
 \[ \zZ_\mathrm{pinned}(\mathfrak{C}) \coloneqq \{K \in \zZ_\mathrm{inside}(\mathfrak{C})~|~K \cap A \cap B \neq \emptyset\}. \]
We make the following immediate observation.
\begin{claim}\label{cl:wZsplit}
\[ \zZ_\mathrm{inside}(\mathfrak{C}) = \zZ_\mathrm{inside}(\mathfrak{C}^A) \uplus \zZ_\mathrm{inside}(\mathfrak{C}^B) \uplus \zZ_\mathrm{pinned}(\mathfrak{C}) \]
is a partition. 
\end{claim}

We now analyse the output $(A,B,C,\sC)$ regardless of the mode. 
\begin{claim}\label{cl:simple-ABC-props}
  The following inequalities hold:
\begin{align}
 |C|  &\leq c_F(\Delta+1) + \Oh_{d,p}(k \log k) + \Oh_{d,p}\left(\frac{|R| \log k}{\sqrt{k}}\right), \label{eq:Cbound} \\
 |\sC| &\leq \Oh_{d,p}\left(\left(\sqrt{k}+ \log^{[d \geq 1]}(k+\Delta)\right) \log k\right), \label{eq:sCbound} \\
 C &\supseteq N_{G}^d[R] \cap A \cap B, \label{eq:NdR-survives} \\
 \forall_{v \in V(G)}\quad \left|N_G^d[v] \cap (C \setminus \sC)\right| &= \Oh_{d,p}(\log^2 k).\label{eq:Nvbound}
\end{align}
\end{claim}
\begin{claimproof}
  If the tuple $(A,B,C,\sC)$ comes from an application of Theorem~\ref{thm:improve},
  the claim follows directly from the guarantees of Theorem~\ref{thm:improve}. Otherwise,
  we have $(A,B,C,\sC) = (A^\circ, B^\circ, A^\circ \cap B^\circ, \emptyset)$;
  then \eqref{eq:Cbound} and~\eqref{eq:Nvbound} follow from the guarantees of \cref{cor:XF}
  while \eqref{eq:sCbound} and~\eqref{eq:NdR-survives} are vacuously true.
\end{claimproof}
We also note that in the regular mode we always have 
$|A \cap W| \geq \theta$ and $|B \cap W| \geq \theta$, as guaranteed by \cref{thm:improve}
for explicit $W$ and $\theta$.
In the pattern mode, as $W$ and $\theta$ are unknown to the algorithm, these inequalities are true
if the application of \cref{thm:improve} is successful with regards to $W$ and $\theta$.

If \cref{thm:improve} is invoked, we apply its promise to the set $\wZ$ defined as follows.
\[ \wZ := \begin{cases} \zZ_\mathrm{inside} & \mathrm{in\ the\ regular\ mode,} \\
  \zZ & \mathrm{in\ the\ pattern\ mode.} \end{cases} \]
That is, with probability at least
\[ \left(\Oh_{d,p}(k\Delta)\right)^{-\Oh_p(\log k)}, \]
we have 
\begin{equation}\label{eq:Zsurvives} A \cap B \cap \bigcup \wZ \subseteq C \end{equation}
and
\begin{equation}
|C \cap \bigcup \wZ| = \Oh_{d,p}\left(\left(\sqrt{k} + \frac{|R|}{\sqrt{k}}\right) \log^{1+[d \geq 1]} k\right).\label{eq:CCbound}
\end{equation}

Furthermore, within the same probability bound, 
\begin{itemize}
\item in the pattern mode, we also have $|A \cap W| \geq \theta$ and $|B \cap W| \geq \theta$
  (assuming that $|A^\circ \cap W| \geq \theta$ and $|B^\circ \cap W| \geq \theta$, which happens
  with probability at least $(2|R|)^{-1}$ independently of the application
  of \cref{thm:improve});
\item in the regular mode, if $X$ is indeed rich,
then additionally we have
\begin{equation}\label{eq:Cheavy}
	 |\zZ_{\mathrm{pinned}}| \geq \begin{cases} \frac{\sqrt{k}}{c_F(2d+1)} & \mathrm{if\ }d \geq 1,\\ \sqrt{k} & \mathrm{if\ }d=0.\end{cases} 
	\end{equation}
\end{itemize}
Finally, observe that~\eqref{eq:Zsurvives} and~\eqref{eq:CCbound} are trivially
true in the regular mode
if either $\zZ_\mathrm{inside} = \emptyset$
or if we are in the poor case and we correctly guessed it.

\paragraph{Analysis: Controlling sizes of sets and the recursion depth.}
We now analyse how the regular mode keeps sizes of various sets in check
and how it bounds the recursion depth.

\begin{claim}\label{cl:Wdrops}
  In the regular mode mode, 
    for every $\Gamma \in \{A,B\}$,
    it holds that (depending on the remainder mod 5 of the depth of recursion).
    \begin{description}[itemsep=0px]
    \item[Remainder 0:] $|M^\Gamma \setminus R^\Gamma| \leq \frac{6}{7} |M \setminus R|$;
    \item[Remainder 1:] $|R^\Gamma \setminus C| \leq \max(3, \frac{6}{7}|R|)$;
    \item[Remainder 2:] $|\cR^\Gamma \setminus C| \leq \max(3, \frac{6}{7}|\cR|)$.
    \item[Remainder 3:] $|\oldR^\Gamma| \leq \max(3, \frac{6}{7}|\oldR|)$.
    \end{description}
  Consequently, the depth of the recursion is bounded by
  \[ 5 \log_{7/6} n \leq 23 \log n. \]
\end{claim}
\begin{claimproof}
  In the ``Remainder 0'' case,
  as we are not in the leaf case, $|M \setminus R| > 4\sqrt{k} \geq 4$.
  In the ``Remainder 1--3'' cases, the subcases $|R| \leq 3$, $|\cR| \leq 3$, and $|\oldR| \leq 3$ respectively
  are immediate as $R^\Gamma \setminus C \subseteq R$, $\cR^\Gamma \setminus C \subseteq \cR$,
  and $\oldR^\Gamma \subseteq \oldR$,
  so we can assume that in all cases $|W| \geq 4$.
  Hence, $\theta = \lfloor |W|/4 \rfloor \geq |W|/7$. 
  Recall $|A \cap W|, |B \cap W| \geq \theta$. 
  This implies that there are at least $\theta$ vertices of $W$
  outside 
  $M^\Gamma \setminus R^\Gamma$, $R^\Gamma \setminus C$, $\cR^\Gamma \setminus C$, or $\oldR^\Gamma$,
  respectively, for $\Gamma \in \{A,B\}$. 
  The first part of the claim follows.

  The second part of the claim is immediate from the first part for ``Remainder 0''
  and the fact that every fifth call (starting from the root one, as the root is at depth $0$) is in the regular mode
  with the ``Remainder 0'' case.
\end{claimproof}
The second part of Claim~\ref{cl:Wdrops} in particular implies
that the entire algorithm runs in randomized polynomial time.

Let $c_{d,p}$ be the larger of the two constants hidden in the $\Oh_{d,p}$ notation
in~\eqref{eq:Cbound}. 
We assume $\kconst_{d,p}$ is large enough so that $k > \kconst_{d,p}$ implies
\[ c_{d,p} \cdot \frac{\log k}{\sqrt{k}} < 0.01 \]
For an integer $x \geq 0$, let
 \[f(x) = \left\lfloor c_F(\Delta + 1) + c_{d,p} k \log k + x \cdot \left(1 + c_{d,p} \frac{\log k}{\sqrt{k}}\right) \right\rfloor. \]
Since $R^\Gamma \subseteq R \cup C$, we have $|R^\Gamma| \leq f(|R|)$ thanks to~\eqref{eq:Cbound}.
Note that $k > \kconst_{d,p}$ implies
 \[ f(x) < c_F(\Delta+1) + c_{d,p} k \log k + 1.01x. \]
This implies the following bound on the size of $R$.
\begin{claim}\label{cl:RboundLambda}
  In every call $\mathfrak{C}$, 
  \begin{equation}\label{eq:RboundLambda}
|R(\mathfrak{C})| \leq 100 c_F(\Delta + 1) + 100 c_{d,p} k \log k =: \Lambda_R.
\end{equation}
\end{claim}
\begin{claimproof}
If $\mathfrak{C}'$ is an ancestor of $\mathfrak{C}$ at distance $a$, 
then $|R(\mathfrak{C}')| \leq f^{(a)}(|R(\mathfrak{C}|))$, where $f^{(a)}$ is the composition
of $a$ applications of the function $f$. 
Claim~\ref{cl:Wdrops} for Remainder 1 and~\eqref{eq:Cbound} implies that if $\mathfrak{C}$ is a child of $\mathfrak{C}'$
and $\mathfrak{C}'$ is in the regular mode with $W = R$ (i.e., remainder 1 mod 5),
then 
 \[ |R(\mathfrak{C})| \leq \max\left(3,\frac{6}{7}|R(\mathfrak{C}')|\right) + c_F(\Delta+1) + c_{d,p} k \log k + 0.01|R(\mathfrak{C}')|. \]
The claim follows by a straightforward induction on the depth of the call.
\end{claimproof}
Using~\eqref{eq:Cbound}, we have 
\begin{equation}\label{eq:CboundLambda}
|C| \leq 2c_F(\Delta+1) + 2c_{d,p} k \log k =: \Lambda_C.
\end{equation}

Since every created bag is the union of the set $R$ and a subset of the set $C$,
\eqref{eq:RboundLambda} and~\eqref{eq:CboundLambda} imply
that the constructed tree decomposition has bags of size at most 
\[ \Lambda_R + \Lambda_C = 102 c_F (\Delta + 1) + 102 c_{d,p} k \log k = \Oh_{d,p}(\Delta + k \log k). \]

From~\eqref{eq:sCbound} and \cref{cl:Wdrops} we infer that
\begin{equation}\label{eq:cR-bound}
  |\cR| = \Oh_{d,p}\left(\left(\sqrt{k}+ \log^{[d \geq 1]}(k+\Delta)\right) \log k\right).
\end{equation}

We also observe that~\eqref{eq:RboundLambda} allows us to simplify~\eqref{eq:CCbound} into
\begin{equation}
|C \cap \bigcup \wZ| = \Oh_{d,p}\left(\left(\sqrt{k} \log k + \frac{\Delta}{\sqrt{k}}\right) \log^{1+[d \geq 1]} k\right).\label{eq:CCbound2}
\end{equation}

We now use $\oldR$ to argue about property~\eqref{eq:Nv-property}.
\begin{claim}\label{cl:Nv-property}
  In the recursion tree, on every upward paths of length at least
  \[ 5 \log_{7/6} \Lambda_R = \Oh_{d,p}(\log(\Delta + k)),\]
  the algorithm at least once resets $\oldR \coloneqq R$.

  Consequently, for every call, the set $R$ is contained in a union of
  sets $C$ from the $\Oh_{d,p}(\log(\Delta+k))$ closest ancestor calls
  and at most $\max(3,p)$ other vertices.
\end{claim}
\begin{claimproof}
  The first claim is immediate from the bound~\eqref{eq:RboundLambda}
  and the ``Remainder 3'' case of \cref{cl:Wdrops}.
  
  For the second claim, consider a call $\mathfrak{C}_0$ and
  define calls $\mathfrak{C}_i$ for $i \geq 1$ as follows: 
  $\mathfrak{C}_{i}$ is the closest proper ancestor of $\mathfrak{C}_{i-1}$
  where the reset $\oldR \coloneqq R$ happened, or the root call if no such call exists.
  The first claim implies that there are $\Oh_{d,p}(\log(\Delta+k))$ calls between
  $\mathfrak{C}_{i-1}$ and $\mathfrak{C}_i$ for each $i \geq 1$.
  Hence, for $i \geq 0$ the set $R(\mathfrak{C}_i)\setminus \cR(\mathfrak{C}_i)$ is contained
  in the union of $\Oh_{d,p}(\log(\Delta + k))$
  sets $C$ computed in the calls between $\mathfrak{C}_{i}$ and $\mathfrak{C}_{i+1}$. 
  
  We focus on the calls $\mathfrak{C}_0$ and $\mathfrak{C}_1$.
  We have
  \begin{align*}
  R(\mathfrak{C}_0) &= (R(\mathfrak{C}_0) \setminus \cR(\mathfrak{C}_0)) \cup \cR(\mathfrak{C}_0)\\
  &\subseteq (R(\mathfrak{C}_0) \setminus \cR(\mathfrak{C}_0)) \cup R(\mathfrak{C}_1)\\
  &=(R(\mathfrak{C}_0) \setminus \cR(\mathfrak{C}_0)) \cup 
  (R(\mathfrak{C}_1) \setminus \cR(\mathfrak{C}_1)) \cup \cR(\mathfrak{C}_1).
  \end{align*}
  The first two sets in the union above are each contained in the union 
  of $\Oh_{d,p}(\log(k+\Delta))$ sets $C$ computed in the calls between $\mathfrak{C}_2$
  and $\mathfrak{C}_0$, while the last set is of size at most $\max(3,p)$.
\end{claimproof}
The property~\eqref{eq:Nv-property} follows now directly
from \cref{cl:Nv-property} and the bound~\eqref{eq:Nvbound}.

We now analyse the promised property regarding $N_G^d[R]$.
\begin{claim}\label{cl:NdR-ok}
We have $M \cap N_G^d[R] \subseteq \RetSet$.
\end{claim}
\begin{claimproof}
  We prove the claim in a bottom-up induction, so we assume
  $M^\Gamma \cap N_G^d[R^\Gamma] \subseteq \RetSet^\Gamma \subseteq \RetSet$ for $\Gamma \in \{A,B\}$.
  In the application of Theorem~\ref{thm:improve} (that happens either in adhesion mode or in pattern mode)
  has $Z_0 := R$ passed as an argument.
  Consequently, $N_G^d[R] \cap \bar{C} = \emptyset$; note that this is also true 
  if Theorem~\ref{thm:improve} is not used (i.e., regular mode, ``poor'' guess) as then $\bar{C} = \emptyset$. 
  Hence, $M \cap N_G^d[R] \cap A \cap B \subseteq \RetSet$ and the claim follows.
\end{claimproof}

\paragraph{Analysis: Successful run.}
We say that a call $\mathfrak{C}$ is
\begin{description}
\item[major] if $|\zZ_\mathrm{inside}(\mathfrak{C})| > \sqrt{k}$, 
\item[minor] if $1 \leq |\zZ_\mathrm{inside}(\mathfrak{C})| \leq \sqrt{k}$,
\item[futile] if $\zZ_\mathrm{inside}(\mathfrak{C}) = \emptyset$. 
\end{description}

We say that a call $\mathfrak{C}$ is \emph{locally successful} if
\begin{itemize}[itemsep=0px]
\item $\mathfrak{C}$ is futile, or
\item $\mathfrak{C}$ is minor, and: 
\begin{itemize}[itemsep=0px]
  \item $\mathfrak{C}$ is in the regular mode and it guessed ``poor'', or
  \item $\mathfrak{C}$ is in the pattern mode and it decided to do nothing, 
\end{itemize}
or
\item $\mathfrak{C}$ is major, and:
\begin{itemize}[itemsep=0px]
  \item $\mathfrak{C}$ is in the regular mode, it is poor and we have correctly guessed it, or
  \item $\mathfrak{C}$ is in the regular mode, it is rich, we have correctly guessed it, 
      and the subsequent call to Theorem~\ref{thm:improve} returned a tuple that satisfies its promise
      with respect to $\wZ = \zZ_\mathrm{inside}(\mathfrak{C})$, or
  \item $\mathfrak{C}$ is in the pattern mode, 
  it correctly guessed the tuple $((A^\circ, B^\circ), F) \in \mathcal{X}$ that separates $W$ in a balanced way,
  and the subsequent call to Theorem~\ref{thm:improve} returned a tuple that satisfies its promise
  with respect to $\wZ = \zZ$ and $W(\mathfrak{C}),\theta(\mathfrak{C})$.
\end{itemize}
\end{itemize}
For a locally successful call, we can bound the size of $\zZ_\mathrm{pinned}$ as follows.
\begin{claim}\label{cl:Zpinned-bound}
  Let $\mathfrak{C}$ be a nonleaf locally successful call.
  Then either we do nothing or
    \[ \left|\zZ_{\mathrm{pinned}}(\mathfrak{C})\right| = \Oh_{d,p}\left(\left(\sqrt{k} \log k + \frac{\Delta}{\sqrt{k}}\right) \log^{1+[d \geq 1]} k \right).\]
\end{claim}
\begin{claimproof}
  If the call is futile, then $\zZ_\mathrm{inside} = \emptyset$, thus $\zZ_{\mathrm{pinned}} = \emptyset$ and the claim is trivial.
  
  If the call is minor, then, as the call is locally successful, 
  we either do nothing or we are in the regular mode and we guessed ``poor''.
  The claim follows from the definition of a minor call, as $|\zZ_{\mathrm{pinned}}| \leq |\zZ_\mathrm{inside}| \leq \sqrt{k}$.

  If the call is major, we are either the regular mode or the pattern mode. 
  In the regular mode, we correctly guess whether $A^\circ \cap B^\circ$ is rich or poor. 
  In the ``rich'' case, we invoke Theorem~\ref{thm:improve} and expect success with regards to $\wZ = \zZ_\mathrm{inside}$
  and thus the claim follows from~\eqref{eq:CCbound2}.
  In the ``poor'' case, we have $(A,B) = (A^\circ,B^\circ)$ and the claim follows from the definition of the poor case.
  Finally, in the pattern mode, we invoke Theorem~\ref{thm:improve} and expect success with regards to $\wZ = \zZ$
  and thus the claim follows from~\eqref{eq:CCbound2}.
\end{claimproof}

Let $\Lambda_\mathrm{pinned}$ be the maximum of $3+p$, $\sqrt{k}$, the bound of~\eqref{eq:CCbound2},
and the bound on $\zZ_\mathrm{pinned}$ of Claim~\ref{cl:Zpinned-bound}.

We say that the run is \emph{successful} for $\zZ$ if every call is locally successful and, 
furthermore, for every major call $\mathfrak{C}$ in the pattern mode, 
we do nothing if and only if
\[ \left|\zZ_\mathrm{border}(\mathfrak{C})\right| \leq 30 \Lambda_{\mathrm{pinned}}. \]

We emphasize here the following important trick. In the definition of a successful run, we do not require
anything from a futile call. In this call, we have $M \cap \bigcup \zZ \subseteq N_G^d[R]$ and 
all we need is ensured by the properties regarding $N_G^d[R]$ (that are satisfied with probability 1).

\begin{claim}\label{cl:Zsurvives}
  In a successful run, if $G'$ is the induced subgraph returned by the root call,
  then $\bigcup \zZ \subseteq V(G')$.
\end{claim}
\begin{claimproof}
  We prove by bottom-up induction that for a call with input $M$, $R$, $\cR$, and $\oldR$,
  that returned $\RetSet$ and $(T,\beta)$, 
  if $\bigcup \zZ$ is disjoint with $N(M) \setminus R$, then 
  $M \cap \bigcup \zZ \subseteq \RetSet$.
  This claim at the root call will give the claim. 

  The claim is immediate for leaf calls where $\RetSet = M$. 
  If the call is futile, then $M \cap \bigcup \zZ \subseteq N_G^d[R]$ and $N_G^d[R] \cap M \subseteq \RetSet$
  by Claim~\ref{cl:NdR-ok}.
  
  Otherwise, the call is minor or major. In particular, the success of the 
  potential run of Theorem~\ref{thm:improve} implies $\bar{C} \cap M \cap \bigcup \zZ = \emptyset$. 
  Consequently, the child calls satisfy $\bigcup \zZ \cap (N(M^\Gamma) \setminus R^\Gamma) = \emptyset$. 
  The claim follows from the inductive hypothesis as $M^A \cup M^B = M$ and $\RetSet^A \cup \RetSet^B = \RetSet$.
\end{claimproof}

\paragraph{Analysis: Bound on the intersection with $\zZ$.}
We are now ready to prove the promised bound on the intersection of a single bag with $\bigcup \zZ$.
\begin{claim}\label{cl:ZR-bound}
  Assume the run is successful. Then, for every call $\mathfrak{C}$ it holds that 
  \[ \left|\zZ_\mathrm{border}(\mathfrak{C})\right| \leq 
     37 \Lambda_\mathrm{pinned}.
   \]
   Furthermore, if $\mathfrak{C}$ is in the pattern mode
   and the definition of a successful run mandates an action
   (i.e., $\mathfrak{C}$ is major and $|\zZ_\mathrm{border}(\mathfrak{C})| > 30\Lambda_\mathrm{pinned}$),
   then denoting by $\mathfrak{C}_1$ and $\mathfrak{C}_2$ the two children calls of $\mathfrak{C}$,
   we have for $i=1,2$
   \[ |\zZ_\mathrm{border}(\mathfrak{C}_i) \setminus \zZ_\mathrm{all}(\mathfrak{C}_{3-i})| \geq 6\Lambda_\mathrm{pinned}. \]
\end{claim}
\begin{claimproof}
  First, note that the claim 
  is trivially true at the root call
  as there we have $R = Y_0$ and $|Y_0| \leq p$, so $|\zZ_\mathrm{border}| \leq p < \Lambda_\mathrm{pinned}$.
  Furthermore, we are in the regular mode at the root call, so the second claim is vacuously true.

  We now prove the claim by induction on the depth of the call. 
  We will slightly strengthen the inductive hypothesis: for each
  call $\mathfrak{C}$ that is either the root call or a child of a major call, we show 
  an improved bound of $|\zZ_\mathrm{border}(\mathfrak{C})| \leq (31+\rho) \Lambda_\mathrm{pinned}$ where $\rho$ is the remainder
  of the depth of the call $\mathfrak{C}$ mod 5. 
  (We will prove the promised statement for all other calls $\mathfrak{C}$ later.)
  Note that if a call is major, then all its ancestors
  are major, too. 
  We already verified the claim for the root call.
  
  Consider then a call $\mathfrak{C}$
  with child calls $\mathfrak{C}_1,\mathfrak{C}_2$ such that $\mathfrak{C}$ is major or the root call. 
  Let $\rho$ be the remainder mod 5 of the depth of $\mathfrak{C}$.
  We assume $|\zZ_\mathrm{border}(\mathfrak{C})| \leq (31+\rho)\Lambda_\mathrm{pinned}$;
  we aim to prove the analogous bound for $\mathfrak{C}_i$, $i=1,2$ and the second part
  of the claim for $\mathfrak{C}$ if applicable.
  Note that this is immediate if $\rho \neq 4$ or
  $|\zZ_\mathrm{border}(\mathfrak{C})| \leq 30\Lambda_\mathrm{pinned}$
  thanks to \cref{cl:Zpinned-bound}, as for $i=1,2$ we have
  $\zZ_\mathrm{border}(\mathfrak{C}_i) \subseteq \zZ_\mathrm{border}(\mathfrak{C}) + \zZ_\mathrm{pinned}(\mathfrak{C})$.

  Assume then $\rho = 4$ (so $\mathfrak{C}$ is major and in the pattern mode)
  and $|\zZ_\mathrm{border}(\mathfrak{C})| > 30\Lambda_\mathrm{pinned}$.
  By the inductive hypothesis, 
  $|\zZ_\mathrm{border}(\mathfrak{C})| \leq 35 \Lambda_{\mathrm{pinned}}$.
  The definition of a successful run mandates an action at $\mathfrak{C}$.
  In the call $\mathfrak{C}$, we define $W(\mathfrak{C}) = \{w_K~|~K \in \zZ_\mathrm{border}(\mathfrak{C})\}$.
  Since the clusters of $\zZ$ are pairwise disjoint, the vertices $w_K$ are pairwise
  distinct and $|W(\mathfrak{C})| = |\zZ_\mathrm{border}(\mathfrak{C})|$.
  Hence, $\theta(\mathfrak{C}) \geq 7\Lambda_\mathrm{pinned}$. 
  By the promise of a successful application of Theorem~\ref{thm:improve},
  \[ |A(\mathfrak{C}) \cap W(\mathfrak{C})|, |B(\mathfrak{C}) \cap W(\mathfrak{C})| \geq \theta(\mathfrak{C}). \]
  Note that for $i=1,2$, every $K \in \zZ_\mathrm{border}(\mathfrak{C}_i)$
  is either intersected by $A(\mathfrak{C}) \cap B(\mathfrak{C})$ or $K \notin \zZ_\mathrm{all}(\mathfrak{C}_{3-i})$. 
  Hence, together with~\eqref{eq:CCbound2}, we have:
  \[ |\zZ_\mathrm{border}(\mathfrak{C}_i) \setminus \zZ_\mathrm{all}(\mathfrak{C}_{3-i})|
   \geq  \theta(\mathfrak{C}) - \Lambda_\mathrm{pinned} \geq 6\Lambda_\mathrm{pinned}. \]
  Hence, we have
  \[ |\zZ_\mathrm{border}(\mathfrak{C}_i)| \leq 35 \Lambda_\mathrm{border} - 6\Lambda_\mathrm{pinned} + \Lambda_\mathrm{border} = 30 \Lambda_\mathrm{pinned}. \]
  This completes the inductive proof for the calls whose parent is major
  and the root call.

  Consider now a call $\mathfrak{C}$ that is not the root and whose parent is not major.
  Let $\mathfrak{C}'$ be the lowest ancestor of $\mathfrak{C}$
  that is major or is the root call. Assume that
  $\mathfrak{C}'$ produced a separation $(A',B')$. 
  Let $\mathfrak{C}''$ be the child call of $\mathfrak{C}'$ that is an ancestor of $\mathfrak{C}$
  (possibly $\mathfrak{C}'' = \mathfrak{C}$).
  Then, every element of $\zZ_\mathrm{border}(\mathfrak{C})$ is one of the following three types:
  \begin{itemize}[itemsep=0px]
    \item an element of $\zZ_\mathrm{border}(\mathfrak{C}')$; there are at most $35 \Lambda_\mathrm{pinned}$ such elements from the inductive assumption;
    \item intersected by $A' \cap B'$; there are at most $\Lambda_{\mathrm{pinned}}$ such elements;
    \item in $\zZ_\mathrm{inside}(\mathfrak{C}'')$; there are at most $\sqrt{k} \leq \Lambda_\mathrm{pinned}$ such elements, as $\mathfrak{C}''$ is not major.
  \end{itemize}
  This completes the proof of the claim.
\end{claimproof}

\begin{claim}\label{cl:success-to-small}
  If the run is successful, then every bag of the output tree decomposition 
  contains at most
  \[ \Oh_{d,p}\left(\left(\sqrt{k} \log k + \frac{\Delta}{\sqrt{k}}\right) \log^{1+3[d \geq 1]} k \cdot \log^{[d \geq 1]}(k + \Delta) \right).\]
  vertices of $\bigcup \zZ$.
\end{claim}
\begin{claimproof}
  Since every bag of the output decomposition is a union of the set $R$
  and either part of the separator $A \cap B$ (non-leaf node) or at most $4\sqrt{k} \leq 4\Lambda_\mathrm{pinned}$ vertices (leaf node),
  every bag of the output decomposition intersects at most $41 \Lambda_\mathrm{pinned}$
  elements of $\zZ$ due to Claims~\ref{cl:Zpinned-bound} and~\ref{cl:ZR-bound}.
  The same bound applies to the number of vertices of $\bigcup \zZ$ if $d=0$
  while the bound for $d \geq 1$ follows from property~\eqref{eq:Nv-property}
  and the bound~\eqref{eq:cR-bound}.
\end{claimproof}

\paragraph{Analysis: Probability of being successful.}
It remains to bound the probability that the run is successful. 
There are two critical observations made in the two subsequent claims: that the
number of calls where we require to take action in the pattern mode, 
or we require regular mode to make the ``rich'' choice, 
is bounded by $\Oh_{d,p}(|\zZ|/\sqrt{k})$.

\begin{claim}\label{cl:num-rich}
  In a successful run, there are at most $\Oh_{d,p}(|\zZ|/\sqrt{k})$ calls
  where we are in the regular mode and $A^\circ \cap B^\circ$ is rich. 
  The bound becomes just $|\zZ|/\sqrt{k}$ for $d=0$.
\end{claim}
\begin{proof}
First, observe that a cluster $K \in \zZ$ belongs to sets $\zZ_\mathrm{inside}$ in one downward
path (possibly empty if $K \cap Y_0 \neq \emptyset$) from the root of the recursion tree,
until a leaf call or a call where $K$ is intersected by the computed set $C$. 
Denote this last call by $\mathfrak{C}(K)$. 
If $\mathfrak{C}(K)$ is nonleaf, in the regular mode, and rich, then by~\eqref{eq:Cheavy} 
this call equals $\mathfrak{C}(K)$ for $\Omega_{d,p}(\sqrt{k})$ other sets $K \in \zZ$. 
The claim follows.

The bound becomes $|\zZ|/\sqrt{k}$ for $d=0$, as then $\Omega_{d,p}(\sqrt{k})$ in~\eqref{eq:Cheavy} 
is in fact just $\sqrt{k}$.
\end{proof}
\begin{claim}\label{cl:num-panic}
  In a successful run, there are $\Oh(|\zZ| / \sqrt{k})$ calls 
  in the pattern mode
  there where the definition of a successful run mandates an action.
  (i.e., $|\zZ_\mathrm{border}| > 30 \Lambda_{pinned}$).
\end{claim}
\begin{claimproof}
  A call is \emph{interesting} if it is major, in the pattern mode,
  and the definition of a successful run mandates an action.

  The recursion tree is a binary tree --- every non-leaf node has exactly two recursive child calls. 
  For a non-leaf call $\mathfrak{C}$ with children $\mathfrak{C}_1$ and $\mathfrak{C}_2$, 
  we designate the \emph{heavy} child as the one with larger $\zZ_\mathrm{inside}(\mathfrak{C}_i)$, breaking ties arbitrarily;
  the second child is called the \emph{light} child. 
  After removing the edges between a node and its light child, the recursion tree decomposes into \emph{heavy paths}.
  The \emph{leader} of a heavy path is its topmost node.

  Consider now a heavy path $P$ with its leader $\mathfrak{C}_P$. 
  If $\mathfrak{C}_P$ is not major, then no call on $P$ is interesting, so assume $\mathfrak{C}_P$ is major.
  Observe that if $\mathfrak{C}$ is an interesting call on $P$,
  and $\mathfrak{C}_h$ and $\mathfrak{C}_l$ are the heavy and the light child of $\mathfrak{C}$, respectively,
  then \cref{cl:ZR-bound} implies
   \[ \left| \zZ_\mathrm{border}(\mathfrak{C}_l) \setminus \zZ_{\mathrm{all}}(\mathfrak{C}_h) \right| \geq 6\Lambda_\mathrm{pinned}. \]
   Consequently, 
   \[ |\zZ_\mathrm{all}(\mathfrak{C}_h)| \leq |\zZ_\mathrm{all}(\mathfrak{C})| - 6\Lambda_{\mathrm{pinned}}. \]
   This implies that the number interesting calls on $P$
   is bounded by
   \[ \frac{|\zZ_\mathrm{all}(\mathfrak{C}_P)|}{6\Lambda_{\mathrm{pinned}}} \leq
   \frac{|\zZ_\mathrm{free}(\mathfrak{C}_P)| + |\zZ_{\mathrm{border}}(\mathfrak{C}_P)|}{6\Lambda_{\mathrm{pinned}}} \leq
   \frac{|\zZ_\mathrm{free}(\mathfrak{C}_P)|}{6\Lambda_{\mathrm{pinned}}} + \frac{37}{6} . \]
   Here, the last inequality follows from Claim~\ref{cl:ZR-bound}.

   Define the \emph{level} of a heavy path $P$ to be $\lfloor \log_2 |\zZ_\mathrm{free}(\mathfrak{C}_P)| \rfloor$.
   Observe that, if $\mathfrak{C}_P$ is not the root and $P'$ is the heavy path that contains the parent of 
   $\mathfrak{C}_P$, then the level of $P'$ is strictly larger than the level of $P$.
   We infer that for a fixed level $i$, the calls $\mathfrak{C}_P$
   over paths $P$ of level $i$ are pairwise not in ancestor/descendant relation in the recursion tree, and hence
   the sets $\zZ_{\mathrm{free}}(\mathfrak{C}_P)$ are pairwise disjoint.
   Consequently, we can bound the total number of interesting calls by summing over levels
   of the paths they lie on as follows:
   \begin{align*} 
    &\sum_{i=\lfloor \log_2 \sqrt{k} \rfloor}^{\lfloor \log_2 |\zZ| \rfloor} \frac{|\zZ|}{2^i} \cdot \left(\frac{2^{i+1}}{6\Lambda_{\mathrm{pinned}}} + \frac{37}{6}\right) \\
    &\qquad \leq \sum_{i=\lfloor \log_2 \sqrt{k} \rfloor}^{\lfloor \log_2 k \rfloor} \frac{2|\zZ|}{6\Lambda_{\mathrm{pinned}}} + \frac{37|\zZ|}{6\cdot 2^i} \\
    &\qquad = \Oh\left(\frac{|\zZ|}{\sqrt{k}}\right). 
   \end{align*}
   This finishes the proof.
\end{claimproof}

Another ingredient of the analysis is the following simple bound on the number of minor and major calls. 
\begin{claim}\label{cl:num-poor}
  In a successful run, the number of major and minor calls is at most
  \[ |\zZ| \cdot 23 \log |V(G)|. \]
\end{claim}
\begin{claimproof}
  Let $\mathfrak{C}$ be a minor or major call.
  That is, $\zZ_{\mathrm{free}}(\mathfrak{C}) \neq \emptyset$.
  Charge $\mathfrak{C}$ to one element of $\zZ_{\mathrm{free}}(\mathfrak{C})$.

  One element of $\zZ$ can be part of $\zZ_{\mathrm{free}}$
  only one an upward path in the recursion tree. 
  Since the depth of the recursion is bounded by $23 \log |V(G)|$, 
  the number of times every $K \in \zZ$ is charged is bounded by $23 \log |V(G)|$. The claim follows.
\end{claimproof}
Claims~\ref{cl:num-rich}, \ref{cl:num-panic}, and~\ref{cl:num-poor} imply the following lower bound
on the probability that the run is successful.
\begin{align*}
& \overbrace{\left( k \cdot \left( \Oh_{d,p}(k\Delta)^{\Oh_{d,p}(\log k)} \right) \right)^{-\Oh_{d,p}\left(\frac{|\zZ|}{\sqrt{k}}\right)} }^{\mathrm{regular\ and\ rich\ calls}} 
\cdot \overbrace{\left( k \cdot 2\Lambda_R \cdot \left( \Oh_{d,p}(k\Delta)^{\Oh_{d,p}(\log k)} \right) \right)^{-\Oh_{d,p}\left(\frac{|\zZ|}{\sqrt{k}}\right)} }^{\mathrm{pattern\ and\ mandate\ an\ action\ calls}}\\
&\qquad \overbrace{\cdot \left(1-\frac{1}{k}\right)^{|\zZ| \cdot 23\log |V(G)|}}^{\mathrm{regular\ and\ poor,\ pattern\ and\ minor\ calls}}
\geq (k\Delta)^{-\Oh_{d,p}\left(\frac{|\zZ| \log k}{\sqrt{k}}\right)} \cdot |V(G)|^{-\Oh\left(\frac{|\zZ|}{k}\right)}. 
 \end{align*} 
In the above:
\begin{itemize}[itemsep=0px]
  \item In the first term, the $k$ factor corresponds to the coin flip of choosing the ``rich'' case.
  \item In the second term, the $k$ factor corresponds to the coin flip of performing an action,
  $2\Lambda_R$ is a lower bound for the $(2|R|)^{-1}$ probability lower bound for choosing $((A^\circ, B^\circ), W) \in \mathcal{X}$
  that satisfy $|A^\circ \cap W|, |B^\circ \cap W| \geq \theta$.
  \item The the third term, the $1-\frac{1}{k}$ factor corresponds to the coin flip of choosing
  either the poor case in the regular mode or decide not to do anything in the pattern mode.
\end{itemize}
This finishes the proof of \cref{thm:main2}.
\end{proof}

\subsection{Lifting to $H$-minor-free graphs for $d=0$}\label{sec:lift}
In this section we lift the result \cref{thm:main2} to the generality of graph excluding a fixed minor, thereby proving \cref{thm:main}. In essence, this will involve applying \cref{thm:main2} in every torso of the tree decomposition provided by \cref{thm:excl-minor}, and gluing together the outcomes of those applications. 

Let $(T,\beta)$ be a rooted tree decomposition of a graph $G$. For a non-root $t$, let $\pi(t)$ be the parent of $t$.
Recall the notation $\sigma(t) = \beta(t) \cap \beta(\pi(t))$ with $\sigma(r) = \emptyset$ for the root $r$.

The following simple lemma facilitates a gluing procedure over a tree decomposition.

\begin{lemma}\label{lem:glue-RS}
    Assume we are given a graph $G$, an integer $k \geq 2$, 
    a rooted tree decomposition $(T,\beta)$ of $G$ of maximum adhesion
    size at most $q$, and for every $t \in V(T)$, 
    an induced subgraph $G_t$ of the torso of $t$ and a tree decomposition $(T_t,\beta_t)$ of $G_t$. 
    Let
    \[ \ell \coloneqq \max_{t \in V(T)} \max_{s \in V(T_t)} |\beta(s)|. \]
    Then one can in polynomial time compute an induced subgraph $G'$ of $G$ and a tree decomposition
    $(T',\beta')$ of $G'$ of maximum bag size at most $\ell + q$ 
    with the following guarantee: for every $Z \subseteq V(G)$, if for every $t \in V(T)$ it holds that
    $Z \cap (\beta(t) \setminus \sigma(t)) \subseteq V(G_t)$, then $Z \subseteq V(G')$. 
    Furthermore, if for every $t \in V(T)$, every bag of $(T_t,\beta_t)$ contains at most $\ell_Z$ vertices
    of $Z$, then every bag of $(T',\beta')$ contains at most $\ell_Z+q$ vertices of $Z$.
\end{lemma}
\begin{proof}
    Let
    \[ G' = G\left[\bigcup_{t \in V(T)} V(G_t) \setminus \sigma(t)\right]. \]

    Fix $Z \subseteq V(G)$ and let $Z_t = Z \cap (\beta(t) \setminus \sigma(t))$
    for every $t \in V(T)$. Clearly, if $Z_t \subseteq V(G_t)$, then 
    $Z_t \subseteq V(G')$. Note that $(Z_t)_{t \in V(T)}$ is a partition of $Z$
    as $(\beta(t) \setminus \sigma(t))_{t \in V(T)}$ is a partition of $V(G)$. 
    Hence, $Z \subseteq V(G')$. 

    For every $t \in V(T)$, we define an induced subgraph $G_t'$ of the torso of $t$
    in a top-to-bottom manner as follows. 
    For the root, set $G_t' = G_t$. 
    Otherwise let $G_t'$ be the induced subgraph of the torso of $t$ induced by the union of
    $V(G_t) \setminus \sigma(t)$
    and $\sigma(t) \cap V(G_{\pi(t)}')$.
    Modify $(T_t,\beta_t)$ to a tree decomposition $(T_t',\beta_t')$ of $G_t'$ as follows:
    for $t$ being the root of $T$ set $(T_t',\beta_t') = (T_t,\beta_t)$, and otherwise
    remove vertices of $\sigma(t) \cap V(G_t)$ from bags of $(T_t,\beta_t)$
    and insert vertices of $\sigma(t) \cap V(G_{\pi(t)}')$ into every bag. 
    This increases the size of every bag by at most $q$ vertices
    and every bag of $(T_t',\beta_t')$ contains at most $\ell_Z+q$ vertices of $Z$. 
    
    Combine tree decompositions $(T_t',\beta_t')$ into a tree decomposition $(T',\beta')$ of $G'$
    as follows: for every non-root $t$, pick a node of $T_{\pi(t)}'$ whose bag contains
    all vertices of $\sigma(t) \cap V(G_{\pi(t)}')$ (such a node exists since the adhesions are cliques in the torso)
    and connect an arbitrary note of $T_t'$ to this chosen node. 
    The tuple $(G',(T',\beta'))$ is the desired output.
\end{proof}

We now use \cref{lem:glue-RS} to lift \cref{thm:main2}
to apex-minor-free graphs of small diameter.

\begin{lemma}\label{lem:main-apex-diam}
  For every apex graph $H$, there exists a constant $d_H > 0$ and randomized polynomial-time algorithm that, 
  given an $H$-minor-free graph $G$ of diameter $\Delta \geq 1$ and an integer $k \geq 2$,
  outputs an induced subgraph $G'$ of $G$ and 
  a tree decomposition $(T,\beta)$ of $G'$ of maximum bag size at most
  \[ d_H(k \log k+\Delta) \]  
  such that for every $Z \subseteq V(G)$ of size at most $k$,
  with probability at least 
  \[ \left(\left(k\Delta\right)^{\frac{d_H |Z| \log k}{\sqrt{k}} } |V(G)|^{\frac{d_H |Z|}{k}}\right)^{-1} \]
  we have $Z \subseteq V(G')$ and every bag of $(T,\beta)$ contains
  at most $d_H  \left(\sqrt{k} \log k + \frac{\Delta}{\sqrt{k}}\right) \log k$ vertices
  of $Z$. 
\end{lemma}
\begin{proof}
For the sake of analysis, fix $Z \subseteq V(G)$ as in the statement.

Apply \cref{thm:excl-minor} to $G$, obtaining a tree decomposition
$(T,\beta)$ of maximum adhesion size $p_H$ such that the torso of every
bag is $p_H$-nearly-embeddable without apices, where the constant $p_H$ depends on $H$ only.

Root $T$ at an arbitrary bag. For $t \in V(T)$, let $Z_t = Z \cap (\beta(t) \setminus \sigma(t))$.

For every $t \in V(T)$, flip a coin: with probability $1/k$ say that $t$ is \emph{plentiful}
and with the remaining probability $1-1/k$ say that $t$ is \emph{scarce}.

If $t$ is scarce, we set $G_t$ to be the (entire) torso of $t$
and $(T_t,\beta_t)$ to be the tree decomposition of width $\Oh_H(\Delta)$ of $G_t$
given by \cref{cor:local-tw-nearly}.
If $t$ is plentiful, apply \cref{thm:main2} to the torso of $t$ and $k$,
obtaning an induced subgraph $G_t$ and a tree decomposition $(T_t,\beta_t)$ of $G_t$.
Recall that the torso of $t$ is $p_H$-nearly-embeddable without apices; let $c_H$
be the resulting constant from \cref{thm:main2}.

Apply Lemma~\ref{lem:glue-RS} to $G$, $(T,\beta)$ and $(G_t,(T_t,\beta_t))$ for $t \in V(T)$,
obtaining an induced subgraph $G'$ and its tree decomposition $(T',\beta')$.
Since every bag of every decomposition $(T_t,\beta_t)$ is of size $\Oh_H(k \log k + \Delta)$
and every adhesion of $(T,\beta)$ is of size $\Oh_H(1)$, every bag of $(T',\beta')$
is of size $\Oh_H(k \log k + \Delta)$. 
Furthermore, Lemma~\ref{lem:glue-RS} asserts
if for every $t \in V(T)$, every bag of $(T_t,\beta_t)$ contains $\Oh_H( \left(\sqrt{k} \log k + \frac{\Delta}{\sqrt{k}}\right) \log k)$
vertices of $Z$, then the same bound (but with a larger constant hidden in the $\Oh_H$-notation)
holds for $(T',\beta')$.

We now lower bound the probability that for every $t \in V(T)$, we have $Z_t \subseteq V(G_t)$
and every bag of $(T_t,\beta_t)$ contains $\Oh_H( \left(\sqrt{k} \log k + \frac{\Delta}{\sqrt{k}}\right) \log k)$ vertices of $Z_t$.
To this end, we observe the following.
\begin{claim}\label{cl:main-apex-bag-guarantee}
    There exists a constant $c_H' > 0$ depending on $H$ only so that the following holds.
    For every $t \in V(T)$ and for every $Z' \subseteq \beta(t)$ of size at most $k$,
    the probability that $Z' \subseteq V(G_t)$ and every bag of $(T_t,\beta_t)$ contains
    at most $c_H'  \left(\sqrt{k} \log k + \frac{\Delta}{\sqrt{k}}\right) \log k $ vertices of $Z'$ is at least
    \[ \left((k\Delta)^{\frac{c_H' |Z'| \log k}{\sqrt{k}} } |V(G)|^{\frac{c_H'|Z'|}{k}}\right)^{-1}.\]
\end{claim}
\begin{claimproof}
    For $Z' = \emptyset$, the promised guarantee holds with probability $1$
    (regardless of whether we choose ``scarce'' or ``plentiful'')
    and the expression evaluates to $1$ for any $c_H' > 0$, so the claim holds. 
    Assume now $Z' \neq \emptyset$. 

    If $1 \leq |Z'| \leq \sqrt{k}$, we want to guess ``scarce'' for $t$. This happens with probability
    \[ 1-\frac{1}{k} \geq 2^{-2/k} \geq 2^{-2\frac{|Z'|}{\sqrt{k}}} \geq k^{-2\frac{|Z'|}{\sqrt{k}}}. \]
    (We used $k \geq 2$ above.)
    Then, we trivially have $Z' \subseteq V(G_t)$ due to $V(G_t) = \beta(t)$, and 
    also trivially every bag of $(T_t,\beta_t)$ contains at most $\sqrt{k}$ vertices of $Z'$, because
     $|Z'| \leq \sqrt{k}$. 

    If $|Z'| > \sqrt{k}$, we want to guess ``plentiful'' for $t$. 
    This happens with probability
    \[ \frac{1}{k} = k^{-1} \geq k^{-\frac{|Z'|}{\sqrt{k}}}. \]
    Then, note that if the guarantee of \cref{thm:main2} holds
    for $Z'$, then the guarantee of the claim is satisfied as long as $c_H' \geq c_H$ (recall that $c_H$
    is the constant steming from Theorem~\ref{thm:main2}).
    Together with the previous guess, this holds with probability at least
   \[ \left((k\Delta)^{\frac{(c_H+2)|Z| \log k}{\sqrt{k}} } |V(G)|^{\frac{c_H|Z|}{k}}\right)^{-1} \] 
    Hence, the claim holds with $c_H' = c_H+2$. 
\end{claimproof}

Recall that $(Z_t)_{t \in V(T)}$ is a partition of $Z$.
Since the random choices for every $t \in V(T)$ are independent of each other, the probability
that for every $t \in V(T)$ we have $Z_t \subseteq V(G_t)$ and every bag of $(T_t,\beta_t)$
contains at most $c_H'  \left(\sqrt{k} \log k + \frac{\Delta}{\sqrt{k}}\right) \log k$ vertices of $Z_t$ is at least:
\[ \prod_{t \in V(T)} \left((k\Delta)^{\frac{c_H' |Z_t| \log k}{\sqrt{k}} } |V(G)|^{\frac{c_H'|Z_t|}{k}}\right)^{-1}
 = \left((k\Delta)^{\frac{c_H' |Z| \log k}{\sqrt{k}} } |V(G)|^{\frac{c_H'|Z|}{k}}\right)^{-1}. \]
This finishes the proof.
\end{proof}

We now lift the diameter assumption of \cref{lem:main-apex-diam}
with a classic application of Baker's layering technique.
\begin{lemma}\label{lem:main-apex}
  For every apex graph $H$, there exists a constant $d_H' > 0$ and randomized polynomial-time algorithm that, 
  given an $H$-minor-free graph $G$ and an integer $k \geq 2$,
  outputs an induced subgraph $G'$ of $G$ and 
  a tree decomposition $(T,\beta)$ of $G'$ of maximum bag size at most
   $d_H' k \log k $
  such that for every $Z \subseteq V(G)$ of size at most $k$,
  with probability at least 
  \[ \left(k^{\frac{d_H' |Z| \log k}{\sqrt{k}} } |V(G)|^{\frac{d_H' |Z|}{k}}\right)^{-1} \]
  we have $Z \subseteq V(G')$ and every bag of $(T,\beta)$ contains at most $d_H' \sqrt{k} \log^2k$ vertices
  of $Z$. 
\end{lemma}

\begin{proof}
Fix an excluded apex graph $H$. 
By possibly adding some edges to $H$, we assume that $H$ is $2$-connected;
note that this can be done while keeping $H$ an apex graph.
Let $G$ and $k$ be given on input, where $G$ is $H$-minor-free. 
For the sake of analysis, fix $Z \subseteq V(G)$ of size at most $k$. 

For every connected component $C$ of $G$, create a new vertex $v_C$ adjacent
to one arbitrarily chosen vertex in $C$ and perform
breadth-first search from $v_C$ in $C$. 
Since $H$ is $2$-connected, the addition of vertices $v_C$ does not break the assumption that $G$
is $H$-minor-free.
For $j \geq 0$, let $L_j$ be the set of vertices
within distance exactly $j$ from their corresponding vertex $v_C$; $(L_j)_{j \geq 0}$ 
is a partition of $V(G)$. We define $L_j = \emptyset$ for $j < 0$. 
Denote $L_{\leq i} \coloneqq \bigcup_{j \leq i} L_j$. 

Uniformly at random guess a remainder $a \in \{0,1,\ldots,k\}$ such that $\bigcup_{i \geq 0} L_{a+(k+1)i}$
does not contain a vertex of $Z$.
This is successful with probability at least
\begin{equation}\label{eq:est-baker}
 1-\frac{|Z|}{k+1} \geq 2^{-2 \frac{|Z|}{k} \log k} \geq 2^{-2 \frac{|Z|}{\sqrt{k}} \log k}. 
\end{equation}
 (In the first inequality we used $k \geq 2$ and $|Z| \leq k$.)

For a connected component $C$ of $G$ and $i \geq 0$
we define the graph $G_{C,i}$ as $G[C \cap L_{\leq a-1 + (k+1)i}]$ with, for $i > 0$, the graph
$G[C \cap L_{\leq a + (k+1)(i-1)}]$ contracted onto $v_C$;
the vertex resulting from this contraction will be denoted by $v_{C,i}$, and we set $v_{C,0} = v_C$ in $G_{C,0}$. 
Note that 
\[ V(G_{C,i}) = \{v_{C,i}\} \cup \left( \bigcup_{j=a + 1 + (k+1)(i-1)}^{a-1 + (k+1)i} C \cap L_j \right). \]
Furthermore, $V(G_{C,i})$ has radius at most $k$, and thus diameter at most $2k$. 
If the guess of the remainder $a$ is correct,
the set $Z$ is contained in the union of vertex sets $V(G_{C,i} \setminus v_{C,i})$.%
\footnote{The only reason we added vertices $v_C$ instead of just taking arbitrary vertex in $C$ as $v_C$
is to know that $v_C \notin Z$, which makes the presentation cleaner.}
Let $\mathcal{D}$ be the family of pairs $(C,i)$ such that $V(G_{C,i}) \setminus \{v_{C,i}\} \neq \emptyset$.

For every $(C,i) \in \mathcal{D}$, we apply 
\cref{lem:main-apex-diam} to $G_{C,i}$,
obtaining an induced subgraph $G_{C,i}'$ and its tree decomposition $(T_{C,i}',\beta_{C,i}')$. 
We remove $v_{C,i}$ from each $G_{C,i}'$ and its corresponding tree decomposition, 
and connect the tree decompositions
arbitrarily into a single tree decomposition $(T',\beta')$
of  the disjoint union of graphs $(G_{C,i}' \setminus \{v_{C,i}\})_{(C,i)\in \mathcal{D}}$, which is an induced subgraph of $G$.

For every $(C,i) \in \mathcal{D}$ the vertex $v_{C,i}$ is not in $Z$,
and by \cref{lem:main-apex-diam} we have $Z \cap V(G_{C,i}) \subseteq V(G_{C,i}')$
and every bag of $(T_{C,i}',\beta_{C,i}')$ contains at most 
$d_H\sqrt{k} \log^2 k$
 vertices of $Z \cap V(G_{C,i})$ with probability at least 
\[ \left(\left(k \cdot 2k\right)^{\frac{d_H |Z \cap V(G_{C,i})| \log k}{\sqrt{k}} } |V(G)|^{\frac{d_H |Z \cap V(G_{C,i})|}{k}}\right)^{-1}. \]
By multiplying the above for every $(C,i) \in \mathcal{D}$, we obtain that $Z \subseteq V(G')$
and every bag of $(T',\beta')$ contains at most $d_H \sqrt{k} \log^2 k$ vertices of $Z$ with probability
at least 
\[ \left(k^{\frac{3d_H|Z|\log k}{\sqrt{k}}} |V(G)|^{\frac{d_H |Z|}{k}}\right)^{-1}. \]
Multiplying the above by the estimate~\eqref{eq:est-baker} for the correct choice of $a$,
we obtain the desired claim for $d_H' = 3d_H + 2$.
\end{proof}

We are now ready to prove \cref{thm:main}.
\begin{proof}[Proof of \cref{thm:main}.]
    Fix a graph $H$.
    Let $G$ and $k \geq 2$ be given on input.
    Apply \cref{thm:excl-minor} to $G$;
    let $p_H$ be the constant and $(T,\beta)$ the resulting decomposition. 
    
    Root $T$ in an arbitrary node.
    Fix $Z \subseteq V(G)$ of size at most $k$ and let $Z_t = Z \cap (\beta(t) \setminus \sigma(t))$ for
    every $t \in V(T)$.
    
    Let $H'$ be the apex graph given by \cref{thm:nearly-exclude}
    for graphs that are $p_H$-nearly-embeddable without apices. 

    For every $t \in V(T)$, compute an induced subgraph $G_t$ and its tree decomposition $(T_t,\beta_t)$
    as follows.
    Let $G_t^1$ be the torso of $t$, and let $A_t$ be the set of at most $p_H$ vertices
    of $G_t^1$ such that $G_t^2 \coloneqq G_t^1 \setminus A_t$ is $p_H$-nearly-embeddable without apices.
    Apply \cref{lem:main-apex} to the $H'$-minor-free graph $G_t^2$,
    obtaining a constant $f_H \coloneqq d'_{H'}$ and a tuple $(G_t',(T_t',\beta_t'))$. 
    Modify $(G_t',(T_t',\beta_t'))$ to $(G_t,(T_t,\beta_t))$ by adding 
    all vertices of $A_t$ to $G_t$ and to every bag of the decomposition.

    By \cref{lem:main-apex}, 
    the probability that $Z_t \setminus A_t \subseteq V(G_t')$ and every bag of $(T_t',\beta_t')$
    contains at most $f_H \sqrt{k} \log^2 k$ vertices of $Z_t \setminus A_t$ is at least
      \[ \left(k^{\frac{f_H |Z_t \setminus A_t| \log k}{\sqrt{k}} } |V(G)|^{\frac{f_H|Z_t \setminus A_t|}{k}}\right)^{-1}. \]
    Hence, this is also a lower bound on the probability that $Z_t \subseteq V(G_t)$
    as $A_t \subseteq V(G_t)$. Furthermore, since $|A_t| \leq p_H$, with the above probability every bag
    of $(T_t,\beta_t)$ contains at most $f_H \sqrt{k} \log^2k + p_H \leq (f_H+p_H)\sqrt{k}\log^2k$ vertices
    of $Z_t$ (recall $k \geq 2$). 

    Multiplying over all $t \in V(T)$, as $(Z_t)_{t \in V(T)}$ is a partition of $Z$, the above event
    happens for every $t \in V(T)$ with probability at least
    \[ \left(k^{\frac{f_H |Z| \log k}{\sqrt{k}} } |V(G)|^{\frac{f_H|Z|}{k}}\right)^{-1} \geq 
    \left(k^{f_H \sqrt{k} \log k } |V(G)|^{f_H}\right)^{-1}. \]
    If this happens, the tuples $(G_t,(T_t,\beta_t))$ satisfy the prerequisites
    of \cref{lem:glue-RS}
    \cref{lem:glue-RS} gives us a tuple $(G',(T',\beta'))$ 
    such that $Z \subseteq V(G')$ and every bag of $(T',\beta')$ contains 
    at most $(f_H+p_H)\sqrt{k}\log^2k + p_H \leq (f_H+2p_H)\sqrt{k}\log^2k$ vertices of $Z$. 
    Furthermore, every bag of $(T',\beta')$ is of size $\Oh_H(k \log k)$. 
    This finishes the proof.
\end{proof}

\subsection{Lifting to $K_{3,h}$-minor-free graphs for $d\ge 1$}\label{ss:algo-d}

In this section we prove Theorem~\ref{thm:main-d}. 
The proof follows the general outline of the one from the previous section: we first
perform a relatively standard application of the Baker's technique do bound the diameter
linearly in $k$, compute the decomposition of Theorem~\ref{thm:excl-minor},
and use Theorem~\ref{thm:main2} to the torsos of bags.

However, there are two complications that we need to handle. First, a torso
of a bag may no longer be $K_{3,h}$-minor-free, so we need to be a bit more careful
and apply Theorem~\ref{thm:main2} to a subgraph of the torso that is a minor of the original graph. 
Second, with $d \geq 1$, $d$-clusters can span between many bags of the tree decomposition
of Theorem~\ref{thm:main2}. We solve this issue by making the topmost bag intersecting
a cluster ``responsible'' for the cluster, and all other bags handle the cluster
by setting $Y_0$ to be their top adhesion for the application of Theorem~\ref{thm:main2}.

We first wrap Theorem~\ref{thm:main2} into a decomposition of Theorem~\ref{thm:excl-minor}.
This is an analog of Lemma~\ref{lem:main-apex-diam}.
\begin{lemma}\label{lem:main-d-diam}
    For every integers $d \geq 1$ and $h \geq 3$ there exists a constant $c_{d,H} > 0$
    and a randomized polynomial-time algorithm that, given an $K_{3,h}$-minor-free graph $G$ of diameter
    $\Delta \geq 1$ and an integer $k \geq 2$, 
    outputs an induced subgraph $G'$ of $G$ and a tree decomposition $(T,\beta)$ of $G'$
    of maximum bag size at most 
    \[ c_{d,H}(k \log k +\Delta) \]
    such that for every family $\zZ$ of $d$-clusters with $|\zZ| \leq k$,
    with probability at least
    \[ \left(\left(k\Delta\right)^{\frac{c_{d,H} |\zZ| \log k}{\sqrt{k}}} |V(G)|^{\frac{c_{d,H}|\zZ|}{k}}\right)^{-1}\]
    we have $\bigcup \zZ \subseteq V(G')$ and every bag of $(T,\beta)$ contains at most $c_{d,H}  \left(\sqrt{k} \log k + \frac{\Delta}{\sqrt{k}}\right) \log^6k$ vertices of $\bigcup \zZ$.
\end{lemma}
\begin{proof}
For the sake of analysis, fix $\zZ$ as in the statement.

Apply Theorem~\ref{thm:excl-minor} to the $K_{3,h}$-minor-free graph $G$,
obtaining a tree decomposition $(T,\beta)$ of maximum
adhesion size $p_h$ such that the torso of every bag is $p_h$-nearly-embeddable without apices, 
where the constant $p_h$ depends on $h$ only. (Note that $K_{3,h}$ is an apex graph.)

Root $T$ at an arbitrary bag. 
Without loss of generality, we can assume that $(T,\beta)$ satisfies the following connectivity property:
for every $t \in V(T)$, if $T_{\downarrow t}$ is the subtree
of $T$ rooted at $t$, then
$D_{\downarrow t} \coloneqq \bigcup_{t' \in V(T_{\downarrow t})} \beta(t') \setminus \sigma(t)$
is connected in $G$ and $N_G(D_{\downarrow t}) = \sigma(t)$. 

For every $t \in V(T)$, proceed as follows. 
First, flip a coin: with probability $1/k$ say that $t$ is \emph{plentiful}
and with the remaining probability $1-1/k$ say that $t$ is \emph{scarce}.

If $t$ is scarce, we set $G_t$ to be the entire torso of $t$ and $(T_t,\beta_t)$ to be the tree decomposition
of $G_t$ obtained from Corollary~\ref{cor:XF}. 

It $t$ is plentiful, we proceed as follows.
Let $G_t^\bullet$ be the torso of $t$. Since $G_t^\bullet$ is $p_h$-nearly-embeddable without apices,
by Theorem~\ref{thm:nearly-exclude} it excludes some apex graph $H$ as a minor; 
in particular, $G_t^\bullet$ is $\delta_h$-degenerate for some $\delta_h$ depending only on $h$. 
Let $\preceq$ be a degeneracy ordering of $G_t^\bullet$: every $v \in \beta(t)$ has at most 
$\delta_h$ neighbors in $G_t^\bullet$ that are $\preceq$-smaller than $v$. 

Let $\mathcal{D}_t$ be the set of connected components of $G\setminus \beta(t)$. 
Note that $\mathcal{D}_t$ contains all components $D_{\downarrow t'}$ for children $t'$ of $t$,
as well as some (possibly many) components contained in the bags of the connected component of $T \setminus \{t\}$ that
contains the parent of $t$ (unless $t$ is the root of $T$). 

Let $G_t^\ast$ be constructed as follows: we start from $G$ and, for every $D \in \mathcal{D}_t$
we contract $D$ onto the $\preceq$-maximum vertex $v_D$ in $N_G(D)$.
This is possible, as $G$ is connected. Clearly, $G_t^\ast$ is both a minor of $G$
and a subgraph of the torso of $t$ (as $D$ is contained in the bags of one connected component
of $T \setminus \{t\}$ and $N_G(D)$ is contained in the adhesion of the edge between $t$ and this component).
Consequently, $G_t^\ast$ is both $K_{3,h}$-minor-free and $p_h$-nearly-embeddable without apices. 
The diameter of $G_t^\ast$ is at most $\Delta$, as it is constructed from $G$ by contractions only.
Furthermore, for every $v \in \beta(t)$ we have the following bound:
\begin{equation}\label{eq:vDbound}
    \left| \{ N_G(D)~|~D \in \mathcal{D}_t \wedge v_D = v\}   \right| \leq 2^{\delta_h}.
\end{equation}
Indeed, if $v = v_D$, then $v$ is the $\preceq$-maximum vertex of $N_G(D)$, so $N_G(D) \setminus \{v\}$
is a subset of the neighbors of $v$ in $G_t^\bullet$ that are $\preceq$-smaller than $v$, and there
are only at most $\delta_h$ such neighbors. This proves~\eqref{eq:vDbound}.

We apply Theorem~\ref{thm:main2} to $G_t^\ast$ with constants $4(d+4)$, $p \coloneqq \max(h,p_h)$
and $Y_0 \coloneqq \sigma(t)$. We denote the result by $G_t$ and $(T_t,\beta_t)$.

For every $t \in V(T)$, initialize $X_t \coloneqq \emptyset$.
For every edge $st \in E(T)$ with $s$ being a child of $t$, 
\begin{itemize}
\item If $v_{D_{\downarrow s}} \in V(G_t)$, then 
   modify $\beta_t$ by adding $\sigma(s) \cap V(G_t)$ to every bag
   that contains $v_{D_{\downarrow s}}$. 
\item Otherwise, add $\sigma(s) \cap V(G_t)$ to $X_t$.
\end{itemize}
The first option above increases the maximum bag size of $(T_t,\beta_t)$ by at most
$p_h 2^{\delta_h}$ factor thanks to~\eqref{eq:vDbound}, while creates a node $t_{s \nearrow t}$ with a bag of $\beta_t$
that contains $\sigma(s) \cap V(G_t)$. In the second case, we denote by $t_{s \nearrow t}$ and arbitrary
node of $(T_t,\beta_t)$.

We define 
\[ G' \coloneqq G\left[\left(\bigcup_{t \in V(T)} \left(V(G_t) \setminus \sigma(t)\right)\right) \setminus \left(\bigcup_{t \in V(T)} X_t \right) \right]. \]
For $t \in V(T)$, let $\beta_t'$ be equal to $\beta_t$, restricted to $V(G') \cap V(G_t)$. 
We construct a tree decomposition $(T',\beta')$ of $G'$ from a disjoint union
of all decompositions $(T_t,\beta_t')$ by adding, for every $st \in E(T)$ with $s$ being a child of $t$,
an edge between the root of $T_s$ and the node $t_{s \nearrow t}$ of $T_t$.

The fact that $(T',\beta')$ is a tree decomposition of $G'$ follows from the definition of $t_{s \nearrow t}$:
either $\beta_t(t_{s \nearrow t})$ contains all vertices of $\sigma(s) \cap V(G_t)$,
and hence $\beta_t'(t_{s \nearrow t})$ contains all vertices of $\sigma(s)$ left in $G'$,
or $\sigma(s) \cap V(G_t) \subseteq X_t$ and no vertex of $\beta(st)$ is in $G'$. 
Also, the width of $(T',\beta')$ is bounded by a multiplicative constant (depending on $h$)
times the maximum width of the decompositions $(T_t,\beta_t)$. 
It remains to analyse the probability of the promises concerning $\zZ$.

For $K \in \zZ$, let $t_K \in V(T)$ be the topmost node such that $\beta(t) \cap K \neq \emptyset$.
For $t \in V(T)$, let
\[ \mathrm{load}(t) = \left|\{ K \in \zZ~|~t_K = t\} \right|. \]
We construct a family $\zZ_t$ of $4(d+4)$-clusters in $G_t^\ast$ as follows.
First, start with $\zZ_t^1 \coloneqq \{K \in \zZ~|~t_K = t\}$. 
Second, for every $K \in \zZ_t^1$, let $K^2$ be a $d$-cluster in $G_t^\ast$
constructed from $K$ by, whenever $K \cap D \neq \emptyset$ for some $D \in \mathcal{D}_t$,
replacing $K \cap D$ with $v_D$. Let $\zZ_t^2 = \{K^2~|~K \in \zZ_t^1\}$.
Third, 
for every $K^2 \in \zZ_t^2$, create a $(d+4)$-cluster $K^3$ in $G_t^\ast$
as follows: start with $K^3 \coloneqq K^2$ and
for every child $s$ of $t$ in $T$,
if $K^2 \cap N_G(D_{\downarrow s}) \neq \emptyset$, then add all vertices of $N_G(D_{\downarrow s})$ to $K^3$.
This increases the diameter by at most $4$, hence $\zZ_t^3 = \{K^3~|~K^2 \in \zZ_t^2\}$ is a family
of $(d+4)$-clusters in $G_t^\ast$. Finally, apply Lemma~\ref{lem:clusters-disjoint}
to $\zZ_t^3$, obtaining a family $\zZ_t^4$ of pairwise disjoint $4(d+4)$-clusters in $G_t^\ast$.
We set $\zZ_t \coloneqq \zZ_t^4$.
Note that 
\begin{equation}\label{eq:zZtbound}
 |\zZ_t| = |\zZ_t^4| \leq |\zZ_t^3| = |\zZ_t^2| = |\zZ_t^1| = \mathrm{load}(t). 
\end{equation}

We say that $t \in V(T)$ is
\begin{description}[itemsep=0pt]
\item[of high load] if $|\zZ_t|  > \sqrt{k}$;
\item[of low load] if $1 \leq |\zZ_t| \leq \sqrt{k}$;
\item[of zero load] if $\zZ_t = \emptyset$.
\end{description}
Thanks to~\eqref{eq:zZtbound}, 
there are at most $|\zZ|/\sqrt{k}$ nodes of high load and at most $|\zZ|$ nodes of low load.
We want to guess \emph{plentiful} for all nodes of high load and \emph{scarce} for all nodes
of low load (and we do not care about the guess for nodes of zero load). This happens with
probability at least:
\[ \left(\frac{1}{k}\right)^{-\frac{|\zZ|}{\sqrt{k}}} \cdot \left(1-\frac{1}{k}\right)^{-|\zZ|} 
\geq 2^{-\Oh\left(\frac{|\zZ| \log k}{\sqrt{k}}\right)}. \]
Furthermore, for every node $t$ of high load, we want the application of Theorem~\ref{thm:main2}
to satisfy the promise with the set $\zZ_t$. Multiplying over all such nodes,
this happens with probability at least 
\begin{align*}
& \prod_{t\ \mathrm{of\ high\ load}} \left((k\Delta)^{\Oh_{d,h}\left(\frac{|\zZ_t| \log k}{\sqrt{k}}\right)} |V(G)|^{\Oh_{d,h}\left(\frac{|\zZ_t|}{k}\right)}\right)^{-1}\\
& \qquad \geq 
\prod_{t\ \mathrm{of\ high\ load}} \left((k\Delta)^{\Oh_{d,h}\left(\frac{\mathrm{load}(t) \log k}{\sqrt{k}}\right)} |V(G)|^{\Oh_{d,h}\left(\frac{\mathrm{load}(t)}{k}\right)}\right)^{-1}\\
& \qquad \geq 
\left((k\Delta)^{\Oh_{d,h}\left(\frac{|\zZ| \log k}{\sqrt{k}}\right)} |V(G)|^{\Oh_{d,h}\left(\frac{|\zZ|}{k}\right)}\right)^{-1}
\end{align*} 
It remains to argue that this event guarantees that $G'$ and $(T',\beta')$ have the desired properties.

\begin{claim}\label{cl:what-Xt-kills}
 For every $t,t' \in V(T)$ such that $t'$ is a descendant of $t$, 
 $X_{t'}$ is disjoint with $N_G^{d+2}[\sigma(t)]$. 
\end{claim}
\begin{claimproof}
    Note that $X_{t'} \subseteq \beta(t')$. 
    Since $N_G^{d+2}[\sigma(t)] \cap \beta(t') \subseteq N_G^{d+2}[\sigma(t')]$, it suffices to prove the claim for $t'=t$.

    Note that regardless of whether we proclaim $t$ to be plentiful or scarce, 
    we have $N_{G_t^\ast}^{4(d+4)}[\sigma(t)] \subseteq V(G_t)$. 
    Hence, for every child $s$ of $t$, if $N_G(D_{\downarrow s}) \cap N_G^{d+2}[\sigma(t)] \neq \emptyset$, 
    then $N_G(D_{\downarrow s}) \subseteq N_{G_t^0}^{d+4}[\sigma(t)] \subseteq V(G_t)$. 
    Hence, $v_{D_{\downarrow s}} \in V(G_t)$, that is, we do not insert the vertices of $\sigma(s)$
    into $X_t$ while processing the edge $st$.
    Since the choice of $st$ was arbitrary with $N_G(D_{\downarrow s}) \cap N_G^{d+2}[\sigma(t)] \neq \emptyset$,
    this finishes the proof.
\end{claimproof}
Consider $K \in \zZ$.
If the node $t_K$ is of low load, then $G_t = \beta(t_K)$. 
Otherwise, by the construction of $K^3 \in \zZ_{t_K}^3$, for every child $s$ of $t_K$,
if $N_G(D_{\downarrow s}) \cap K \neq \emptyset$, then $N_G(D_{\downarrow s}) \subseteq K^3 \subseteq V(G_{t_K})$. 
This the choice of $s$ is abitrary, we have $X_{t_K} \cap K = \emptyset$. With Claim~\ref{cl:what-Xt-kills}
and the fact that for every proper descendant $t$ of $t_K$ we have $\beta(t) \cap K \subseteq N_{G_t^\ast}^d[\sigma(t)]$,
it implies that $K \subseteq V(G')$.

Consider now a bag of $(T',\beta')$. By construction, it is equal to $\beta_t'(x)$ for some
$t \in V(T)$ and $x \in V(T_t)$. 
We observe that 
\[ \beta_t'(x) \cap \bigcup (\zZ \setminus \zZ_t^1) \subseteq N_G^d[\sigma(t)] \cap \beta(t). \]
Hence, a bound of $\Oh_{d,h}(\left(\sqrt{k} \log k + \frac{\Delta}{\sqrt{k}}\right) \log^5k)$ on the size of the left hand size above follows from 
either properties of the decomposition of Corollary~\ref{cor:XF} if $t$ is scarce
or the promise of Theorem~\ref{thm:main2} if $t$ is plentiful.
(Note that here this holds also if $t$ is of zero load.)

Hence, it remains to bound the size of $\beta_t'(x) \cap \bigcup \zZ_t^1$.
\begin{claim}\label{cl:zZt1bound}
    \[ |\beta_t'(x) \cap \bigcup \zZ_t^1| = \Oh_{d,h}\left(\left(\sqrt{k} \log k + \frac{\Delta}{\sqrt{k}}\right) \log^5k\right). \]
\end{claim}
\begin{claimproof}
If $t$ is of zero load, then $\zZ_t^1 = \emptyset$ and there
is nothing to prove. 
Otherwise, consider $K \in \zZ_t^1$ and let $v \in K \cap \beta_t'(x)$ for some $x \in V(T_t)$. 
We have either $v \in K \cap \beta_t(x)$ or there is a child $s$ of $t$
with $v \in N_G(D_{\downarrow s})$ and $v_{D_{\downarrow s}} \in \beta_t(x)$. 
Then, $v_{D_{\downarrow s}} \in K^3 \cap \beta_t(x)$ and every 
vertex $u \in \beta(t)$ can play the role of $v_{D_{\downarrow s}}$ as above for at
most $p_h \cdot 2^{\delta_h} = \Oh_h(1)$ vertices $v$ due to~\eqref{eq:vDbound}.

Hence, it suffices to bound the size of $\beta'_t(x) \cap \bigcup \zZ_t^3$.
If $t$ is of low load, we guess that it is scarce, and $(T_t,\beta_t)$ is obtained from
Corollary~\ref{cor:XF}. 
This implies that every $K^3 \in \zZ_t^3$ can contribute only $\Oh_{d,h}(1)$ vertices
to $K^3 \cap \beta_t(x)$. 
As $|\zZ_t^3| = |\zZ_t^1| \leq \sqrt{k}$, we have $|\beta_t'(x) \cap \bigcup \zZ_t^1| = \Oh_{d,h}(\sqrt{k})$. 
Finally, if $t$ is of high load, then 
the size of $\beta_t(x) \cap \bigcup \zZ_t^3 = \beta_t(x) \cap \bigcup \zZ_t^4$ is bounded
by $\Oh_{d,h}\left(\left(\sqrt{k} \log k + \frac{\Delta}{\sqrt{k}}\right) \log^5k\right)$ due to the promise of Theorem~\ref{thm:main2}.
This finishes the proof of the claim.
\end{claimproof}
This finishes the proof of Lemma~\ref{lem:main-d-diam}.
\end{proof}

The next lemma is an analog of Lemma~\ref{lem:main-apex} and immediately concludes the proof
of Theorem~\ref{thm:main-d}.
\begin{lemma}\label{lem:main-d}
  For every $d \geq 1$ and $h \geq 3$, there exists a constant $d_H' > 0$ and randomized polynomial-time algorithm that, 
  given an $K_{3,h}$-minor-free graph $G$ and an integer $k \geq 2$,
  outputs an induced subgraph $G'$ of $G$ and 
  a tree decomposition $(T,\beta)$ of $G'$ of maximum bag size at most
   $d_H' k \log k $
  such that for family $\zZ$ of $d$-clusters with $|\mathcal{Z}| \leq k$,
  with probability at least 
  \[ \left(k^{\frac{d_H' |\zZ| \log k}{\sqrt{k}} } |V(G)|^{\frac{d_H' |\zZ|}{k}}\right)^{-1} \]
  we have $\bigcup \zZ \subseteq V(G')$ and every bag of $(T,\beta)$ contains at most $d_H' \sqrt{k} \log^6k$ vertices
  of $\zZ$. 
\end{lemma}

\begin{proof}
For every connected component $C$ of $G$, create a new vertex $v_C$ adjacent
to one arbitrarily chosen vertex in $C$ (this does not break being $K_{3,h}$-minor-free) and perform
breadth-first search from $v_C$ in $C$. 
For $j \geq 0$, let $L_j$ be the set of vertices
within distance exactly $j$ from their corresponding vertex $v_C$; $(L_j)_{j \geq 0}$ 
is a partition of $V(G)$. We define $L_j = \emptyset$ for $j < 0$. 
Denote $L_{\leq i} \coloneqq \bigcup_{j \leq i} L_j$. 

Let $k' \coloneqq (d+1)k$. 
Uniformly at random guess a remainder $a \in \{0,1,\ldots,k'\}$ such that $\bigcup_{i \geq 0} L_{a+(k'+1)i}$
does not contain a vertex of $\bigcup \zZ$.
This is successful with probability at least
\begin{equation}\label{eq:est-baker-d}
 1-\frac{|\zZ|}{k'+1} \geq 2^{-2 \frac{|\zZ|}{k'} \log k'} \geq 2^{-\Oh_d\left(\frac{|\zZ|}{\sqrt{k}} \log k\right)}. 
\end{equation}
 (In the first inequality we used $k \geq 2$ and $|\zZ| \leq k$.)

For a connected component $C$ of $G$ and $i \geq 0$
we define the graph $G_{C,i}$ as $G[C \cap L_{\leq a-1 + (k'+1)i}]$ with, for $i > 0$, the graph
$G[C \cap L_{\leq a + (k'+1)(i-1)}]$ contracted onto $v_C$;
the vertex resulting from this contraction will be denoted by $v_{C,i}$, and we set $v_{C,0} = v_C$ in $G_{C,0}$. 
Note that 
\[ V(G_{C,i}) = \{v_{C,i}\} \cup \left( \bigcup_{j=a + 1 + (k'+1)(i-1)}^{a-1 + (k'+1)i} C \cap L_j \right). \]
Furthermore, $V(G_{C,i})$ has radius at most $k'$, and thus diameter at most $2k' = 2(d+1)k$. 
If the guess of the remainder $a$ is correct,
every cluster of $\zZ$ is contained in the union of vertex sets $V(G_{C,i} \setminus v_{C,i})$.%
\footnote{The only reason we added vertices $v_C$ instead of just taking arbitrary vertex in $C$ as $v_C$
is to know that $v_C \notin \bigcup \zZ$, which makes the presentation cleaner.}
Let $\zZ_{C,i} = \{K \in \zZ~|~K \subseteq V(G_{C,i})$. 
Let $\mathcal{D}$ be the family of pairs $(C,i)$ such that $V(G_{C,i}) \setminus \{v_{C,i}\} \neq \emptyset$.

For every $(C,i) \in \mathcal{D}$, we apply 
\cref{lem:main-d-diam} to $G_{C,i}$,
obtaining an induced subgraph $G_{C,i}'$ and its tree decomposition $(T_{C,i}',\beta_{C,i}')$. 
We remove $v_{C,i}$ from each $G_{C,i}'$ and its corresponding tree decomposition, 
and connect the tree decompositions
arbitrarily into a single tree decomposition $(T',\beta')$
of  the disjoint union of graphs $(G_{C,i}' \setminus \{v_{C,i}\})_{(C,i)\in \mathcal{D}}$, which is an induced subgraph of $G$.

For every $(C,i) \in \mathcal{D}$ the vertex $v_{C,i}$ is not in $\bigcup \zZ$,
and by \cref{lem:main-d-diam} we have $\bigcup \zZ_{C,i} \subseteq V(G_{C,i}')$
and every bag of $(T_{C,i}',\beta_{C,i}')$ contains at most 
$\Oh_{d,h}(\sqrt{k} \log^6 k)$
 vertices of $\bigcup \zZ_{C,i}$ with probability at least 
\[ \left(\left(k \cdot 2k'\right)^{\Oh_{d,h}\left(\frac{|\zZ_{C,i}| \log k}{\sqrt{k}}\right) } |V(G)|^{\Oh_{d,h}\left(\frac{|\zZ_{C,i}|}{k}\right)}\right)^{-1}. \]
By multiplying the above for every $(C,i) \in \mathcal{D}$, we obtain that $\bigcup \zZ \subseteq V(G')$
and every bag of $(T',\beta')$ contains $\Oh_{d,h}(\sqrt{k} \log^6 k)$ vertices of $\bigcup \zZ$ with probability
at least 
\[ \left(k^{\Oh_{d,h}\left(\frac{|\zZ|\log k}{\sqrt{k}}\right)} |V(G)|^{\Oh_{d,h}\left(\frac{|\zZ|}{k}\right)}\right)^{-1}. \]
Multiplying the above by the estimate~\eqref{eq:est-baker-d} for the correct choice of $a$,
we conclude the proof.
\end{proof}

\bibliographystyle{plain}
\bibliography{references}

\appendix
\section{Proof of Theorem~\ref{thm:dual0}}\label{app:dual0}
\begin{proof}[Proof of Theorem~\ref{thm:dual0}.]
  We phrase the quest of finding the sequence 
  $(P_1,\ldots,P_q)$ as a minimum-cost maximum flow problem;
  the cuts $C_j$ will be read from the optimum solution of the dual
  linear program. 

  To this end, construct a flow network $G'$ with vertex weights
  as follows. For every $v \in V(G) \setminus \{s,t\}$, introduce
  two copies of $v$: $v_0$ of cost $0$ and capacity $1$
  and $v_1$ of cost $1$ and capacity $+\infty$. 
  The vertex $s$ becomes a source of capacity $q$ and cost $0$
  and the vertex $t$ becomes a sink of capacity $q$ and cost $0$.
  The edges of $G'$ are defined naturally: every edge of $G$
  becomes at most four edges of $G'$, between every pairs 
  of copies of corresponding vertices. 
  
  We ask for a minimum-cost flow of value $q$ in $G'$. 
  As every capacity is either integral or $+\infty$,
  we can find a minimum-cost flow that decomposes into
  $q$ flow paths $P_1',\ldots,P_q'$, each carrying a unit flow. 
  Each flow path $P_i'$ projects onto a walk from $s$ to $t$ in $G$;
  we shorten this walk to a path and denote this path as $P_i$. 

  To define the chain $(C_1,\ldots,C_p)$, we phrase
  the flow problem as a linear program and look at its dual. 
  This is the primal program:
  \begin{align*}
  \mathrm{min} \qquad & \sum_{v \in V(G) \setminus \{s,t\}} \sum_{a \in N_{G'}(v_1)} f(v_1,a) &\\
  \mathrm{s.t.} \qquad & \sum_{b \in N_{G'}(a)} f(a,b) - f(b,a) = 0 & \forall a \in V(G') \setminus \{s,t\}\\
  & \sum_{a \in N_{G'}(s)} f(s,a) - f(a,s) = q &\\
  & \sum_{a \in N_{G'}(t)} f(t,a) - f(a,t) = -q &\\
  & \sum_{a \in N_{G'}(v_0)} f(v_0,a) \leq 1 & \forall v \in V(G) \setminus \{s,t\} \\
  & f(a,b) \geq 0 & \forall ab \in E(G').
  \end{align*}
  And this is the dual:
  \begin{align*}
  \mathrm{max} \qquad & q(y_t-y_s) - \sum_{v \in V(G) \setminus \{s,t\}} z_v & \\
  \mathrm{s.t.} \qquad & y_{v_0} \leq y_a + z_v & \forall v \in V(G) \setminus \{s,t\}, a \in N_{G'}(v_0) \\
  & y_{v_1} \leq y_a + 1 & \forall v \in V(G) \setminus \{s,t\}, a \in N_{G'}(v_1) \\
  & y_s \leq y_a & \forall a \in N_{G'}(s) \\
  & y_t \leq y_a & \forall a \in N_{G'}(t) \\
  & z_v \geq 0 & \forall v \in V(G) \setminus \{s,t\}\\
  \end{align*}
  The dual program can be interpreted as a distance problem, 
  where we maximize the distance from $s$ to $t$ with weight $q$,
  each vertex $v \in V(G) \setminus \{s,t\}$ incurs a travel cost
  of $\mathrm{min}(1,z_v)$, and we are penalized for every $z_v$.
  Since the dual program is invariant under adding an offset
  to all variables $y_a$, we can add a constraint $y_s=0$.

  Since the primal program is a flow problem, which is encoded
  as a totally unimodular matrix, we can find an optimum solution to the dual
  program that is integral. Observe that this optimum solution
  satisfies $z_v \in \{0,1\}$ for every $v \in V(G) \setminus \{s,t\}$,
  as it is not useful to set any value higher than $1$ for a variable~$z_v$. 

  We set $p = y_t$ and for $1 \leq j \leq p$ we take
  \[ C_j = \{v \in V(G) \setminus \{s,t\}~|~y_{v_0} = j \textrm{ and }z_v = 1\}. \]

  We first argue that $(C_1,\ldots,C_p)$ is a chain of $(s,t)$-separators.
  We denote $s_0 := s$ and $t_0 := t$. 
  Let $P$ be a path in $G$ from some $u$ to some $w$; and suppose $j_u := y_{u_0}$, 
  $j_w := y_{w_0}$
  are such that $0 \leq j_u < j_w \leq p$. 
  Let $P'$ be the corresponding path in $G'$ from $u_0$ to $w_0$
  that traverses $v_0$ whenever $P$ traverses $v$. 
  Let $j_u < j \leq j_w$.
  Since for every $ab \in E(G')$ we have $|y_a - y_b| \leq 1$,
  there exists a vertex $v_0$ on $P'$ with $y_{v_0}=j$; 
  let $v_0$ be the first such vertex. 
  Since $j > j_u \geq 0$, $v_0 \neq u_0$ and the predecessor $b$
  of $v_0$ on $P'$ satisfies $y_b=j-1$. This implies 
  $v \in V(G) \setminus \{s,t\}$ and $z_v=1$.
  Hence, $v \in C_j$, that is, $C_j$ intersects~$P$. 
  This proves that $C_j$ is an $(s,t)$-separator for every $1 \leq j \leq p$
  and that
  every $1 \leq j < j' \leq p$, $C_j$ interesects every path
  from $s$ to $C_{j'}$ and $C_{j'}$ interesects every path from $C_j$ to $t$.

  By complementary slackness condition, for every $1 \leq i \leq q$,
  every edge of $P_i'$ is tight with regards to its corresponding
  inequality in the dual. This implies that $P_i'$ is a shortest
  path from $s$ to $t$ with regards to weights $z_v$ on $v_0$
  and $1$ on $v_1$ for $v \in V(G) \setminus \{s,t\}$. 
  Consequently, $P_i'$ never visits both $v_0$ and $v_1$
  (i.e., $P_i$ is a projection of $P_i'$ without any shortcuts)
  and
  $|V(P_i) \cap C_j| \leq 1$ for every $1 \leq j \leq p$.
  Since $C_j$ is an $(s,t)$-separator, we actually
  have $|V(P_i) \cap C_j| = 1$. 
  This proves the first property of a $(p,q)$-structure.

  Consider now $v \in V(P_i) \cap V(P_{i'})$ for some
  $i \neq i'$. 
  Since $v_0$ has capacity $1$, either $P_i'$ or $P_{i'}'$
  passes through $v_1$; without loss of generality,
  assume $P_i'$ passes through $v_1$. 
  By complementary slackness condition, 
  the corresponding inequality
  $y_{v_1} \leq y_a + 1$ is tight for the predecessor $a$
  of $v_1$ on $P_i'$. Since $P_i'$ is a shortest path from 
  $s$ to $t$, passing through $v_0$ instead of $v_1$ does
  not shorten the distance, hence $z_v = 1$. 
  Furthermore, $0 < y_{v_0} = y_{v_1} \leq p$. 
  This implies $v \in C_{y_{v_0}}$, proving the second property 
  of a $(p,q)$-structure.

  Again by complementary slackness condition, 
  if $z_v=1$ for some $v \in V(G) \setminus \{s,t\}$, 
  then the corresponding inequality expressing the capacity
  of $v_0$ is tight. Hence, one of the flow paths $P_i'$
  passes through $v_0$. Since $P_i$ is the projection of $P_i'$
  without any shortcuts, $P_i$ passes through $v$. 
  This proves the third property of a $(p,q)$-structure
  and finishes the proof of Theorem~\ref{thm:dual0}.
\end{proof}

\end{document}